\documentclass[11pt]{article}
\RequirePackage{algorithm,algorithmic,graphicx,amssymb,epsf,epic,url}

\newtheorem{question}{Question}
\renewcommand{\deg}{\mathsf{deg}}
\def\denseformat{
\setlength{\textheight}{9.5in}
\setlength{\textwidth}{6.9in}
\setlength{\evensidemargin}{-0.3in}
\setlength{\oddsidemargin}{-0.3in}
\setlength{\headsep}{10pt}
\setlength{\topmargin}{-0.44in}
\setlength{\columnsep}{0.375in}
\setlength{\itemsep}{0pt}
}
\newtheorem{theorem}{Theorem}[section]

\newtheorem{claim}[theorem]{Claim}
\newtheorem{lemma}[theorem]{Lemma}

\newtheorem{corollary}[theorem]{Corollary}
\newtheorem{fact}[theorem]{Fact}
\newtheorem{comment}[theorem]{Comment}

\newtheorem{observation}[theorem]{Observation}

\def\boldhead#1:{\par\vskip 7pt\noindent{\bf #1:}\hskip 10pt}
\def\ithead#1:{\par\vskip 7pt\noindent{\it #1:}\hskip 10pt}
\def\ceil#1{\lceil #1\rceil}

\def\inline#1:{\par\vskip 7pt\noindent{\bf #1:}\hskip 10pt}
\def\midinline#1:{\par\noindent{\bf #1:}\hskip 10pt}
\def\dnsinline#1:{\par\vskip -7pt\noindent{\bf #1:}\hskip 10pt}
\def\ddnsinline#1:{\newline{\bf #1:}\hskip 10pt}
\def\largeinline#1:{\par\vskip 7pt\noindent{\large\bf #1:}\hskip 10pt}
\long\def\comment #1\commentend{}
\long\def\commhide #1\commhideend{}
\long\def\commfull #1\commend{#1}
\long\def\commabs #1\commenda{}
\long\def\commtim #1\commendt{#1}
\long\def\commb #1\commbend{}
\long\def\commedit #1\commeditend{} % Editing comments, marked also by $>>>$ 

\long\def\commB #1\commBend{}       % Omit in 1996 (both TR and Siena)
\long\def\commex #1\commexend{}     % LN home exercise (hide solutions)

\long\def\commsiena #1\commsienaend{}  % omit in Siena, show in TR
                                         
\long\def\commBI #1\commBIend{}  % omit in Bar-Ilan
                                         
\long\def\CProof #1\CQED{}

\def\blackslug{\hbox{\hskip 1pt \vrule width 4pt height 8pt
    depth 1.5pt \hskip 1pt}}
\def\QED{\quad\blackslug\lower 8.5pt\null\par}
\def\inQED{\quad\quad\blackslug}

\long\def\PPP#1{\noindent{\bf Proof:}{ #1}{\quad\blackslug\lower 8.5pt\null}}

\long\def\denspar #1\densend
{#1}
\newif\ifnotesw\noteswtrue% T to show box & marginal notes; F supresses.
   {\ifnotesw\marginpar[\hfill\(\top\)]{\(\top\)}\fi}%
   {\ifnotesw\marginpar[\hfill\(\bot\)]{\(\bot\)}\fi}

\newcommand{\mnote}[1]%
    {\ifnotesw\marginpar%
        [{\scriptsize\it\begin{minipage}[t]{\marginparwidth}
        \raggedleft#1%
                        \end{minipage}}]%
        {\scriptsize\it\begin{minipage}[t]{\marginparwidth}
        \raggedright#1%
                        \end{minipage}}%
    \fi}

\def\cH{{\cal H}}

\def\MathF{\hbox{\rm I\kern-2pt F}}
\def\MathP{\hbox{\rm I\kern-2pt P}}
\def\MathR{\hbox{\rm I\kern-2pt R}}
\def\MathZ{\hbox{\sf Z\kern-4pt Z}}
\def\MathN{\hbox{\rm I\kern-2pt I\kern-3.1pt N}}
\def\MathC{\hbox{\rm \kern0.7pt\raise0.8pt\hbox{\footnotesize I}
\kern-4.2pt C}}
\def\MathQ{\hbox{\rm I\kern-6pt Q}}

\newsavebox{\ttop}\newsavebox{\bbot}

\def\eps{\epsilon}

\def\nin{{~\not \in~}}

\denseformat
\newcommand {\ignore} [1] {}

\begin{document}

%Net-Trees to
\title{From Hierarchical Partitions to Hierarchical Covers: \\Optimal Fault-Tolerant Spanners for Doubling Metrics}
\author{
Shay Solomon \thanks{Department of Computer Science and Applied Mathematics, The Weizmann Institute of Science, Rehovot 76100, Israel.
E-mail: {\tt shay.solomon@weizmann.ac.il}.
This work is supported by the Koshland Center for basic Research.}}

\date{\empty}

\begin{titlepage}
\def\thepage{}
\maketitle

\begin{abstract}
A \emph{$(1+\eps)$-spanner} for a doubling metric $(X,\delta)$ is a subgraph $H$ of the complete graph corresponding to $(X,\delta)$,
which preserves all pairwise distances to within a factor of $1+\eps$. A natural requirement from a spanner is to be robust against
node failures, so that even when some of the nodes in the network fail, the remaining part would still provide a $(1+\eps)$-spanner.
The spanner $H$ is called a \emph{$k$-fault-tolerant $(1+\eps)$-spanner}, for any $0 \le k \le n-2$, if for any subset
$F \subseteq X$ with $|F| \le k$, the graph $H \setminus F$ (obtained by removing the vertices of $F$, as well as their incident edges, from $H$) is a $(1+\eps)$-spanner for $X \setminus F$.

In this paper we devise an optimal construction of fault-tolerant spanners for doubling metrics.
Specifically, for any $n$-point doubling metric, any $\eps > 0$, and any integer   $0 \le k \le n-2$,
our construction provides a $k$-fault-tolerant $(1+\eps)$-spanner with optimal degree $O(k)$ within optimal time $O(n  \log n + k n)$.

We then strengthen this result to provide near-optimal (up to a factor of $\log k$) guarantees on the diameter and weight
of our spanners,
namely, diameter $O(\log n)$ and weight $O(k^2 + k \log n) \cdot \omega(MST)$,
while preserving the optimal guarantees on the degree $O(k)$ and the running time $O(n  \log n + k n)$.

Our result settles several fundamental open questions in this area, culminating a long line of research that started with the STOC'95 paper of Arya et al.\ and the STOC'98 paper of
Levcopoulos et al.

On the way to this result we develop a new technique for constructing spanners in doubling metrics.
In particular, our spanner construction is based on a novel \emph{hierarchical cover} of the metric,
whereas most previous constructions of spanners for doubling and Euclidean metrics (such as the net-tree spanner)
are based on \emph{hierarchical partitions} of the metric. 
We demonstrate the power of hierarchical covers in the context of  geometric spanners
by  improving    the state-of-the-art results in this area.   
\end{abstract} 
\end{titlepage}

\pagenumbering {arabic} 

%Maybe provide subsection -- number of edges, degree and weight. No I don't think so...]]
\section{Introduction}

%\subsection{Background and Results}
%\vspace{-0.05in}
%\subsection
{\bf 1.1 ~Euclidean Spanners.~}
Consider a set $P$ of $n$ points in $\mathbb R^d$ and a number $\eps > 0$,
and let $H = (P,E)$ be a graph in which the weight $\omega(x,y)$ of each edge $(x,y)\in E$ 
is equal to the Euclidean distance $\|x-y\|$ between $x$ and $y$.   
The graph $H$ is called a \emph{$(1+\eps)$-spanner} (or simply \emph{spanner} if the \emph{stretch} $1+\eps$ is clear from the context) for $P$ if for every $p,q \in P$,
there is a path in $H$ between $p$ and $q$ whose weight
(the sum of all edge weights in it) 
is at most $(1+\eps) \cdot \|p-q\|$.  
Such a path 
is called a \emph{$(1+\eps)$-spanner path}. The problem of constructing
Euclidean spanners 
has been studied intensively \cite{Chew86,Clark87,Vai91,Sal91,KG92,ADDJS93,ADMSS95,AWY05,DES09tech};
see also the treatise on Euclidean spanners and their applications by Narasimhan and Smid \cite{NS07}.  %  titled ``Geometric Spanner Networks''.

Euclidean spanners find applications in geometric approximation algorithms, network topology design, distributed systems, and other areas.  
In many applications it is required to construct a
$(1+\eps)$-spanner $H = (P,E)$ that satisfies some useful
properties. First, the spanner should contain $O(n)$ (or nearly
$O(n)$) edges. 
Second, the {\em (maximum) degree} $\deg(H)$ of $H$ should be small.
Third, its {\em weight}\footnote{For convenience, we
will henceforth refer to the normalized notion of weight $\Psi(H) = {{\omega(H)}
\over {\omega(MST(P))}}$, which we call {\em lightness}.} $\omega(H) = \sum_{e
\in E}\omega(e)$ should not be much greater than the weight $\omega(MST(P))$ of
the minimum spanning tree $MST(P)$ of $P$. 

A natural requirement from a spanner is to be robust against
node failures, so that even when some of the nodes in the network fail, the remaining part would still provide a $(1+\eps)$-spanner.
The spanner $H$ is called a \emph{$k$-fault-tolerant $(1+\eps)$-spanner}, for any $0 \le k \le n-2$, if for any subset
$F \subseteq X$ with $|F| \le k$, the graph $H \setminus F$ (obtained by removing the vertices of $F$, as well as their incident edges, from $H$) is a $(1+\eps)$-spanner for $X \setminus F$; we will henceforth write FT as a shortcut for fault-tolerant. Note that the basic (non-FT) setting when all nodes are functioning
corresponds to the case $k = 0$. 

The notion of FT spanners was introduced in the pioneering work of Levcopoulos et al.\ \cite{LNS98} from STOC'98.
A basic (non-FT) construction of $(1+\eps$)-spanners with constant degree (and thus $O(n)$ edges) and constant lightness can be built in time $O(n \log n)$ \cite{AS94,DN94,GLN02}.
Levcopoulos et al.\ \cite{LNS98} observed that the $(k+1)$-power\footnote{An $s$-power of a graph $G$ is a graph with the same 
vertex set as $G$, and there is an edge between any pair of vertices that are connected by a path in $G$ with at most $s$ edges.} 
of a spanner is a $k$-FT spanner with the same stretch. Therefore, by taking the $(k+1)$-power of the basic construction of \cite{GLN02},
one can obtain in $O(n \log n) + n 2^{O(k)}$ time, a $k$-FT $(1+\eps)$-spanner with degree and lightness both bounded by $2^{O(k)}$.
Notice that any $k$-FT spanner (with an arbitrarily large stretch) must have degree $\Omega(k)$, and thus $\Omega(k n)$ edges; also, it is easy to see that
the lightness must be $\Omega(k^2)$ \cite{CZ03}. Thus there is a gap between the exponential upper bound $2^{O(k)}$
and the polynomial lower bounds $\Omega(k)$ and $\Omega(k^2)$ on the degree and lightness, respectively.
%this is an optimal combination of parameters for $k = O(1)$.
In addition, Levcopoulos et al.\ \cite{LNS98} devised two other constructions of $k$-FT spanners, one with $O(k^2 n)$ edges and running
time $O(n \log n + k^2 n)$, and another with edges and running time both bounded by $O(k n \log n)$.
In WADS'99 Lukovszki \cite{Luk299} devised 
two constructions of $k$-FT spanners, one with the optimal number of edges $O(k n)$ and with running time $O(n \log^{d-1} n + k n  \log \log n)$,
and another with degree $O(k^2)$.
%of $k$-FT spanners 

In SoCG'03, Czumaj and Zhao \cite{CZ03} showed that the optimal guarantees $O(k)$ and $O(k^2)$ on the degree and lightness
of $k$-FT spanners can be achieved using a greedy algorithm.\footnote{The ${n \choose 2}$ edges of the underlying complete Eucldiean graph 
are traversed by increasing order of weight; each edge $(x,y)$ is added to the current graph if and only if it does not contain 
$k+1$ vertex-disjoint $t$-spanner paths between $x$ and $y$.} However, it is unclear whether the greedy algorithm of \cite{CZ03} can be implemented
efficiently; a naive implementation requires time $O(n^2) \cdot Paths(n,k+1,t)$, where $Paths(n,k+1,t)$ is the time needed to check whether an $n$-vertex
graph (with $O(k n)$ edges) contains $k+1$ vertex-disjoint $t$-spanner paths between an arbitrary pair of vertices.
In addition, Czumaj and Zhao \cite{CZ03} devised a construction of $k$-FT spanners with the optimal degree $O(k)$ and with lightness $O(k^2 \log n)$,
within time $O(k n \log^{d} n + n k^2 \log k)$. (See Table \ref{tab1} for a reference. We also refer to Chapter 18 in the 
treatise \cite{NS07} for an excellent survey on Euclidean FT spanners.)

We remark that the time needed to compute a $k$-FT spanner is $\Omega(n \log n + k n)$. The first term $\Omega(n \log n)$
is a lower bound on the time needed to compute a basic (non-FT) spanner in the \emph{algebraic computation tree model} \cite{CDS01},
and the second term $\Omega(k n)$ is a lower bound on the number of edges in any $k$-FT spanner.
Up to this date no construction of $k$-FT spanners could combine the optimal running time
%\footnote{The time $O(n \log n)$ is optimal
%in the algebraic computation tree model 
$O(n \log n + kn)$ with either the optimal number of edges $O(k n)$ (and thus  the optimal degree $O(k)$) or with the optimal lightness $O(k^2)$. 
In particular, the following questions are stated in the treatise of Narasimhan and Smid \cite{NS07}; all these questions remained open even for 2-dimensional point sets.
%{\vspace{0.05in}\\{\bf Conjecture 1 (\cite{ADMSS95})} \emph{~For any $t > 1$, and any
%dimension $d$, there is a $t$-spanner, constructible in $O(n \cdot \log n)$ time,
%with bounded degree, $O(\log n)$ diameter, and weight $O(\omega(MST) \cdot \log n)$.}
\vspace{-0.06in}
\begin{question} [Open Problem 26 in \cite{NS07}]
Is there an algorithm that constructs a $k$-FT $(1+\eps)$-spanner with $O(k n)$ edges in $O(n \log n + kn)$ time?
\end{question}
\vspace{-0.08in}

\vspace{-0.12in}
\begin{question} [Open Problem 27 in \cite{NS07}]
Is there an algorithm that constructs a $k$-FT $(1+\eps)$-spanner with degree $O(k)$ in $O(n \log n + kn)$ time?
(This question subsumes Question 1.)
\end{question}
\vspace{-0.08in}

\vspace{-0.14in}
\begin{question} [Open Problem 28 in \cite{NS07}]
Is there an algorithm that constructs a $k$-FT $(1+\eps)$-spanner with lightness $O(k^2)$ in $O(n \log n + kn)$ time?
\end{question}
\vspace{-0.08in}

%Another important measure of quality for a spanner $G$ is a small \emph{(hop-)diameter} $\Lambda(G)$, that is, for every $p,q \in P$ there should be a %$(1+\eps)$-spanner path
%that consists of at most $\Lambda(G)$ edges.
Another important measure of quality is the \emph{(hop-)diameter} of spanners;
%where the objective is to achieve as small as possible diameter.
we say that a spanner $H$ has diameter  $\Lambda(H)$ if it provides a $(1+\eps)$-spanner path with at most $\Lambda(H)$ edges, for every $p,q \in P$.
%which should be small. That is,
Achieving a small diameter is desirable for some practical applications (e.g., in network routing protocols); see \cite{AMS94,ADMSS95,AM04,AWY05,CG06,DES09tech},
and the references therein. Euclidean Spanners that combine small diameter with some of the other parameters (among small number of edges, degree, lightness
and running time) have been well studied in the last twenty years.  % \cite{AMS94,ADMSS95,AWY05,DES09tech,Sol11}.
In particular, in a seminal STOC'95 paper by Arya et al.\ \cite{ADMSS95}, a single construction of spanners that combines all these parameters
was presented, 
having constant degree, diameter $O(\log n)$, lightness $O(\log^2 n)$ and running time $O(n \log n)$. Moreover, Arya et al.\  conjectured
that the lightness can be improved to $O(\log n)$, without increasing any of the other parameters; having lightness and diameter both bounded by $O(\log n)$  is optimal due to a lower bound of \cite{DES09tech}.
%It was later shown that any spanner with diameter $O(\log n)$ must have lightness $\Omega(\log n)$ \cite{AWY05,DES09tech}.
%which implies that the spanner construction of Arya et al.\ \cite{ADMSS95} is optimal, except for a slac
%This combination of parameters is optimal up to constant factors, except for the lightness bound where there is a logarithmic slack, and the lower bound $\Omega(\log n)$. 
Arya et al.'s long-standing conjecture was resolved in the affirmative in STOC'13 by Elkin and this author \cite{ES13}.

Chan et al.\ \cite{CLN12} generalized the result of \cite{ES13} for the FT setting.
Specifically, they showed that there exist $k$-FT spanners with degree $O(k^2)$,
diameter $O(\log n)$ and lightness $O(k^3 \log n)$.
The running time of the construction of \cite{CLN12} was not analyzed; a naive implementation requires quadratic time. 
In the ``Future Direction''  Section of \cite{CLN12}, Chan et al.\ asked whether the dependence on $k$
in the degree and lightness bounds of their construction can be improved.
A partial answer to this question was given by this author \cite{Sol12}, %which improved the results of Chan et al.\ \cite{CLN12}, achieving 
who achieved degree $O(k^2)$, diameter $O(\log n)$ and lightness $O(k^2 \log n)$, within time $O(n \log n + k^2 n)$. 
A merging of the two manuscripts \cite{CLN12} and \cite{Sol12} gave rise to the ICALP'13 paper \cite{CLNS13}.
(See Table \ref{tab1} for a reference.)
%Chan et al.\ \cite{CLN12} asked the following question, which was answered  partially in \cite{Sol12}.
%(Refer to the  for the accurate statement
%of their question.)
We remark that the upper bounds of \cite{CLN12,Sol12,CLNS13} are pretty far from the known lower bounds:
degree $O(k^2)$ versus $\Omega(k)$, lightness $O(k^2 \log n)$ versus $\Omega(k^2)$, and running time $O(n \log n + k^2 n)$ versus $\Omega(n \log n + k n)$.
In particular, the following question was left open.
\vspace{-0.08in}
\begin{question} \cite{CLN12,Sol12,CLNS13} 
%Can the dependence on $k$ in the constructions of \cite{CLN12,Sol12} be improved?
Is there an algorithm that constructs a $k$-FT spanner with degree $o(k^2)$, diameter $O(\log n)$
and lightness $O(k^2) + o(k^2 \log n)$,  within time $O(n \log n) + o(k^2 n)$?
\end{question}

%Observe that for our $k$-fault tolerant spanner construction, the dependence on $k$ for
%maximum degree is $O(k^2)$ and for lightness is $O(k^3)$, where the lower bounds are 
%$\Omega(k)$ (trivial) and $\Omega(k^2)$ (a simple example is given in [CZ03]) respectively; it is interesting to see whether the dependence on $k$ 
%can be improved. Indeed, for low-dimensional Euclidean metrics, Czumaj and Zhou [CZ03] constructed $k$-fault
%tolerant spanners with maximum degree $O(k)$ and $O(k^2)$ lightness, but no guarantee on hop-diameter.

%was given
%is desirable that the spanner achieves a small diameter.
%That is, for every $p,q \in P$ there should be a path $\Pi$ in $G$
%that contains at most $\Lambda(G)$ edges and has weight $\omega(\Pi) = \sum_{e
%\in E(\Pi)} \omega(e) \le (1+\eps) \cdot \|p-q\|$. 

\noindent{\bf Our Results.~}
We show that for any set $P$ of $n$ points in $\mathbb R^d$, any $\eps > 0$, and any integer $0 \le k \le n-2$,
one can build a $k$-FT $(1+\eps)$-spanner with optimal degree $O(k)$ within  time $O(n \log n + kn)$.
The running time of our construction holds in the algebraic computation tree model,
and is therefore optimal as well. 
\\In particular, this result settles in the affirmative Questions 1 and 2 above.

%In addition,  
Moreover, our spanners also achieve diameter $O(\log n)$ and lightness $O(k^2 + k \log n)$, which settles Question 4.
(See Table \ref{tab1} for a summary of previous and our constructions of Euclidean FT spanners.)

Consider the lightness bound. Notice that it is equal to $O(k^2)$, for all $k = \Omega(\log n)$,
thus matching the na\"{\i}ve lower bound $\Omega(k^2)$ in this range.
In other words, it provides a positive answer to Question 3 for all $k = \Omega(\log n)$.
More generally, it exceeds the na\"{\i}ve  bound $\Omega(k^2)$ by a factor of $\frac{\log n}{k}$.

However, the na\"{\i}ve  bound $\Omega(k^2)$ on the lightness can be strengthened if we take into account the diameter $O(\log n)$
of our spanners. %disregards the fact that the diameter is bounded by $O(\log n)$,
Indeed, any spanner with diameter $O(\log n)$ must have lightness $\Omega(\log n)$ \cite{AWY05,DES09tech}. 
More generally, for any parameter $\rho \ge 2$, any spanner with diameter $O(\log_\rho n)$ must have lightness $\Omega(\rho \log_\rho n)$ \cite{DES09tech}. Note also that any spanner with degree $O(\rho)$ must have diameter $\Omega(\log_\rho n)$.
By combining all these lower bounds together (and substituting $\rho$ with $k$), we get that one cannot do better than the following tradeoff: 
$k$-FT spanners with degree $\Omega(k)$, diameter $\Omega(\log_k n)$ and lightness $\Omega(k^2 + k \log_k n)$. Notice that the parameters of our spanners are pretty close to this lower bound tradeoff,
with only a slack of $\log k$ on the diameter and a slack of $\min\{\log k, \frac{\log n}{k}\} = O(\log \log n)$ on the lightness.
%the degree is the same, our diameter $O(\log n)$ is greater than the optimal one by a factor of $\log k$, and our lightness is greater than
%the optimal one by a factor of $\log k$ as well. 

%Consider first the lightness bound $O(k^2 + k \log n)$. 
% Also, the degree $O(k)$ of our spanners is optimal too. 

\begin{table*}
\begin{center}
\resizebox {\textwidth }{!}{
\footnotesize
\begin{tabular}{|c|c|c|c|c|c|c|}
\hline  Reference & \# Edges & Degree &  Diameter & Lightness & Running Time & Metric  \\
\hline \cite{LNS98}  & $n 2^{O(k)}$ & $2^{O(k)}$ & unspecified  & $2^{O(k)}$ & $O(n  \log n) + n 2^{O(k)}$ & Euclidean \\
\hline \cite{LNS98}   & $O(k^2  n)$ & unspecified & unspecified & unspecified & $O(n  \log n + k^2 n)$ & Euclidean \\
\hline \cite{LNS98}   & $O(k n \log n)$ & unspecified & unspecified & unspecified & $O(k n \log n)$ & Euclidean \\
\hline \cite{Luk299}   & $O(k n)$ & unspecified & unspecified & unspecified & $O(n \log^{d-1} n + k n  \log \log n)$ & Euclidean \\
%\hline \cite{Luk299}   & $O(n k)$ & unspecified & unspecified & unspecified & $O(n \log^{d-1} n + n k \log \log n)$ & Euclidean \\
\hline \cite{CZ03}   & $O(k n)$ & $O(k)$ & unspecified & $O(k^2)$ & $O(n^2) \cdot Paths(n,k+1,t)$ & Euclidean \\
\hline \cite{CZ03}   & $O(k n)$ & $O(k)$ & unspecified & $O(k^2 \log n)$ & $O(k n \log^{d} n + n k^2 \log k)$ & Euclidean \\
%\hline \cite{NS07}   & $n \cdot \eps^{-O(dim) \cdot k}$ & $\eps^{-O(dim) \cdot k}$ & unspecified & $k \cdot \eps^{-O(dim) \cdot k}$ & $n (\eps^{-O(dim)} \cdot % \log n + \eps^{-O(dim) \cdot k})$ & Euclidean \\
\hline \cite{CLN12a}   & $O(k n)$ & unspecified & unspecified & unspecified & unspecified & doubling \\
\hline \cite{CLN12a}   & $O(k m)$ & unspecified & $O(\alpha(m,n))$ & unspecified & unspecified & doubling \\
\hline \cite{CLN12a}   & $O(k^2 n)$ & $O(k^2)$ & unspecified & unspecified & unspecified & doubling \\
%\hline \cite{CLN12}   & $O(k^2 n)$ & $O(k^2)$ & $O(\log n)$ & $O(k^3 \log n)$ & unspecified & doubling \\
\hline \cite{CLN12,Sol12,CLNS13}   & $O(k^2 n)$ & $O(k^2)$ & $O(\log n)$ & $O(k^2 \log n)$ & $O(n \log n + k^2 n)$ & doubling \\
\hline \hline {\bf New} &  {\boldmath $O(k n)$} & {\boldmath $O(k)$} &{\boldmath $O(\log n)$} &  {\boldmath $O(k^2 + k \log n)$} &  {\boldmath $O(n \log n + k n)$} & {\bf doubling} \\
%\hline {\bf New} &   {\boldmath $O(\rho)$} & we & {\boldmath $O(\log_\rho n + \alpha(\rho))$} &  {\boldmath $O(\rho \cdot \log_\rho n)$} &  {\boldmath $O(n %\cdot \log n)$} & {\bf doubling} \\
\hline
\end{tabular}
}
\end{center}
\vspace{-0.15in}
\caption[]{\label{tab1} \footnotesize A  comparison between the previous state-of-the-art and our  
constructions of FT spanners for low-dimensional Euclidean and doubling metrics. See Theorem 1.1 for the dependencies of our construction
on $\eps$ and $d$.}
%maybe write CLN12 and Sol12 in the same entry. Also, maybe add rows from the NS07 book, and use the NS07 book to verify
%the correctness of existing rows in the table; instead of LARGE specify $disjoint(k+1)$ as the time needed to compute
%$k+1$-edge disjoint paths, and specify the running time using this term]]}
\end{table*}
\vspace{-0.15in}

%somewhere need to mention that there is a lower bound $\Omega(k^2 + \log n)$ on the lightness by combining DES08 (regardless of k-FT)
%and the trivial lower bound on k-FT lightness]]

%\subsection
\vspace{0.07in}
\noindent{\bf 1.2 ~Doubling Metrics.~}
The \emph{doubling dimension} of a metric $(X,\delta)$ is the smallest value $d$
such that every ball $B$ in the metric can be covered by at most
$2^{d}$ balls of half the radius of $B$.  % \cite{Ass83,Clark99,GKL03}.
This notion generalizes the Euclidean dimension, since the doubling dimension
of the Euclidean space $\mathbb R^d$ (equipped with any of the $\ell_p$ norms) is $\Theta(d)$.
A metric is called \emph{doubling} if its doubling dimension is constant.
Doubling metrics were implicit in the works of Assoud \cite{Ass83} and Clarkson \cite{Clark99}, and were explicitly defined
  by Gupta et al.\ \cite{GKL03}. Since then they have been studied intensively (see, e.g., \cite{KL04,Tal04,HPM06,CG06,Smid09,ABN11,BGK12,CLNS13}).
  % and the references therein.

Spanners for doubling metrics were also subject of intensive research  \cite{GGN04,CGMZ05,CG06,HPM06,Rod07,GR081,GR082,Smid09,CLN12a,ES13,CLN12,Sol12,CLNS13}.
%and the references therein. 
In many of these works the objective is to devise spanners that achieve one parameter or more among 
small number of edges, degree, diameter, lightness, and running time.
Spanners for doubling metrics were also found useful for approximation algorithms \cite{BGK12}
and for machine learning \cite{GKK12}. %Spanners for doubling metrics that achieve some of the parameters of interest to us (among

The literature on FT spanners for doubling metrics is quite sparse, and consists of the following papers \cite{CLN12a,CLN12,Sol12,CLNS13}.
In ICALP'12, Chan et al.\ \cite{CLN12a} devised three constructions of $k$-FT spanners for doubling metrics.
Their first construction achieves the optimal number of edges $O(k n)$, the second  achieves $O(k m)$ edges
and diameter $O(\alpha(m,n))$, where $\alpha$ is the two-parameter inverse-Ackermann function, and the third construction
achieves degree $O(k^2)$. The other papers are the follow-ups \cite{CLN12,Sol12,CLNS13} to \cite{ES13} that were discussed in Section 1.1, 
which provide within $O(n \log n + k^2 n)$ time, a $k$-FT spanner with degree $O(k^2)$, diameter $O(\log n)$ and lightness $O(k^2 \log n)$. The constructions 
of \cite{CLN12,Sol12,CLNS13} apply to doubling metrics.

Obviously, Questions 1-4 from Section 1.1 are relevant for  doubling metrics as well. Moreover, even much weaker variants
of these questions can be asked.  % for this wider family of doubling metrics. 
In particular, the previous state-of-the-art bounds on the degree and lightness of $k$-FT spanners for doubling metrics
are $O(k^2)$ and $O(k^2 \log n)$, respectively; hence it is natural to ask if these bounds can be improved, regardless of the running time issue.
% MENTION THAT lightness $o(\log n)$ for doubling metrics is unknown, regardless of other parameters;
% this relates to Questions 3 and 4; also, mention that this is a central question, and was raised in a few places,
% with applications to metric tsp; cite BGK12 and also Chan et al., and also Smid09]]

\ignore{
In SODA'05 Chan et al.\ \cite{CGMZ05} showed that
for any doubling metric there exists a spanner with constant
degree. In SODA'06, Chan and Gupta \cite{CG06} devised  a construction of
spanners with diameter $O(\alpha(n))$.
Smid \cite{Smid09} showed that in doubling metrics a greedy
construction produces spanners with logarithmic lightness. 
Gottlieb et al.\ \cite{GKK12} devised a construction of spanners
with constant degree and logarithmic diameter, within $O(n \cdot \log n)$ time.
To the best of our knowledge, prior to our work, there were no known constructions of
spanners for doubling metrics that
provide logarithmic diameter and lightness simultaneously (even allowing arbitrarily large degree).\footnote{On the other hand,
as was mentioned in Section , for Euclidean metrics such a construction was devised by
Arya et al.\ \cite{ADMSS95}. However, the degree in the latter construction is unbounded.}
}

We demonstrate that our construction  extends to doubling metrics without incurring any overhead (beyond constants)
in the degree, diameter, lightness, and running time. 
%In other words, Theorem \ref{tm1} applies to doubling metrics.
Notice that our construction improves all the previous state-of-the-art constructions of FT spanners for both Euclidean and doubling metrics,
disregarding the construction of \cite{CZ03} which requires a somewhat unreasonable running time of $O(n^2) \cdot Paths(n,k+1,t)$.
(See Table \ref{tab1} for a  summary of previous and our results.)

In particular, our result settles Questions 1, 2 and 4 in the affirmative 
%, and provides a positive answer to Question 3 for all $k = \Omega(\log n)$,
for doubling metrics. 

Question 3 is settled for all $k = \Omega(\log n)$, but remains open in the complementary range of $k = o(\log n)$).
We remark that the problem of constructing spanners for doubling metrics with sub-logarithmic lightness is a central open question in this area
(see \cite{BGK12}),
%We remark that up to this date no construction of spanners for doubling metrics could achieve lightness $o(\log n)$, 
even for basic (non-FT) spanners, and regardless of the diameter or any other parameter of the construction. 
%In fact, improving this logarithmic 
%lightness bound is a central open question in this area \cite{BGK12}.
In other words, a positive solution to Question 3 in the range $o(\log n)$ seems currently out of reach for doubling metrics, even for the basic case $k = 0$.

%constructions of FT spanners for both Euclidean and doubling metrics.)

The main result of this paper is summarized  in the following theorem. 
\vspace{-0.07in}
\begin{theorem} \label{tm1}
Let $(X,\delta)$ be an $n$-point doubling metric, with an arbitrary doubling dimension $d$.
For any $\eps > 0$ and any integer $0 \le k \le n-2$,
a $k$-FT $(1+\eps)$-spanner with degree $\eps^{-O(d)} \cdot k$, diameter $O(\log n)$ and lightness $\eps^{-O(d)}  (k^2 + k \log n)$
can be built within $\eps^{-O(d)}  (n \log n + kn)$ time.
\end{theorem}

%Copy paste some motivation from the ``Robust Geometric Spanners'' paper, or maybe from other papers on this topic ]]

%\clearpage

\vspace{0.07in}
\noindent
{\bf 1.3 ~Our and Previous Techniques.~}
Most previous works on geometric FT spanners \cite{LNS98,Luk299,CZ03,NS07} apply to Euclidean metrics,
and rely heavily on profound geometric properties of low-dimensional Euclidean metrics, such as the \emph{gap property} \cite{CDNS92,AS94} and the \emph{leapfrog property} \cite{DHN93}.
In particular, these constructions do not extend to arbitrary doubling metrics.
In contrast, in our work we only use standard packing arguments, and do not rely
on any other geometric property of low-dimensional Euclidean metrics.
Consequently, our construction applies to arbitrary doubling metrics. 
%In particular, 

On the other hand, there are some similarities between our techniques and the techniques used in the previous works on FT spanners for doubling metrics \cite{CLN12a,CLN12,Sol12,CLNS13}.
The starting point in both our work and the works of \cite{CLN12a,CLN12,Sol12,CLNS13} is 
the standard \emph{net-tree spanner} of \cite{GGN04,CGMZ05},
which is induced from a \emph{net-tree} $T = T(x)$ that corresponds to a \emph{hierarchical partition} of the metric $(X,\delta)$. 
Any tree node $x$ is associated with a single point $p(x)$ that belongs to the point set $D(x)$ of its descendant leaves. 
For any pair of tree nodes at the same level that are close together (with respect to the distance scale at that level),
a \emph{cross edge} is added. The net-tree spanner is the union of the tree edges (i.e., the edges of $T$) and the cross edges.
(See Section 2 for more details.)
%For any descendant $y$ of $x$ in $T$, we have $D(y) \subseteq D(x)$.
%The general idea is to assign each internal node of $T$ a \emph{label} from its descendant leaves.
To obtain a $k$-FT spanner, \cite{CLN12a,CLN12} use $k+1$ net-trees rather than one, and in each of these trees, each tree node $x$
is associated with a single point from $D(x)$. Another idea that is used in \cite{Sol12,CLNS13} is to use a single net-tree, but to associate each tree node $x$ with
$\Theta(k)$ points from $D(x)$ (if any) rather than a single point. To achieve degree $O(k^2)$, \cite{CLN12a,CLN12} used 
single-sink spanners similar to those in \cite{ADMSS95},
whereas \cite{Sol12,CLNS13} used a rerouting technique from the bounded degree spanner construction of \cite{GR082}.
However, achieving degree $o(k^2)$ using these ideas is doomed to failure. First, in \cite{CLN12a,CLN12} there are $k+1$ net-trees,
each of which contributes $\Omega(k)$ units to the degree load, thereby giving a total degree of $\Omega(k^2)$.
Incurring a degree of $\Omega(k^2)$ from the approach of \cite{Sol12,CLNS13} is also inevitable. 
The main problem is that a node $x$ is associated with $\Theta(k)$ points from $D(x)$ (if any). 
Consequently, the same leaf point $p$ may belong to $\Theta(k)$ different sets $D(x)$ of internal nodes $x$.
For each cross edge that is incident on any of these $\Omega(k)$ nodes, 
we must connect $p$ with $\Omega(k)$ edges, and so the degree of $p$ becomes $\Omega(k^2)$.
%Then, each edge of the basic spanner construction is translated to a bipartite
%because sometimes a node $x$ will have
%only a few  

The basic idea of associating nodes of net-trees and other hierarchical tree structures (such as split trees and dumbbell trees)
with points from their descendant leaves has been widely used in the geometric spanner literature; see \cite{CK92,ADMSS95,GGN04,CGMZ05,CG06,NS07,GR082,CLN12a},
and the references therein. 
%The idea of \cite{Sol12} behind associating a tree node $x$ with $\Theta(k)$ arbitrary points from $D(x)$ 
%is that 
In particular, any point in $D(x)$ is   close to the original net-point $p(x)$ with respect to the distance scale of $x$, denoted $rad(x)$.

Instead of restricting ourselves to $D(x)$, we suggest to look at a larger set $B(x)$ of all points
%namely, those 
that belong to the ball of radius $O(rad(x))$ centered at $p(x)$.   %, where $rad(x)$ is the distance scale of $x$.
By associating nodes $x$ with points from $B(x)$,  % instead of $D(x)$,
%looking at sets $B(\cdot)$ instead of $D(\cdot)$, 
we get a \emph{hierarchical cover} of the metric. Since the doubling dimension is constant,
this cover will have a constant \emph{degree} (every point will belong to a constant number of sets $B(\cdot)$ at each level).
We use this constant degree hierarchical cover to obtain degree $O(k)$ as follows.  %in the following way.
A node $x$ will appoint to itself $\Theta(k)$ points from $B(x)$ whose degree due to cross edges at lower levels is small enough; let $S(x)$ be the set of these points, called \emph{surrogates}.
A cross edge $(x,y)$ increases the degree of surrogates of $S(x)$. However, every unit of increase to a surrogate $p \in S(x)$ is due to some surrogate $q \in S(y)$ that is close to $p$--the distance between $p$ and $q$ will be $O(\frac{1}{\eps} \cdot rad(x))$. 
This surrogate $q$, as well as other points that are within distance $O(rad(x))$ from $q$, will not necessarily be in $B(x)$.
However, since the distance scales increase exponentially with the level of a tree node, all these points must belong to a \emph{close ancestor} $x'$ of $x$,
where the level of $x'$ is higher than that of $x$ by at most $O(\log \frac{1}{\eps})$ units. 
The key idea here is to ``compensate'' one of the close ancestors $x'$ of $x$ at a value that matches the ``cost'' of $x$ due to cross edges.
By re-appointing the same surrogates of $S(x)$ for $O(\log \frac{1}{\eps})$ levels continuously (at which stage we must reach a close ancestor $x'$ of $x$ that is 
compensated), and only then appointing new surrogates for $S(x')$, we will be able to achieve the optimal degree $O(k)$. 
While in the approach of \cite{Sol12,CLNS13} each point of $X$ belongs to lists $D(x)$ of $\Omega(k)$ nodes $x$,
our approach can guarantee (though this does not follow immediately from the above discussion) that each point of $X$ will serve as a surrogate of only $O(\log \frac{1}{\eps})$ nodes. Having assigned surrogates for each node, we replace each edge $(x,y)$ of the basic spanner by a bipartite clique between 
the corresponding surrogate sets $S(x)$
and $S(y)$. Observe that $S(x)$ and $S(y)$ may contain $\Theta(k)$ points each, and so the size of this bipartite clique may be as large as $\Theta(k^2)$.
In other words, we may replace a single edge of the basic spanner by $\Theta(k^2)$ edges. Since the basic spanner has $\Theta(n)$ edges, one might get
the \emph{wrong} impression 
that our FT spanner contains $\Theta(k^2 n)$ edges, which is far too much (by a factor of $k$)!
%Indeed, if the number of edges in our FT spanner were $\Theta(k^2 n)$, then its degree would have to be $\Theta(k^2)$, which stands
%in contradiction to our claimed results.
There are two reasons why the number of edges in our FT spanner is in check.
First, it turns out that many edges $(x,y)$ of the basic spanner satisfy $|S(x)| + |S(y)| \ll k$, which implies that there are many sparse bipartite cliques.
More specifically, we show that the overall number of edges in all these sparse bipartite cliques is $O(k n)$.
Second, even though there are still many edges in the basic spanner which satisfy $|S(x)| + |S(y)| = \Theta(k)$,
we demonstrate that most of them are in fact \emph{redundant}, and such redundant edges need not be replaced by bipartite cliques. 
(See Section 3.1 for more details.)

%Since the doubling dimension is constant,
%this cover will have a small degree.

%stress that the logarithmic term cannot be shaved because of the diameter. But regardless of other parameters
%and fault-tolerance, sub-log lightness is an open question]]

%Put an emphasis on the hierarchical cover thing, that it might be useful for other proximity data structures,
%not necessarily in the context of fault-tolerant spanners. 
%Comparing myself to the previous work of myself and (Chan et al.) is that I don't let a point be the surrogate
%of more than a constant number of nodes (this is true only after pruning redundant nodes), whereas Solomon and Chan
%may have the same point serve as surrogate of $\Theta(k)$ nodes.

%The most involved requirement from our construction is to achieve degree $O(k)$.
The main challenge of this work was to achieve the optimal bound $O(k)$ on the degree. 
However, small degree (which also implies small number of edges) is just one out of four parameters that are of interest to us.
%Our construction is tailored in a delicate way that guarantees degree $O(k)$.
%This part of the construction is quite involved and elaborate.
Our construction is tailored in such a way
%the interesting thing about our construction is 
that once degree $O(k)$ is achieved, the rest of the parameters come almost for free.
\\\emph{Lightness.} The standard net-tree analysis shows that the basic spanner construction has lightness $O(\log n)$.
%This implies that the lightness of the basic spanner construction is logarithmic.
Replacing each edge of the basic spanner with a bipartite clique of size $O(k^2)$ between the corresponding surrogate sets
gives rise to lightness $O(k^2 \log n)$.
Improving the lightness bound to $O(k^2 + k \log n)$ is more involved.   %, and is based on the following idea.
The idea is to try replacing each edge $(x,y)$ of the basic spanner with a bipartite matching instead of a bipartite clique;
this is possible only when both $S(x)$ and $S(y)$ contain $\Theta(k)$ surrogates, which is not always the case (see Section 4). 
The total lightness of the bipartite matchings is clearly $O(k \log n)$. The more interesting part is to show that 
%Alas, getting rid of all 
the total lightness of the bipartite cliques is $O(k^2)$.
%impossible,
%and bounding their total weight by  is the main obstacle.
%and analyzing their weight contribution is the main obstacle. 
%We bound the weight of the bipartite cliques separately (and this is the main obstacle), and show that their total weight is bounded by $O(k^2) w(MST)$.
\\\emph{Diameter.} We can achieve diameter $O(\log n)$ by shortcutting the net-tree using the 1-spanners of \cite{SE10}. To control the lightness,
we follow an idea from \cite{CLN12,Sol12,CLNS13} and shortcut the \emph{light subtrees} (those at small enough distance scales). 
If this is done carelessly, the degree and lightness will increase to at least $O(k^2)$ and $O(k^2 \log n)$, respectively.
However, by pruning a certain subset of nodes from the net-tree, and using bipartite matchings for the shortcut edges
whenever possible, the desired bounds can be achieved.
\\\emph{Running time.} Given the standard net-tree spanner construction and the 1-spanners of \cite{SE10}, our construction
can be implemented within $O(k n)$   time in a rather straightforward manner.
Since the standard net-tree spanner can be built within time $O(n \log n)$ \cite{Rod07,GR082},
and since the same amount of time suffices to build the 1-spanners of \cite{SE10}, the overall running time of our construction is $O(n \log n +kn)$.
We remark that for low-dimensional Euclidean metrics, some variants of the net-tree (such as the fair split tree of \cite{CK92}) can
be built within time $O(n \log n)$ in the algebraic computation tree model. Consequently, the basic spanner construction 
can also be built within time $O(n \log n)$ in this model.
%spanner (obtain by using, say, the fair split tree instead of the net-tree) can be built
The   time bound $O(n \log n)$ for the 1-spanners of \cite{SE10} also applies to this model.
Therefore, for low-dimensional Euclidean metrics, the time bound of our construction applies to this model as well.
For arbitrary doubling metrics, one should also need a rounding operation (to allow one to find the $i$ for which
$2^i < x \le 2^{i+1}$, for any $x$) \cite{Gottlieb13}.

%Follows from the time needed to construct a net-tree spanner. Given this spanner, the rest of the construction
%takes only $O(k n)$ time, except for the time required for the shortcutting (see diameter below).
%[[S:Explain that net-trees can be computed within this time in the alge' model in Euclidean metrics,
%but in doubling metrics we make no claims...]]
%should mention that the difficult part is to handle incomplete nodes. There are much more incomplete nodes than complete nodes;
%roughly speaking, the number of complete nodes is $O(n/k)$, and thus we can afford to spend $O(k^2)$ edges for every such node (edges
%incident to this node). 

\vspace{0.07in}
\noindent
{\bf 1.4~ Similar Results by Kapoor and Li.~}
Independently of our work, Kapoor and Li \cite{KL13} obtained similar results for Euclidean FT spanners.
In particular, they showed that for any low-dimensional Euclidean metric, one can build within time $O(n \log n + k n)$,
a $k$-FT $(1+\eps)$-spanner with degree $O(k)$ and lightness $O(k^2)$.
Similarly to our result, their running time and degree bounds are optimal. In addition, they obtain the optimal lightness bound $O(k^2)$,
while our lightness bound $O(k^2 + k \log n)$ has a slack of $\frac{\log n}{k}$.  %for $k = o(\log n)$.
However, while the diameter of their spanners is unbounded, our spanners have diameter $O(\log n)$.
Moreover, their result relies heavily on geometric properties of low-dimensional
Euclidean metrics and thus cannot be extended to doubling metrics, while our result applies to all doubling metrics. 
In particular, the techniques and ideas used in the two papers are inherently different.

%done independently, different techniques (their techniques rely heavily on geometric properties),
%inherently different result -- their diameter is unbounded.
It should be noted that the results of Kapoor and Li \cite{KL13} were obtained much prior to our results.\footnote{Private communication with Sanjiv Kapoor, March 2013.}
However, we stress that our results  were achieved before this author became aware of the results of \cite{KL13}.
% author of this paper was unaware of their results, and 
%Stress that the results are from an unpublished manuscript.

\vspace{0.07in}
\noindent
{\bf 1.5~ Organization.~} 
In Section 2 we define the basic notions and present the notation that is used throughout the paper.
Section 3 is devoted to our construction of FT spanners with optimal degree. More specifically, the description of the construction is given in Section 3.1,
whereas its analysis is given in Section 3.2 (degree analysis) and Section 3.3 (fault-tolerance and stretch analysis).
%and its analysis is
In Sections 4 and 5 we show that small lightness and diameter, respectively, 
can be obtained using a few simple modifications to the construction of Section 3. The running time issue is addressed in Section 6.
%See \cite{Sol13} for the full version of this paper.]]

\vspace{-0.04in}
%\clearpage
\section{Preliminaries}
\label{sec:prelim}
\vspace{-0.04in}
%For any positive integer $m$, we denote $[m] = \{1, 2, \ldots, m\}$.
Let $(X,\delta)$ be an arbitrary $n$-point doubling metric.
% and let $1 \leq k \leq n - 2$ be an integer representing the maximum number of failed nodes allowed.
%We consider the regime of stretch $1+\eps$,
%for an arbitrarily small $0 < \eps < \frac{1}{2}$.
Without loss of generality, we  assume that the minimum inter-point distance of $X$ is equal to $1$.
We denote by $\Delta = \max_{u, v \in X}\delta(u, v)$  the \emph{diameter} of $X$.

%We say that a cluster $C \subseteq X$ has radius at most $r$, if there exists $x \in C$ such that $C \subseteq B(x, r)$.
%Let $r_2 > r_1 > 0$.
%The ring of inner radius $r_1$ and outer radius $r_2$ centered at $x$ is $R(x, r_1, r_2) = B(x, r_2) \setminus B(x, r_1)$.

%The \emph{doubling dimension} of a metric space $(X, d)$, denoted by $\dim(X)$ (or $\dim$ when the context is clear), is the smallest value $\rho$ such that every ball in $X$ can be covered by $2^\rho$ balls of half the radius. A (family of) metric is called \emph{doubling} if its doubling dimension is bounded above by a constant \cite{DBLP:conf/focs/GuptaKL03}.

%For a finite subset $Y \subseteq X$, the \emph{aspect ratio} of $Y$ is $\delta(Y) = \frac{\max_{y \ne y' \in Y}\delta(y, y')}{\min_{y \ne y' \in Y}\delta(y, y')}$. Without loss of generality, we assume that the minimum distance between points in $X$ is $1$. Hence, we have $\delta(X) = \diam(X)$.

%The ball of radius $r > 0$ centered at $x$ is $B(x, r) = \{u \in X : \delta(x, u) \leq r\}$.
A set $Y \subseteq X$ is called an \emph{$r$-cover} of $X$ if for any point $x \in X$ there is a point $y \in Y$, with $\delta(x, y) \le r$.
A  set $Y$ is an \emph{$r$-packing} if for any pair of distinct points $y, y' \in Y$, it holds that $\delta(y, y') > r$.
For $r_1 \geq r_2 > 0$, we say that a set $Y \subseteq X$ is an \emph{$(r_1,r_2)$-net} for $X$ if $Y$ is both an $r_1$-cover of $X$ and an $r_2$-packing;
such a net can be constructed by a greedy algorithm. We will use the following standard packing argument.
%By recursively applying the definition of doubling dimension, we can get the following key fact \cite{GKL03}. remove sentence]]

\vspace{-0.04in}
\begin{fact} \label{prop:small_net}
Let $R \geq 2r > 0$ and let $Y \subseteq X$ be an $r$-packing  in a ball of radius $R$. Then, $|Y| \le (\frac{R}{r})^{2\dim}$.
\end{fact}

\noindent
{\bf Hierarchical Nets.}  We consider
the hierarchical nets that are used by Gottlieb and Roditty \cite{GR082}.
%Let %$r_i = 5^i$ and 
Write $\ell = \ceil{\log_5 \Delta}$,
and let $\{N_i\}_{i \geq 0}^{\ell}$ be a sequence
of hierarchical nets, where $N_0 = X$ and for each $i \in [\ell]$, $N_i$ is
a $(3 \cdot 5^{i}, 5^{i})$-net for $N_{i-1}$.\footnote{As mentioned in \cite{GR082},
this choice of parameters\ignore{, explained by Cole and Gottlieb \cite{CG06},}
is needed in order to achieve running time $O(n \log n)$.} 
For each $i \in [\ell]$, we refer to $N_i$ as the \emph{$i$-level net}
and the points of $N_i$ are called the \emph{$i$-level net points}.
Note that $N_0 = X \supseteq N_1 \supseteq \ldots \supseteq N_\ell$,
and $N_\ell$ contains exactly one point.
The same point of $X$ may have instances in many nets; specifically, an $i$-level net point is necessarily a $j$-level net point,
for every $j \in [0,i]$. When we wish to refer to a specific instance of a point $p \in X$, 
which is determined uniquely by some level $i \in [0,\ell]$ (such that $p \in N_i$),
%We may henceforth denote an $i$-level net point $p \in N_i$ 
we may denote it by a pair $(p,i)$.  % where $i$ is the level.
%to explicate the specific level to which a point $p$ belongs.

\vspace{4pt} 
\noindent
{\bf Net-tree.}
The hierarchical nets induce a hierarchical tree $T = T(X)$, called \emph{net-tree} \cite{GGN04,CGMZ05}.
\\Each node in the net-tree $T$ corresponds to a unique net-point; we will use $(p,i)$ to denote both the
net point and the corresponding tree node.
We refer to the nodes corresponding to the $i$-level net points as the \emph{$i$-level nodes}.
The only $\ell$-level node is the root of $T$, and for each $i \in [\ell]$, each $(i-1)$-level node has a parent
at level $i$ within distance $3 \cdot 5^{i}$; such a node exists since $N_{i}$ is a $3 \cdot 5^{i}$-cover for $N_{i-1}$.
(We may assume that for a pair $(p,i+1),(p,i)$ of net-points, $(p,i+1)$ will be the parent of $(p,i)$ in $T$.)
%We need to stress that if a point $p$ is present in several levels $i,i+1,...$, then it will be a parent of itself in all these levels!!!]]
Thus, any descendant of every $i$-level node can reach it by climbing a path of weight at most $\sum_{j \in [i]} 3 \cdot 5^{i} \le 4 \cdot 5^{i}$.
By Fact \ref{prop:small_net}, since the doubling dimension is constant, each node has only a constant number of children.
For a tree node $v = (p,i)$, let $Ch(v)$ denote the set of children of $v$ in the net-tree $T$.

\vspace{4pt} 
\noindent
{\bf Net-tree Spanner via Cross Edges.}
We recap the basic  spanner construction $H = H(X)$ of \cite{GGN04,CGMZ05}.
%\noindent \emph{Cross Edges.} 
\\In order to achieve stretch $(1 + \eps)$,
for each   $i \in [0,\ell-1]$,
\emph{cross edges} are added between $i$-level nodes that are
close together with respect to the distance scale $5^i$ at that level.
Specifically, for any pair $(p,i),(q,i)$ of distinct $i$-level nodes such that $\delta(p,q) \leq \gamma  5^i$,
for some parameter $\gamma = \Theta(\frac{1}{\eps})$,
we add a cross edge between $(p,i)$ and $(q,i)$ of weight $\delta(p,q)$.
The basic spanner construction $H$ is obtained as the union of the \emph{tree edges} (those connecting
nodes with their children in $T$) and the cross edges,
where an edge between a pair of nodes is translated to an edge between the corresponding points.
%The cross edges constitute the edge set of  $H$,
%where a cross edge $((p,i),(q,i))$ between two tree nodes is translated to an edge $(p,q)$ between the corresponding points.

By Fact \ref{prop:small_net}, the degree of any tree node in the spanner $H$ (due to tree and cross edges) is $\eps^{-O(d)}$.
Moreover, it can be shown that $H$ has   $n \cdot \eps^{-O(d)}$ edges \cite{GGN04,CGMZ05}.
On the other hand, as the same point may have instances in many nets (and thus in many tree nodes),
%Consequently, even though the degree of a tree node in the spanner is $\eps^{-O(d)}$, 
the degree of a point may be unbounded.

The following lemma from  \cite{GGN04,CGMZ05} gives the essence of the basic cross edges framework.

\begin{lemma}[Basic Cross Edges Framework Guarantees Low Stretch]
\label{lemma:cross_edge}
Consider the basic spanner construction $H$,
whose edge set consists of all the tree edges and the cross edges (defined with respect to some parameter $\gamma = O(\frac{1}{\eps}$)).
For each $p,q \in X$, the spanner ${H}$ contains a $(1+\eps)$-spanner path $\Pi_{p,q}$, obtained
by climbing up from the leaf nodes $(p,0)$ and $(q,0)$ to some $j$-level ancestors $(p',j)$ and $(q',j)$, respectively,
where $(p',j)$ and $(q',j)$ are connected by a cross edge and $\delta(p,q) = \Theta(\frac{1}{\eps}) \cdot 5^j$.
\end{lemma}

As mentioned above, 
the degree of the spanner $H$ is unbounded. To reduce the degree of the spanner, we will appoint to each internal node of the net-tree $T$ a \emph{surrogate},
which is a point in $X$ that is nearby (with respect to the distance scale at that level).
More specifically, a node $(p,i)$ of $T$, $i \in [\ell]$, will be appointed a surrogate $q \in X$ such that $\delta(p,q) = O(5^i)$;
the surrogate of a node $(p,i)$ will be denoted by $s(p,i)$. (For technical convenience, we   define the surrogate of a leaf node $(p,0)$ as $p$ itself, i.e., $s(p,0) = p$.)
Another important requirement from a surrogate of an $i$-level node is that its degree due to cross edges at lower levels $j \in [0,i-1]$
would be sufficiently small; we will discuss this issue in detail later on.

Given a surrogate for each  tree node, the basic spanner construction $H$ is translated
to its \emph{surrogate spanner} $s(H)$ by replacing each tree edge $e = ((p,i+1),(q,i))$ (respectively, cross edge $e = ((p,i),(q,i))$) with its \emph{surrogate edge} $s(e) = (s(p,i+1),s(q,i))$ (resp., $s(e) = (s(p,i),s(q,i))$).
We will demonstrate (Section 3) that surrogates can be appointed to nodes, so that the resulting surrogate spanner has constant degree.
Moreover, by increasing the number of surrogates with which each node is appointed from 1 to (at most) $k+1$, and thus replacing each edge of the basic spanner $H$
by a bipartite clique between the corresponding surrogates, we will obtain a $k$-FT spanner with optimal degree $O(k)$.
%0-level cross edges of $H$ are taken as is to $\mathcal H$.
%Since we will later reroute spanner paths and
%app internal node with labels (which we refer to as \emph{s}),
%we also give an extended version here.

\begin{lemma}[Extended Cross Edges Framework Guarantees Low Stretch]
\label{lemma:cross_edge2}
%Consider the spanner $\mathcal H$.
%\noindent \emph{(a)}  Let $\mu >0$ be an arbitrary constant. Suppose a graph $\mathcal{H}$ on the node
%whose edge set consists of all cross edges (defined with respect to some parameter $\gamma = O(\frac{1}{\eps}$)).
% and for each level $i \geq 1$, each level-$(i-1)$ tor is connected via a tree edge
% to some level-$i$ tor within  distance $\mu r_i$. 
For each $p,q \in X$, consider the $(1+\eps)$-spanner path $\Pi_{p,q}$
in $H$ that is guaranteed by Lemma \ref{lemma:cross_edge}. 
The corresponding \emph{surrogate path} $s(\Pi_{p,q})$ in $s(H)$, obtained by replacing each edge $e$ of $\Pi_{p,q}$ by its surrogate edge $s(e)$,
% contains a , obtained 
%by climbing up from the leaf nodes $(p,0)$ and $(q,0)$ to some $j$-level ancestors $(p',j)$ and $(q',j)$, respectively,
%where $(p',j)$ and $(q',j)$ are connected by a cross edge and $\delta(p,q) = \Theta(\frac{1}{\eps}) \cdot 5^j$.
%
%\noindent \emph{(b)} In the above graph $\mathcal{H}$, if for each level $i \geq 1$, each level-$i$ tor $(u,i)$ is
%labeled with a point $\widehat{u}$ such that $\delta(u,\widehat{u}) = O(r_i)$, then the path $\widehat{P_{u,v}}$ induced by the above
%path $P_{u,v}$ and the labels 
is a $(1 + O(\eps))$-spanner path (between the same endpoints $p$ and $q$). We can reduce the stretch of the path from $1+O(\eps)$ back to $1+\eps$ by scaling $\gamma$ up by an appropriate constant.
\end{lemma}

\section{Fault-Tolerant Spanners with Optimal Degree}
In this section we devise a construction of $k$-FT spanners with optimal degree $O(k)$.
The description of the construction is given in Section 3.1,
whereas its analysis is given in Section 3.2 (degree analysis) and Section 3.3 (fault-tolerance and stretch analysis).
%The construction is given in and its analysis are given in Sections 3.1 and 3.2, respectively.

\subsection{The Construction}
\vspace{0.07in}
\noindent
{\bf 3.1.1 ~The Basic Sceheme.~}
Recall that the weight of $i$-level cross edges in the basic spanner construction $H$ is at most $\gamma 5^i$, where $\gamma = O(\frac{1}{\eps})$,
and define $\tau = \lceil \log_5 \gamma \rceil + 1$.
Also, let $\xi = \xi(\eps,dim) = \eps^{-O(d)}$
be an upper bound for the maximum degree of any tree node in the basic spanner $H$.
%We take $\xi$ to be even r than that, so that it upper bounds the maximum 
%sum of degree of nodes that

%The general idea is to try appointing $k+1$ surrogates $s_1(p,i),\ldots,s_{k+1}(p,i)$ for each tree node $(p,i)$, 
%so that for any $k$ faults in the network, at least one surrogate of $(p,i)$ will be functioning.
%The set $S(p,i) = \{s_1(p,i),\ldots,s_{k+1}(p,i)\}$ will be called the \emph{surrogate set} of $(p,i)$.
%A surrogate $s \in S(p,i)$ of an $i$-level node $(p,i)$ is a point $q$ at small distance
%$\delta(p,q) = O(5^i)$ from $p$ (we will explicate the constant hidden by the $O$-notation in the sequel), and
%whose degree due to cross edges at lower levels $j \in [0,i-1]$ is small enough, namely, at most $\eps^{-O(d)} \cdot k$
%(we will explicate the exact degree threshold in the sequel).

The general approach is to try appointing $k+1$ surrogates $s_1(p,i),\ldots,s_{k+1}(p,i)$ for each tree node $(p,i)$, 
so that for any $k$ node failures in the network, at least one surrogate of $(p,i)$ will be functioning.
The set $S(p,i) = \{s_1(p,i),\ldots,s_{k+1}(p,i)\}$ will be called the \emph{surrogate set} of $(p,i)$.
A surrogate $s \in S(p,i)$ of an $i$-level node $(p,i)$ is a point $q$ at distance
$\delta(p,q) = O(5^i)$ from $p$,
%(the leading constant is either 2 or 7, depending on whether we are at host or leech]],  
whose degree due to cross edges at lower levels $j \in [0,i-1]$ is at most $D = (\tau+4) \xi^2 (2k+1)$;
we henceforth refer to $D$ as the \emph{degree threshold}.

The basic spanner  $H$ is given as the union of the tree and cross edges. 
%Specifically, 
%for each $i \in [0,\ell - 1]$, every pair $(p,i),(q,i)$ of $i$-level nodes such that $\delta(p,q) \le \gamma 5^i$ is connected by a cross edge.
The FT-spanner construction $\mathcal H$ is obtained from $H$ by replacing each edge of $H$ with a bipartite
clique. Specifically, every cross edge $((p,i),(q,i))$ is replaced with a bipartite
clique between the corresponding surrogate sets $S(p,i)$ and $S(q,i)$. Such a cross edge  $((p,i),(q,i))$ is referred to as an \emph{$i$-level cross edge}
(of the basic spanner $H$). 
Moreover, we may also refer to the edges of the corresponding bipartite clique as cross edges (of the FT spanner $\cH$).
Similarly, a tree edge $((p,i+1),(q,i))$ is replaced with a bipartite clique between $S(p,i+1)$ and $S(q,i)$;
such a tree edge $((p,i+1),(q,i))$ is referred to as an \emph{$i$-level tree edge} (of the basic spanner $H$), and we may also refer to the edges of the corresponding
bipartite clique as tree edges (of the FT spanner $\cH$).
Roughly speaking, each bipartite clique that replaces an $i$-level tree edge
%however, we will later show that bipartite cliques that replace tree edges
%are usually redundant (except for when new surrogates for $(p,i+1)$ are appointed), as they 
is incarnated in a clique that replace some $i$-level or $(i+1)$-level cross edge; refer to Section 3.2.3 for a rigorous argument.
%; see in particularthe second stage therein.)
% of the current and previous level.
For the intuitive discussion of this section we will henceforth restrict our attention to cross edges.
%(Of course, our argument
%(A rigorous argument that takes into account the issue of tree appears in Section .)

%Therefore, restricting the attention to cross edges.

%Thus the main issue is how to handle the cross edges.
%wait a minute -- what about the tree edges??? should mention that
%there is a non-trivial issue with cross edges, whenever surrogates of a node are not taken from surrogates of its children]] 
%The crux of the problem is to compute the surrogate sets. 

Next, we describe how to compute the surrogate sets, which is the crux of the problem.

For a node $(p,i)$, denote by $D(p,i)$ its \emph{descendant set}, i.e., the set of points in its descendant leaves. 
Note that for each point $q \in D(p,i)$,
we have $\delta(p,q) \le 4 \cdot 5^{i}$. In particular, any point $q \in D(p,i)$ is at a small enough distance from $p$ to serve as $(p,i)$'s surrogate.
If there are less than $k+1$ points in $D(p,i)$ whose degree is at most $D$, we need to look for surrogates
outside $D(p,i)$; in particular,
%i.e., outside the subtree of $T$ rooted at $(p,i)$; in particular, %such surrogates can be taken from descendant sets
we can look for nearby $i$-level nodes $(q,i)$, with $\delta(p,q) = O(5^i)$,
and try to use somehow their descendant  sets $D(q,i)$.
%(This choice of 5 for the leading constant is necessary
%to guarantee that surrogates of $S(p,i)$ are within distance $7 \cdot 5^i$ from it; we will address this issue in the sequel.)
We remark that in order to achieve $k$-fault-tolerance,  it is sufficient
that the  surrogate set $S(p,i)$ of a node $(p,i)$ will contain its entire descendant set $D(p,i)$, 
even if $|D(p,i)| \ll k+1$. Nevertheless, we will see later that even in such cases it is still advantageous for
a node to have $k+1$ surrogates, if possible.
%OK, but we are not actually doing that, because
%we don't want to keep checking whether all descendants of a node have been used (this may also take too much time)]]

The surrogate sets are computed bottom-up, starting at the leaf nodes.
For a leaf $(p,0)$, it is enough to take just $p$ to the surrogate set $S(p,0)$, i.e., $S(p,0) = D(p,0) = \{p\}$. 
More generally, we can take the surrogate set $S(p,i)$ to be $D(p,i)$, for  small levels $i = O(1)$ in the tree.
%As we climb up the tree, things become more complicated.
However, as we climb up the tree, we might not be able to guarantee that $S(p,i)$ contains the entire descendant set $D(p,i)$ due to degree violations;
in this case it is critical to guarantee that $|S(p,i)| = k+1$.
%To guarantee that we can always find enough surrogate points, we preemptive actions.
For any point $a \in D(p,i)$ that cannot be taken to the surrogate set $S(p,i)$, its degree must exceed the degree threshold $D$,
which, in turn, implies that many cross edges at lower levels are incident on nodes of which $a$ is a surrogate; more specifically, since 
each node may be incident on at most $\xi$ cross edges at each level (in the basic spanner $H$)
%$a$ may be incident
%to at most $\xi$ cross edges at each level 
and $D$ is sufficiently large, it follows that some cross edges
at levels smaller than $i-\tau$ must be incident on nodes of which $a$ is a surrogate.
Intuitively, any $j$-level cross edge between a $j$-level node $(\tilde a,j)$ for which $a \in S(\tilde a,j)$ and some other $j$-level node $(q,j)$, where $j \le i - \tau$, should give rise to additional
points that are close to $a$ and thus also to $p$, namely, those in $D(q,j)$. 
Indeed, first note that the distance between $\tilde a$ and its surrogate $a$ is $O(5^j)$.
Second, the weight of $j$-level cross edges is at most $\gamma 5^j$, and so $\delta(a,q) \le O(5^j) + \gamma 5^j \le O(5^j) + \gamma 5^{i-\tau} = O(5^{i})$.
Hence for any point $b \in D(q,j)$, $$\delta(p,b) ~\le~ \delta(p,a) + \delta(a,q) + \delta(q,b) ~\le~ 4 \cdot 5^{i} + O(5^{i}) + 4 \cdot 5^{j} ~=~ O(5^i).$$
%~\le~ \frac{9}{5} \cdot 5^{i} ~<~ 7 \cdot 5^{i}.$$
%In this case
%the distance scales at levels $j$ and $i$ differ by a factor of $\Omega(\gamma)$. 
%Also, 
This means that
we can take the points in $D(q,j)$ of small degree
%Assuming the degree of these points is small enough,
%they can be taken 
to the surrogate set $S(p,i)$ of $(p,i)$; however, if there are points in $D(q,j)$ whose degree exceed  $D$, we will try to apply this argument again. 
Consequently, a surrogate of a node $(p,i)$ need not belong to $D(p,i)$, but rather to a (possibly much larger) set $F(p,i)$ 
which we refer to as the \emph{friend set} of $p$.
%of points
%that are at distance $O(5^i)$ from $p$. 
This set $F(p,i)$ will contain $O(k)$ points of small degree that are at distance $O(5^i)$ from $p$ (hence $F(p,i)$ is some subset
of $O(k)$ points from the \emph{ball set} $B(p,i)$ mentioned in Section 1.3). 
We need to show that whenever the degrees of surrogates of some node $(p,i)$ are about to exceed $D$, 
there are at least $k+1$ points in $F(p,i)$ of small degree which can be appointed as new surrogates instead of the old ones,
and so the degree bound will always be in check.

For now assume that the degree of a node $(p,i)$ increases only due to cross edges (the degree contribution due to tree edges can be easily controlled,
as we shall later see).
Consider a cross edge $((p,i),(q,i))$, which gives rise to a bipartite clique between $S(p,i)$ and $S(q,i)$,
and thus increases the degree of each point in $S(p,i)$ by $|S(q,i)|$ units. We will show that $|F(q,i)| \ge |S(q,i)|$, 
and the key idea is to let $(q,i)$ share its friend set with $(p,i)$.
In this way we compensate $(p,i)$ at a value that matches its cost due to cross edges.
However, note that $p$ and $q$ may be at distance $\gamma 5^i$ apart (recall that $\gamma 5^i$ is the upper bound on the weight
of $i$-level cross edges),
and so the points in $F(q,i)$ may be too far from $(p,i)$ to serve as its friends. Thus the points of $F(q,i)$ are not
moved immediately to $F(p,i)$, but rather to a temporary set $R(p,i)$ called the \emph{reserve set} of $(p,i)$.
All the points in $R(p,i)$ will be moved to $F(p,i)$ within $\tau$ levels.
%This way the cost $(p,i)$ incurs as a result of the cross edge $((p,i),(q,i))$  and compensate
%; we will later show
%that there are at least $|S(q)|$ points of degree at most $D$ that are close enough to $q$, and which can be ``donated'' to $r$
%to compensate for the degree increase of points of $S(p,i)$. 
%--as a result, the leading constant $\frac{9}{5}$ will increase to $9$. 
%Next, we formalize this intuition.
%Each node $(p,i)$ will also have a \emph{reserve set} $R(p,i)$,
%which is a set of $O(k)$ points of degree at most $D$ that are pretty close to $p$;  it is possible that  $R(p,i) \cap S(p,i) \ne \emptyset$.
%I started to change this point -- will continue tomorrow!!!]]
%These points are potential surrogates that are reserved for subsequent levels $i+1,\ldots,\ell$. We will guarantee that the points of $R(p,i)$ are within d%istance $2\gamma \cdot 5^i$ from $p$. Thus, even though some of them may be too far from $p$ to serve as surrogates of $(p,i)$ at  level $i$, %However, 
%we will be able to use all these points as surrogates of $j$-level ancestors of $(p,i)$, for $j \ge i + \tau$.
During these $\tau$ levels in which some point in the reserve set is too far to become a friend, 
we are going to re-appoint the same surrogates over and over.
% i.e., the surrogates that are appointed to the children of the respective ancestor of $(p,i)$ 
%(which may be different than the surrogates of $(p,i)$ itself; details will be given later).  
%, which were also used at the grandchildren of that node, etc. 
That is,  new surrogates of a node are appointed  for a long \emph{term} of $\Omega(\tau)$ levels, 
which means that no ancestor of such a node in the next $\Omega(\tau)$ levels will appoint new surrogates.
During this term of $\Omega(\tau)$ levels, the points of $R(p,i)$ become closer and closer to the respective ancestor of $(p,i)$, until they are close enough
to become its friends and thus appointed as its surrogates--at this stage new surrogates are appointed to that ancestor, also for a term of $\Omega(\tau)$ levels.
%More specifically, when we appoint new surrogates to a node, 
%To overcome this obstacle we use the same surrogates continuously for a \emph{tenure period} of $\tau$ levels.

As mentioned in Section 1.3, one may get the \emph{wrong} impression that our FT spanner $\cH$ contains $O(k^2 n)$ edges rather than $O(k  n)$ edges.
Indeed, each edge $(x,y)$ of the basic spanner $H$ is translated into a bipartite clique in $\cH$ between the corresponding surrogate sets $S(x)$ and $S(y)$.
Since $S(x)$ and $S(y)$ may contain $\Theta(k)$ points each, the size of such a bipartite clique may be as large as $\Theta(k^2)$.
%In other words, we may replace a single edge of the basic spanner by $\Theta(k^2)$ edges. 
Recalling that the basic spanner has $\Theta(n)$ edges, it may seem
that our FT spanner $\cH$ contains $\Theta(k^2 n)$ edges, which is far too much!
% rather than $\Theta(k n)$ , which is far too much (by a factor of $k$)!
%Indeed, if the number of edges in our FT spanner were $\Theta(k^2 n)$, then its degree would have to be $\Theta(k^2)$, which stands
%in contradiction to our claimed results.
There are two reasons why the number of edges in our FT spanner is in check.
First, it turns out that many nodes $x$ have a small surrogate set $S(x)$; such a node $x$ is called \emph{small} (see Section 3.1.2 for more details).
Consequently, many edges $(x,y)$ of the basic spanner satisfy $|S(x)| + |S(y)| \ll k$, which implies that many of the bipartite cliques are sparse.
%More specifically, we show that the overall number of edges in all these sparse bipartite cliques is $O(k n)$.
Second, even though there are still many edges $(x,y)$ in the basic spanner which satisfy $|S(x)| + |S(y)| = \Theta(k)$,
we demonstrate that most of them are in fact \emph{redundant}, and such redundant edges need not be replaced by bipartite cliques. 
The issue of redundant edges is related to the notions of leeches and hosts, and is discussed in detail in Section 3.1.3.

\vspace{0.07in}
\noindent
{\bf 3.1.2 ~Complete Versus Incomplete, Small Versus Large, Clean Versus Dirty.~}
%A node $(p,i)$ is called \emph{large} if $|R(p,i)| \ge k+1$; otherwise it is \emph{small}.
%Denote by $U(p,i)$ the set of useful points from $R(p,i)$ and also from $S(p,i)$]], i.e., those that may be used immediately as surrogates upon need (those %with constant 2)
%perhaps this sentence needs to be given in the previous subsection]].
%The surrogates of a node $(p,i)$ may be taken from its descendant set $D(p,i)$ and from its useful reserve set $U(p,i)$,
%i.e., $S(p,i) \subset
A node $(p,i)$ is called \emph{complete} if $|S(p,i)| \ge k+1$; %We will show later that the reserve set of a small node is empty.
otherwise $|S(p,i)| < k+1$, and it is   \emph{incomplete}. A node is called \emph{large} if $|F(p,i) \cup S(p,i)| \ge 2k+2$;
otherwise $|F(p,i) \cup S(p,i)| < 2k+2$, and it is   \emph{small}. %A large node is called \emph{safe} if $|F(p,i)| \ge
%We also say that a node is \emph{huge} if $|F(p,i)| \ge 2k+2$.
%(respectively, \emph{medium} if 
%$k+1 \le |F(p,i)| \le 2k+2$;  \emph{small} if $|F(p,i)| < k+1$).
%The friend set $F(p,i)$ of an incomplete node must be empty (points of $F(p,i)$ can serve as surrogates, and so we would move them to $S(p,i)$).
%However, this is not the case for regular nodes. A regular node with $|U(p,i)| < 2k+2$ (respectively, $|U(p,i)| \ge 2k+2$) 
%is called \emph{medium} (resp., \emph{large}). 
%The rational for this terminology is that for any small node $(p,i)$, the ball $B(p,c \cdot 5^i)$ of radius $c \cdot 5^i$ around $p$ 
%(for an appropriate constant $c \approx 6$) is small, specifically, it contains all points of $S(p,i)$ and just them;
%we will refer to the ball $B(p,c \cdot 5^i)$ as the \emph{ball around $(p,i)$}.
%The rational for using size terminology is implied by the following properties, which will become clear in the sequel:
Since any point of $F(p,i)$ can serve as a surrogate, our construction guarantees that any large node must be complete.
The friend set $F(p,i)$ of $(p,i)$ will contain $O(k)$ (more specifically, at most $3k+3$) clean \emph{10-friends} of $(p,i)$; we say that a point $q$ or a node $(q,i)$ is a $t$-friend of some 
$i$-level node $(p,i)$ if $\delta(p,q) \le t \cdot 5^i$.

%maybe remove large and small? Why do we need them?? We can just use dirty and clean whenever needed, and say that
%before an appointment we have $2k+2$ points in $F(x)$]]
%and (3) all ancestors of a large node are large. 
%$k+1$ points, specifically $|S(p,i)|$ nodes. Or we may say that the ball of that radius contains

%We say that a point is \emph{clean} if its degree is at most $D$; otherwise it is \emph{dirty}.

%maybe have just small and large nodes; and treat them like I did before, except that the definition of clean/dirty changes
%according to degree; so it is possible that a node is large, but its surrogates are clean. Once the degree of a surrogate
%passes $D$ it becomes dirty and will not be appointed again; otherwise, it may be appointed over and over...
%And this saves the definition of useful reserve set... I'm not sure...]]

Small and large nodes are handled differently. 
For a small node $(p,i)$, all points in $D(p,i)$ are appointed as surrogates, i.e., we assign $S(p,i) = D(p,i)$.
We will show later that $D(p,i) \subseteq F(p,i)$.\footnote{The equation $D(p,i) = S(p,i) \subseteq F(p,i)$ holds for a small \emph{non-leech} $(p,i)$,
but does not necessarily hold for a small \emph{leech}. The definition of a leech (and non-leech) is given in Section 3.1.3.}
%For leech
This means that a small node must have less than $2k+2$ surrogates. 
Even though we may upper bound the number of surrogates by $k+1$, we do not attempt to do this;
it turns out that allowing more than $k+1$ (but up to $2k+2$) surrogates 
is useful for simplifying the analysis to some extent.
%Note that we do not attempt to upper bound the number of surrogates by $k+1$. (Even though we may also choose to stop at $k+1$ surrogates. 
 %This relaxation is not mandatory, i.e., we may also stop at $k+1$ surrogates.)
% for technical reasons that will become clear in the sequel.\footnote{explain that 
%this is not mandatory, only helps to simplfiy the analysis]].}
% For technical reasons we do not5
%require that the number of surrogates is bounded above by $k+1$--it may be larger than $k+1$ but smal
%need
%to have 
%may have more than $k+1$ surrogates but less than $2k+2$.\footnote{We may also stop at $k+1$ surrogates, but we make no effort to do this.}
%We will show later that all children of a sma
Assuming by induction that $S(x) = D(x)$ holds for every child $x$ of $(p,i)$,  we have
$S(p,i) = D(p,i) = \bigcup_{x \in Ch(p,i)} D(x) = \bigcup_{x \in Ch(p,i)} S(x)$. 
In other words, the old surrogates (from the surrogate sets of $(p,i)$'s children) are re-used.
%and additional new surrogates (that belong to the friend set of $(p,i)$ but do not belong to the friend sets of its children, if any) are appointed.
%As mentioned above, the useful reserve set of any small node is empty. 
%Moreover, 
%For a small node, we have $D(p,i) \subseteq F(p,i)$.
% points $q$ that are within distance $10 \cdot 5^i$ from $p$, ; we refer to such a point $q$ as a \emph{10-friend} of $(p,i)$. 
%It is not difficult to show that all 10-friends of a small  node (in particular, all its surrogates) have small degree.
%in particular, all surrogates of a small node have small degree.
%We will henceforth say that a surrogate of a small node is \emph{clean}; such a surrogate  remains clean as long as it belongs to small nodes, even though it %may be appointed as a surrogate of small nodes for many levels.

Note that a small node re-uses the surrogates of its children.  % attempting to become a complete node.
%This persistent effort to expand the surrogate set is repeated until the node becomes large.
%Once it became complete, it will re-use the same surrogates over and over until it becomes large.
%On the other hand,
%aiming to turn incomplete nodes into complete ones at the lowest possible level. 
%Appointments are made only for large nodes.
%Once a point is appointed as a surrogate of a large bag it becomes \emph{dirty}.
Similarly to small nodes, we would like a large node to re-use the surrogates of (one of) its children.
However, sometimes a large node must appoint new surrogates (this happens quite rarely).
Specifically, an appointment of new surrogates to a large node $x$ is made for a
%Exactly $k+1$ clean points from the reserve sets and surrogate sets]] of its children are appointed as surrogates.
%As mentioned above, this appointment
long \emph{term} of at least $\tau + 3 = \Omega(\tau)$ levels; recall that $\tau = \lceil \log_5 \gamma \rceil + 1$.
During this term, each ancestor of $x'$ of $x$ would re-use the same surrogates of $x$ (in fact, $x'$ may choose to use the surrogates of another descendant of $x'$); we call $x$ an \emph{appointing node}. 
The surrogates $S(x)$ of $x$ are chosen arbitrarily from its friend set $F(x)$, and so $S(x) \subseteq F(x)$.
At the end of such a term, $k+1$ new surrogates are appointed to the corresponding large ancestor $y$ of $x$
%which will also be an appointing node.
from the friend set of $y$.  % $y$ is also an appointing node.
 % we call $x$ an \emph{appointing node}.
We need to show that the friend set $F(y)$ of the appointing node $y$ at this stage 
%the friend set of the appointing node contains 
contains at least $k+1$ 10-friends of $y$ of small degree, which can be appointed as its new surrogates.
In fact, we will show something stronger, namely, that $F(y)$ contains at least $2k+2$ (rather than $k+1$) such friends.
Hence $k+1$ arbitrary points out of them will be appointed as the new surrogates of $y$, and the other at least $k+1$ such friends may be used for compensating neighboring (due to cross edges) nodes.   %these \emph{compensating points} will also play a major role in the construction.
%due to cross edges.
%will be donated may be needed to compensate nearby nodes.

%Each point is marked either as \emph{clean} or as \emph{dirty}.

%\vspace{0.07in}
%\noindent
%{\bf 3.1. ~Clean Versus Dirty.~}
At the outset all points are marked as \emph{clean}. A point $p$ remains clean as long as it is appointed as a surrogate of small nodes.
%Whenever a point is appointed as a surrogate of a small node, it remains clean.
%Once a node becomes large we appoint to it 
However, at the first time $p$ is appointed as a \emph{new} surrogate of a large node
it becomes \emph{dirty}. 
%it becomes \emph{semi-dirty}. 
%It becomes \emph{fully dirty} only at the second time it is appointed as a new surrogate (of necessarily a large node).
Once $p$ becomes dirty, it will never be re-appointed  as a new surrogate again.
Thus, each point can be appointed as a new surrogate of at most one large node, and this appointment is for a term of at least $\tau + 3 = \Omega(\tau)$ levels.
% once at an atomically complete node, and another time at a compound node. 
We will make sure that during such a term the degree of any surrogate will increase by at most  $\eps^{-O(d)} \cdot (2k+1)$ units, and so the degree bound will be in check.

\ignore{
We also use another notion of \emph{semi-clean point}, which is on the border between being clean and being dirty.
though it is more clean than dirty]]
In other words, a semi-clean point may be viewed as either clean or dirty, depending on the circumstances.
As mentioned, a large node $(p,i)$ appoints to itself $k+1$ new surrogates from $F(p,i)$.
We said that the points from $F(p,i)$ that do not get appointed as new surrogates of $(p,i)$ may be used for compensating neighboring nodes;
more specifically, these points will be marked as \emph{semi-clean}.
%We distinguish semi-clean points from clean and dirty points as follows. 
We say that a clean (respectively, dirty)
point that is not marked as semi-clean is \emph{fully-clean} (resp., \emph{fully-dirty}).
Once a point $q$ becomes semi-clean, we try to prevent it from getting re-appointed as a surrogate of small nodes again;
more specifically, from that stage on $q$ may serve as a surrogate only for small nodes
$v$ for which $q$ is a surrogate of at least one child of $v$; more details will be provided in Section 3.1.4.
At a subsequent stage $q$ may get appointed as a new surrogate of a large node and then get fully-dirty.
}

%More specifically, 
%We say that $x$ is a \emph{first-term node} if its appointing copy is atomically large.

Our construction will guarantee that either all surrogates of a node are clean, or they are all dirty.
We henceforth say that a node is \emph{clean} (respectively, \emph{dirty}) if  all its surrogates are clean (resp., dirty).
A dirty node will be complete, whereas a clean node may be complete or incomplete.
Also, a large node will be dirty, whereas a small node is usually (except for a leech) clean; see Section 3.1.3 for more details.

%A semi-clean point is viewed as clean by dirty nodes, and can thus be used as a new surrogate for such nodes. On the other hand, we try to prevent 
%clean nodes $v$ from using a semi-clean point as a surrogate whenever possible, specifically, whenever this point is not a surrogate of any child of $v$;
%thus, roughly speaking, a semi-clean point is viewed as dirty by clean nodes.

%We will later show that the surrogate set $S(p,i)$ of a small  node $(p,i)$ contains all its 10-friends (up to $k+1$),and so 

%while the old ones which were 
%marked as dirty will never be appointed again.

\vspace{0.07in}
\noindent
{\bf 3.1.3 ~A Symbiosis Between Leeches and Hosts.~} 
%\label{symb}
Appointing new surrogates to large nodes is a costly operation, and should be avoided whenever possible.
Specifically, before new surrogates are appointed to a large node $(p,i)$ (which happens at the beginning of each new term),
we look for a dirty node $(q,i)$ whose surrogates' term has not finished yet, 
which is a 24-friend of $(p,i)$, i.e., $\delta(p,q) \le 24 \cdot 5^i$. 
%such a node $(q,i)$ is called a 24-friend of $(p,i)$.
If no such node is found, we appoint $k+1$ new surrogates to $(p,i)$ as described above.
Otherwise, %let $(q,i)$ be an arbitrary such node. In this case 
we re-use the surrogates of an arbitrary such 24-friend $(q,i)$ of $(p,i)$ as the surrogates of $(p,i)$,
i.e., we assign $S(p,i) = S(q,i)$. Observe that $(q,i)$ is dirty and thus also complete;
hence this assignment turns $(p,i)$ too into a dirty and thus complete node.  % (though not necessarily to a large node). 
We say that $(q,i)$ is a \emph{leech} and $(p,i)$ is its \emph{host}. 

As mentioned, it is advantageous for a node to have (at least) $k+1$ surrogates, or in other words, to be complete.
We use the leech-host idea described above to turn clean (and possibly incomplete) nodes into dirty (and necessarily complete) ones. Specifically, 
recall that a clean node $(p,i)$ re-uses all the surrogates of its children by appointing them as its surrogates.
Before this is done, we look for a potential host for $(p,i)$, i.e., a dirty non-leech 24-friend $(q,i)$ of $(p,i)$.
If no such  node is found, we assign $S(p,i) = D(p,i)$ as before. Otherwise, we re-use the surrogates
of $(q,i)$ as the surrogates of $(p,i)$, i.e., we assign $S(p,i) = S(q,i)$, which turns $(p,i)$ into a dirty and  complete node.

%(It is easy to see that if $(q,i)$ is a 24-friend of $(p,i)$,
%then $(q,i)$'s parent, as well as all 10-friends of $(q,i)$, are 10-friends of $(p,i)$'s parent; this will be useful in the sequel.)
%we also add a pointer from $(q,i)$ to $(p,i)$ for future reference.
A node may be the host of many leeches, whereas a leech has only one host. Moreover, a node cannot be both a leech and a host, and it may be neither of them. 
Being a leech is an inherited trait: If $(p,i)$ is a leech of $(q,i)$, then $(p,i)$'s parent is close enough (i.e., a 24-friend) to become a leech
of $(q,i)$'s parent (this is irrelevant if the term of the host's surrogates is   over). 
However, in case $(p,i)$ has a dirty non-leech sibling whose surrogates' term is not over, we will prevent $(p,i)$'s parent from becoming a leech.
In this case $(p,i)$'s parent will re-use the surrogates of such a dirty non-leech sibling of $(p,i)$, and will become a dirty non-leech.
This is the only scenario when we prevent a node from becoming a leech. 

A clean node is necessarily small. 
On the other hand, a dirty node is not necessarily large.  % (see Section 3.1.3 for more details).
%In Section 3.1.2 we mentioned that a dirty node is not necessarily large.
The only exception is when the dirty node is a leech. Indeed, a leech may be both small and dirty.
%Note that a non-leech large node must be dirty. This means that a leech (of a large dirty node) is dirty as well.
%However, a leech may be small. 
%On the other hand, a leech is not too small, specifically, we will show that any ancestor of a leech must be large (see Lemma ??).
%
On the other hand, we will show that a non-leech dirty node must be large (see Lemma \ref{key}). 
We will also show that for a clean node $(p,i)$, we have $D(p,i) = S(p,i) \subseteq F(p,i)$. For a dirty node, this equation does not hold.
We remark that a clean node may be complete or incomplete, and it may contain up to $2k+1$ surrogates. On the other hand, a dirty node
must be complete, and it contains exactly $k+1$ surrogates.

%Hence, when the need to look for a host (i.e., a dirty non-leech 24-friend) arises, we actually need to look for a dirty non-leech  24-friend.
%Define \emph{exceptional node}]]
%I'm not sure we need this]]
%Consequently, when our algorithm needs to look for a potential host, it will 
%Moreover, we will show that if $(p,i)$ is a small leech, then any ancestor of $(p,i)$ will be large (see Lemma ??).
%We will also show that a non-leech dirty node is necessarily large
%(See Appendix \ref{symb} for more details.)

%In the sequel we will show that a non-leech node is small (respectively, large) iff its surrogates are clean (resp., dirty).
%We may want to refer to the node itself as dirty/clean;
%Our algorithm will henceforth use this clean versus dirty terminology in the sequel.
%Should also mention that a leech's parent, and all its ancestors, will be large]]

There is an interesting interplay between leeches and hosts.
On the negative side, a leech exploits the host by using its surrogates and overloading their degree due to cross edges that are incident on that leech.
% but noton the host.  % (there are only a few such cross edges).
However, each host will have at most $O(1)^{O(d)}$ leeches at each level, and so the surrogates' degrees
due to leeches will be greater than the surrogates' degrees due to the host by a small factor.
%So the downside of exploiting a host in this way is negligible.
\\On the bright side, there are important advantages of exploiting the host.
First, this allows us to ``ignore'' some cross edges of the basic spanner construction (and  reduce the surrogates' degrees), specifically, 
those between leeches and their host as well as between all leeches of the same host; indeed,
such cross edges connect nodes with the same surrogate set, and are thus \emph{redundant}. %Thus we will ignore such edges in the sequel.
%--i.e.,do not need any attention, and can be removed from the spanner.
Second, since leeches do not appoint new surrogates, the clean friends in $F(y)$ of the corresponding host $y$
will remain clean until the term of its surrogates $S(y)$ is over. Consequently, this approach enables $y$ and its non-leech ancestors to accumulate more and more clean 
points in their friend sets.
%the friends of a leech remain clean (some of them at least), 
%and will remain clean until the term of the surrogates of the relevant host is over.
Whenever a new term starts, new host and leeches are determined. 
At this stage we will have enough clean points in $F(\tilde y)$
%(i.e., the clean friends of the descendant leeches) 
to appoint as surrogates of the new host $\tilde y$.
%and this is due to the fact that the clean friends of the host remain clean during the entire term.
%This way whenever the tenure period of the host is over, we will have
%enough clean points to appoint as surrogates. 
Finally, we remark that an ancestor of a leech (respectively, host) can become a host (resp., leech).
%of some bag, and vice versa, a. 
To summarize, there is a symbiosis between leeches and hosts, from which everyone enjoy.  %both the host and its leeches enjoy.

%should emphasize that we don't add edges between leeches and their hosts, or between each other;
%should emphasize that hosts have disjoint surrogate sets]]

%control this heredity and
%Whenever the term of the host's surrogates is over, the connection between the leeches and their host is broken.
%corresponding host as well as its leeches become neutral.
%We will address these issues later.

\vspace{0.07in}
\noindent
{\bf 3.1.4 ~Putting it All Together.~}
Recall that once the surrogate sets $S(x)$ of all nodes $x$ in $T$ are computed, the FT spanner $\mathcal H$ is obtained
by replacing each edge of the basic spanner $H$ with a bipartite clique. We remark that
\emph{redundant edges} that connect nodes with the same surrogate set are disregarded.

We next show how to compute the surrogate sets of nodes,
by putting all the ingredients that were described in Sections 3.1.1-3.1.3 together.
The surrogate sets and the auxiliary sets (the descendant, friend and reserve sets) are computed bottom-up. 
That is, we first compute these sets for the 0-level nodes (i.e, the leaves of $T$), then for the 1-level nodes, and so forth.
More specifically, 
we employ Procedure $ComputeSets_{(i)}$ to compute these sets for the $i$-level nodes, for $i = 0,\ldots,\ell$.
Procedure $ComputeSets_{(i)}$ has two parts. The first part of this procedure computes the descendant, friend and reserve sets of the $i$-level nodes.
(We omit the  computation of the descendant sets of the $i$-level nodes in the first part of Procedure $ComputeSets_{(i)}$, $i \in [0,\ell]$, as it can be done in a straightforward manner.)
Equipped with these sets, the second part of Procedure $ComputeSets_{(i)}$ computes the surrogate sets of the $i$-level nodes.

%We remark that the second part may turn clean points into dirty ones, and so we may need to update the friend sets of some nodes
%during the second part; alternatively, we may leave the friend sets of $i$-level nodes intact during the second part of Procedure $ComputeSets_{(i)}$, but note %that each such set may contain
%some (at most $k+1$) dirty points.

There is a   difference between Procedure $ComputeSets_{(0)}$ (which corresponds to the case $i=0$) and Prcoedure $ComputeSets_{(i)}$, for $i \ge 1$.
We start with describing Procedure $ComputeSets_{(0)}$.

The first part of Procedure $ComputeSets_{(0)}$ goes over the $0$-level (leaf) nodes in an arbitrary order.
%To determine whether a non-leech node is small or large, we will check whether its surrogates are clean or dirty.
%If they are clean it is small
Consider an arbitrary leaf node $(p,0)$. We first put $p$ in both the reserve set $R(p,0)$ and the friend set $F(p,0)$ of $(p,0)$.
Next, for every cross edge $((p,0),(q,0))$ that is incident on $(p,0)$ in the basic spanner $H$  (including redundant edges),
we put $q$ in the reserve set $R(p,0)$ of $(p,0)$.
%; we stop once $|R(p,0)| = 2k+2$ (keeping in $R(p,0)$ just the $2k+2$ closest points to $p$).
Note that $q$ is a $\gamma$-friend of $(p,0)$. If $q$ is also a 10-friend of $(p,0)$,
we add it to the friend set $F(p,0)$ of $(p,0)$; we stop once $|F(p,0)| = 3k+3$.

Having computed the friend and reserve sets $F(p,0)$ and $R(p,0)$ for all leaf nodes $(p,0)$,  the second part of Procedure $ComputeSets_{(0)}$ starts.
As opposed to the first part of the procedure, here it is important to handle nodes $(p,0)$ with $|F(p,0)| \ge 2k+2$ (if any) first, 
and only later handle the rest of the nodes. 
(See the first remark below to understand why this order is important.)
Consider a leaf node $(p,0)$.
%Hence $|F(p,0)| \le |R(p,0)| \le 2k+2$.
%Thus, Its friend set $F(p,0)$ consists of all its 10-friends
%If one of the 24-friends of $(p,0)$ is large and has already appointed $k+1$ new surrogates (which are dirty), $(p,0)$ will become its leech.
We first look for a potential host for $(p,0)$, i.e., a dirty non-leech 24-friend $(q,0)$ of $(p,0)$. 
%remove we will show (see Lemma \ref{key}) that $(q,0)$ must be large]].
% that has already appointed $k+1$ new surrogates,
%which are now dirty need to change this accordingly to above]]. 
\begin{itemize}
\item If such a node $(q,0)$ is found, then we assign
$S(p,0) = S(q,0)$, and $(p,0)$ will be a leech of $(q,0)$, and thus dirty and complete.
%also, all points of $F(p,0)$ that are fully-clean are marked as semi-clean.
%moreover, we will show (see Lemma ??) that $(p,0)$'s parent must be large.
\item Otherwise we check whether $|F(p,0)| < 2k+2$. 
\begin{itemize}
\item If so, $(p,0)$ will be small and clean, and we assign $S(p,0) = D(p,0) = \{p\}$.
(Notice that $S(p,0) \subseteq F(p,0)$, and so $|F(p,0)\cup S(p,0)| < 2k+2$, and $(p,0)$ is small by definition.)
%put in $S(p,0)$ all fully-clean points in $F(p,0)$, as well as $p$ (even if $p$ is semi-clean).
\item In the complementary case (i.e., $|F(p,0)| \ge 2k+2$) $(p,0)$ will be large; 
in this case we appoint $k+1$ points from $F(p,0)$ as \emph{new surrogates} of $(p,0)$
for a term of $\Omega(\tau)$ levels (see the second remark below),
and these points become dirty. Also, $(p,0)$ itself will be dirty and complete.
\end{itemize}
\end{itemize}
%Also, all points of $F(p,0)$ that are fully-clean are marked as semi-clean.
%; if $(p,0)$ is large, then these surrogates will become dirty.
%If $|S(p,0)| = k+1$,
%$(p,0)$ will be complete.
%We compute the
We remind that at the outset all points are marked as clean.
A point $p$  becomes dirty when it is appointed as a new surrogate of a large node; %, this point is marked as dirty.
%new surrogates are appointed to a compound node (with at least one complete child), these surrogates are marked as dirty;
$p$  will not be re-appointed as a new surrogate again. 
We say that a node is clean (respectively, dirty) if all its surrogates are clean (resp., dirty).
%at subsequent levels.

Next, we describe Procedure $ComputeSets_{(i)}$, for $i \ge 1$,
which computes the surrogate set $S(p,i)$, the descendant set $D(p,i)$, the friend set $F(p,i)$ and the reserve set $R(p,i)$, for all $i$-level nodes $(p,i)$.
We invoke Procedure $ComputeSets_{(i)}$   %, we turn to computing these sets for the $i$-level nodes, for any $i \ge 1$.
only after all Procedures $ComputeSets_{(0)},\ldots,ComputeSets_{(i-1)}$ have terminated.
%which compute the sets of the $i$-level nodes in $T$.
At this stage, the surrogate set $S(q,j)$, the descendant set $D(q,j)$, the friend set $F(q,j)$ and the reserve set $R(q,j)$ have been computed, for all nodes $(q,j)$ at levels $j \in [0,i-1]$.

The first part of Procedure $ComputeSets_{(i)}$ goes over the $i$-level  nodes in an arbitrary order.
Consider an arbitrary $i$-level node $(p,i)$. We  first put in $R(p,i)$ (respectively, $F(p,i)$) all \emph{clean} points
of $R(q,i-1)$ (resp., $F(q,i-1)$), for every child $(q,i-1)$ of $(p,i)$.
In addition, all (clean) points of $R(p,i)$ that are 10-friends of $(p,i)$ but do not belong to $F(p,i)$ are added to $F(p,i)$;
we stop once $|F(p,i)| = 3k+3$.   %  (keeping in $F(p,i)$ just the $2k+2$ closest points to $p$).
This concludes the computation of $F(p,i)$,
but the computation of $R(p,i)$ is not over yet. We pause the computation of $R(p,i)$ for now,
and proceed to computing the friend sets $F(x)$ of all $i$-level nodes in this way.
%Next, 
Having computed the friend sets of all $i$-level nodes, we return to completing the computation of $R(p,i)$.
%we proceed as follows. (Again, $(p,i)$ stands for an arbitrary $i$-level node.)
For every cross edge $((p,i),(q,i))$ that is incident on $(p,i)$ in the basic spanner $H$ (including redundant edges),
we add to the reserve set $R(p,i)$ of $(p,i)$ all (clean) points from the friend set $F(q,i)$ of $(q,i)$ 
as well as  all  (clean) 10-friends of $(q,i)$ from the reserve set $R(q,i)$ of $(q,i)$,
disregarding points that were already in $R(p,i)$.
%and (ii)
%in addition, we add to $(p,i)$ 
 % that were not already in $R(p,i)$.
%(The points of $F(q,i)$ and $R(q,i)$ that are added to $R(p,i)$ at this stage need not be added to $F(p,i)$;
%indeed, it can be easily verified that, unless $|F(p,i)| = 3k+3$, all the 10-friends of $(p,i)$ among these points 
% already belong to $F(p,i)$.)
%not 10-friends of $(p,i)$; this is because the 10-friends of $(p,i)$
%we stop once $|R(p,i)| = 2k+2$ (keeping in $R(p,i)$ just the $2k+2$ closest points to $p$; this can be done in $\eps^{-O(d)} (k)$ time).
This concludes the computation of $R(p,i)$. 
%(We can also check if the new points that are added to $R(p,i)$ d
%It can be verified that all points of $R(p,i)$ that are 10-friends
%at this step
%are not 10-friends of $(p,i)$, and thus should not be added to $R(p,i)$]]
\\
Observe that any 10-friend of a child $(q,i-1)$ of $(p,i)$ is also 
a 10-friend of $(p,i)$ (see Claim \ref{anctrivial} in Section 3.2.1).
Hence, by construction, all points of $F(p,i)$ are 10-friends of $(p,i)$.
% $\gamma$-friend (resp.,  $(\gamma+10)$-friend (respectively,
Similarly, any $(\gamma+10)$-friend of a child $(q,i-1)$ of $(p,i)$ is also 
a $(\gamma+10)$-friend (in fact, a $\gamma$-friend) of $(p,i)$.
Also, recall that the weight of $i$-level cross edges in the basic spanner $H$ is at most $\gamma \cdot 5^i$.
Hence, by construction, all points of $R(p,i)$ are $(\gamma+10)$-friends of $(p,i)$.
%explain  that we  consider $\gamma+68$ due to leeches, inaccurate]]

%Let $\hat R(p,i)$ be the set of points in the reserve sets of $(p,i)$'s children that 
%are within distance $10 \cdot 5^i$ from $p$.

Having computed the friend and reserve sets $F(p,i)$ and $R(p,i)$ for all $i$-level nodes $(p,i)$,  the second part of Procedure $ComputeSets_{(i)}$ starts.
As opposed to the first part of the procedure, here it is important to handle nodes $(p,i)$ with $|F(p,i)| \ge 2k+2$ first, 
and only later handle the rest of the nodes (see the first remark below).
Consider a node $(p,i)$.
%We start with computing $S(p,i)$. 
%Having computed $F(p,i)$ and $R(p,i)$, we proceed to computing $S(p,i)$.
%Add a remark: we need first to go over node $(p,i)$ with $|F(p,i) \ge 2k+2$, and only later move to ''smaller'' nodes!!!
%Because we don't want a clean node to have a dirty non-leech 24-friend]]
The execution of the procedure splits into  two   cases.
\\\emph{Case 1: All children of $(p,i)$ are clean.}
We first look for a potential host for $(p,i)$, i.e., a dirty non-leech 24-friend $(q,i)$ of $(p,i)$.  % we will show (see Lemma ??) that $(q,i)$ must be large.
% that has already appointed $k+1$ new surrogates,
%which are now dirty need to change this accordingly to above]]. 
\begin{itemize}
\item If such a node $(q,i)$ is found, then we assign
$S(p,i) = S(q,i)$, and $(p,i)$ will be a leech of $(q,i)$, and thus dirty and complete.
%also, all points of $F(p,i)$ no, because $F(p,i)$ may contain $3k+3$ points, and thus it may not contain some of the points in the friend
%sets of its children; so we better go over all points in the friend sets of its children (or over the surrogate sets); apply the same change also below;
%this is relevant at least for Claim \ref{fullyclean}]]
%that are fully-clean are marked as semi-clean.
%moreover, we will show (see Lemma ??) that $(p,i)$'s parent must be large.
%put in $S(p,i)$ all points of $S(q,i-1)$, for every child $(q,i-1)$ of $(p,i)$.
%We also add to $S(p,i)$ all points in $F(p,i)$ (that were not already in $S(p,i)$); we stop once $|S(p,i)| = k+1$.
%Since all children of $(p,i)$ are incomplete, 
%the reserve sets of $(p,i)$'s children that are within distance $2 \cdot 5^i$ from $p$.
%it wlll be shown later that 
%all points that are added to $S(p,i)$ are clean.   % (see below a reminder of clean and dirty nodes).
%If $|S(p,i)| < k+1$, then $(p,i)$ is incomplete. Otherwise it will be complete.
\item Otherwise we check whether $|F(p,i)| < 2k+2$. 
\begin{itemize}
\item If so, $(p,i)$ will be small and clean, and we assign $S(p,i) = D(p,i)$.
(We will show in Corollary \ref{bas2} that $S(p,i) \subseteq F(p,i)$, and so $|F(p,i)\cup S(p,i)| < 2k+2$, and $(p,i)$ is small by definition.)
%and we put in $S(p,i)$ all points from the surrogate sets of $(p,i)$'s children, i.e., all points in $\cup_{w \in Ch(p,i)} S(w)$.
%In addition, we put in $S(p,i)$ all fully-clean points in $F(p,i)$ that are not already there.
%\i.e., that are not semi-clean]]; we also put there
%all points from the surrogate sets of $(p,i)$'s children rewrite]].
\item In the complementary case (i.e., $|F(p,i)| \ge 2k+2$) $(p,i)$ will be large; in this case we appoint $k+1$ points from $F(p,i)$ as \emph{new surrogates} of $(p,i)$
for a term of $\Omega(\tau)$ levels (see the second remark below),
and these points become dirty. Also, $(p,i)$ itself will be dirty and complete.
\end{itemize}
\end{itemize}
%Also, all points of $F(p,i)$ that are fully-clean are marked as semi-clean.
%Otherwise, the node $(p,i)$ will be large. 
%These points are marked as dirty, and appointed as surrogates of $(p,i)$ with a tenure period of $\tau$ levels
%need to explain that the tenure period does not start immediately, but only whenever the degree due to cross edges
%becomes k+1]].
\emph{Case 2: At least one child of $(p,i)$ is dirty.}
We first look for the  dirty non-leech child $(q,i-1)$ of $(p,i)$ (if any) whose surrogates' term started most recently.
(In fact, our analysis shows that any dirty non-leech child of $(p,i)$ whose surrogates' term is not over at level $i-1$ will do.
However, it is more instructive and natural to choose the one whose surrogates' term started most recently.) 
% (and is thus going to finish last). 
\begin{itemize}
\item If such a child $(q,i-1)$ is found and the term of its surrogates  is not over at level $i-1$, we assign $S(p,i) = S(q,i-1)$. In this case $(p,i)$ becomes a dirty and complete non-leech. %(Notice that $(p,i)$ will be a non-leech.)
%also, all points of $F(p,i)$ that are fully-clean are marked as semi-clean.
\item Otherwise, the term of the surrogates of each dirty non-leech child of $(q,i-1)$ (if any) will be over at level $i-1$.
%and thus they cannot be re-used.
In this case we proceed similarly to case 1. Specifically, we first look for a potential host $(q,i)$ for $(p,i)$, i.e., a dirty non-leech 24-friend $(q,i)$ of $(p,i)$. 
\begin{itemize}
\item If such a node $(q,i)$ is found, then we assign $S(p,i) = S(q,i)$,
and $(p,i)$ will be a leech of $(q,i)$, and thus dirty and complete.  % also, all points of $F(p,i)$ that are fully-clean are marked as semi-clean.
%Moreover, we will show (see Lemma ??) that  $(p,i)$ is in fact large. 
\item Otherwise, we appoint $k+1$ points from $F(p,i)$ as \emph{new surrogates} of $(p,i)$ 
for a term of $\Omega(\tau)$ levels (see the second remark below), and these points become dirty. Also, $(p,i)$ itself will be dirty and complete.
%Also, all points of $F(p,i)$ that are fully-clean are marked as semi-clean maybe we may say once that the fully-clean points in $F(p,i)$
%are marked as semi-clean for all dirty nodes]].
%Otherwise, we put in $S(p,i)$ all points of $S(q,i-1)$, for every \emph{incomplete} child $(q,i-1)$ of $(p,i)$.
%(All surrogates of incomplete nodes are clean.
%On the other hand, the complete children of $(p,i)$ cannot contribute their old surrogates to $S(p,i)$, because either their term is over or they are leeches
%of hosts whose term is over--this will be shown later.)
%We also add to $S(p,i)$ all points of $F(p,i)$ (that were not already in $S(p,i)$), stopping once $|S(p,i)| = k+1$. 
Unlike case 1 where we may have $|F(p,i)| < 2k+2$ and $(p,i)$ will be clean,
in this case it must hold that $|F(p,i)| \ge 2k+2$, implying that $(p,i)$ is large, and thus dirty and complete.
 % (we will show this later).
%^If all complete children of $(p,i)$ are leeches, then we look for a potential host for $(p,i)$.
%In particular, let $h_1,\ldots,h_m$ be the $(i-1)$-level hosts of the complete children of $(p,i)$. If the term of one of them
%will not be over at level $i$, then its parent will serve as a leech for $(p,i)$. If no such host is found, then we appoint $k+1$ new clean surrogates 
%for $(p,i)$ from $F(p,i)$.   % the reserve sets of $(p,i)$'s children, considering only points that are within distance $2 \cdot 5^i$ from $p$;
%We will later show that in this case, it holds that $|\hat R(p,i)| \ge k+1$.
%Otherwise, we look for 
%If the term of $(q,i-1)$ will be over
%at level $i$, then we appoint $k+1$ new clean surrogates for $(p,i)$ from $F(p,i)$.
%the reserve sets of $(p,i)$'s children,  considering only points that are within distance $2 \cdot 5^i$ from $p$;
To this end, we will show (see Corollary \ref{defined}) that whenever the need to appoint new surrogates for an ancestor $(p,i)$ of a dirty node arises, we have $|F(p,i)| \ge 2k+2$.
\end{itemize}
\end{itemize}

\vspace{0.07in}
\noindent
{\bf Remarks:} 
\begin{enumerate}
\item The second part of Procedure $ComputeSets_{(i)}$ (for any $i$, including $i=0$) %The computation of the sets $S(p,i)$ 
%is carried out 
handles   nodes $(p,i)$ with $|F(p,i)| \ge 2k+2$ (if any) first, and only later handles the rest of the nodes. 
%The reason is that we do not want to miss potential hosts. 
The reason we handle nodes in this order is that when looking for a potential host $(q,i)$ for $(p,i)$, we are in fact
looking for a dirty non-leech 24-friend of $(p,i)$. However, we do not want $(q,i)$ to become dirty after handling $(p,i)$, 
assuming $(p,i)$ is clean, thus missing a potential host for $(p,i)$. Handling nodes in this order  guarantees that all the 24-friends of a dirty non-leech node
must be dirty as well. 
\item We mentioned that new surrogates of large nodes are appointed for a term of $\Omega(\tau)$ levels. %This is a critical issue.
The term of surrogates
consists of two phases. The first phase starts with their appointment, and ends when their degree due to non-redundant cross edges (i.e., disregarding
cross edges that connect nodes with the same surrogate set) \emph{since their appointment}
reaches $k+1$. Notice that the degree of a surrogate before its appointment at a large node may be positive, 
but we only consider its degree since this appointment. Moreover, observe that the degree of a surrogate may also increase
due to cross edges that are NOT incident on the relevant host but rather on the leeches of this host.  % and not on the host.

The first phase may end within a single level, but may also last for many levels. When the second phase
starts, the degree of surrogates (since their appointment) is between $k+1$ and  $k + \xi^2(2k+1)$.
The second phase lasts precisely  $\tau + 2$ levels. Thus any term seems to last at least $\tau + 3$ levels.
However, this is true only if these surrogates are re-used over and over, which need not be the case in general. 
Recall that the computation of $S(p,i)$ in Case 2 above started by looking for the dirty non-leech child $(q,i-1)$ of $(p,i)$ 
whose surrogates' term started most recently, and setting $S(p,i) = S(q,i-1)$ if such a child is found.
Consider another non-leech child $(r,i-1)$ of $(p,i)$. Setting $S(p,i) = S(q,i-1)$ forces the term of points in $S(r,i-1)$ to be over
at level $i-1$. (We will show in Claim \ref{fofx} that the surrogate sets of dirty non-leech nodes at the same level are pairwise disjoint,
and so $S(q,i-1) \cap S(r,i-1) = \emptyset$.)
But this is fine, because the surrogates of $S(p,i) = S(q,i-1)$ that $(p,i)$ chose to re-use
are already sufficient for our needs, and so we can safely disregard the surrogates of all other dirty non-leech children of $(p,i)$.
%are at least as good as the surrogates of $(r,i-1)$ for our needs.
%should last for many levels.
%will last at least until the term of the surrogates of $(r,i)$ should last (this follows from the choice of $(q,i-1)$).
In particular, this approach guarantees that once a large $i$-level node $x$ appoints to itself $k+1$ new surrogates, no ancestor of $x$ at level at most $i + \tau + 2$ 
will appoint to itself new surrogates. 
\end{enumerate}
\subsection{Degree Analysis}
In this section we analyze the degree of the spanner construction $\mathcal H$ described in Section 3.1.

%In this appendix we tie up   the loose ends that were left in Section 3, thus completing the degree analysis.
%\subsection{Auxiliary Statements}

\vspace{0.07in}
\noindent
{\bf 3.2.1 ~Getting Started.~}
We start with  some basic definitions and claims.
% observations and claims.
For a node $x = (p,i)$, recall that $p(x)$ denotes the corresponding net-point $p$.
The distance between two nodes $x$ and $y$ is defined as the distance between their  net-points $p(x)$ and $p(y)$,
i.e., $\delta(x,y) = \delta(p(x),p(y))$. The distance between a node $x$ and a point $p$ can be defined similarly, i.e., $\delta(x,p) = \delta(p(x),p)$.
Recall that an $i$-level node $x$ is called a \emph{$t$-friend} of another $i$-level node $y$ (respectively, of a  point $p$)
if $\delta(x,y) \le t \cdot 5^i$ (resp., $\delta(x,p) \le t \cdot 5^i$). 
For a point $p$ and any $i \in [0,\ell]$, the \emph{($i$-level) base bag} of $p$ is the unique $i$-level node $x$ such that $p \in D(x)$.
%[[S: The $i$-level base bag of $p$ is denoted by $b_i(p)$.]]

%The following claim implies that indeed $F(\cdot)$ and $R(\cdot)$ contain what they are supposed to contain...]]
%The following lemma implies that we can take $F(x)$ and $S(x)$ of a node $x$ to contain all 
%points in $F(y)$ and $S(y)$ of its descendants [[S: nisuach!!!]]
\begin{claim} [Once a Friend, Always a Friend] \label{anctrivial}
~(1) If a point $p$ is a 35-friend (let alone a 10-friend) of an $i$-level node $x$, then it is also a 10-friend of every ancestor of $x$.
~(2) If a point $p$ is a $(\gamma+68)$-friend (let alone  a $\gamma$-friend) of $x$, then it is also a $\gamma$-friend of every ancestor of $x$.
\end{claim}
\begin{proof}
Suppose first that $p$ is a 35-friend of $x$. To prove the first assertion, it suffices to show that $p$ is a 10-friend of the parent $\pi(x)$ of $x$.
Indeed, we have $$\delta(\pi(x),p) ~\le~ \delta(\pi(x),x) + \delta(x,p) ~\le~ 3 \cdot 5^{i+1} + 35 \cdot 5^i ~=~ 10 \cdot 5^{i+1}.$$
%we have $\delta(y,x) \le 24 \cdot 5^i$. 

To prove the second assertion, suppose that $p$ is a $(\gamma+68)$-friend of $x$.  % and show that it is a $\gamma$-friend of $\pi(x)$.
Assuming $\gamma \ge 21$, we have $$\delta(\pi(x),p) ~\le~ \delta(\pi(x),x) + \delta(x,p) ~\le~ 3 \cdot 5^{i+1} + (\gamma+68) \cdot 5^i ~<~ \gamma \cdot 5^{i+1}. \inQED$$
\end{proof}

\begin{claim} [Being a Leech is an Inherited Trait] \label{trait}
If an $i$-level node $y$ is a leech of some $i$-level node $x$, then $y$'s parent $\pi(y)$ is a 24-friend of $x$'s parent $\pi(x)$, and is thus close enough to become its leech.
\end{claim}
\begin{proof}
If $y$ is a leech of $x$, then the net-point $p(y)$ of $y$ is a 24-friend of $x$. The first assertion of Claim \ref{anctrivial} implies that $p(y)$ is  a 10-friend of $\pi(x)$,
i.e., $\delta(p(y),\pi(x)) \le 10 \cdot 5^{i+1}$.
Since $\delta(\pi(y),y) = \delta(\pi(y),p(y)) \le 3 \cdot 5^{i+1}$, it follows that $$\delta(\pi(y),\pi(x)) ~\le~ \delta(\pi(y),p(y)) + \delta(p(y),\pi(x)) ~\le~ 3 \cdot 5^{i+1} + 10 \cdot 5^{i+1} ~<~ 24 \cdot 5^{i+1}.\inQED$$
\end{proof}
%\begin{claim} \label{crosses}

%\end{claim}

\begin{claim} [Distant Friends Soon Become Close Friends]  \label{friendship}
~(1) For any $i$-level cross edge $(x,y)$ of the basic spanner $H$ and any $p \in S(x), q \in S(y)$, $\delta(p,q) \le (\gamma + 68) 5^i$.
In other words, for any $i$-level cross edge $(p,q)$ in the FT spanner $\cH$, $\delta(p,q) \le (\gamma + 68) 5^i$.
~(2) Let $p$ be a 10-friend of some $i$-level node $x$.
Any point $q$ for which  $\delta(p,q) \le (\gamma+68) 5^i$ is a 10-friend of the $(i+\tau)$-level ancestor of $x$.
\end{claim}
\begin{proof}
The weight of an $i$-level cross edge $(x,y)$ in $H$ is at most $\gamma 5^{i}$.  % hence $\delta(x,y) \le 
For any $p \in S(x)$, we have $\delta(x,p) \le 34 \cdot 5^{i}$; 
%(This upper bound is realized when $x$ is a leech;
%if $x$ is a non-leech, we have $\delta(x,p) \le 10 \cdot 5^{i}$. 
refer to Observation \ref{farsur} and the  paragraph preceding it to see why this upper bound holds.
Similarly, we have $\delta(y,q) \le 34 \cdot 5^{i}$. 
Consequently, $\delta(p,q) \le \delta(p,x) + \delta(x,y) + \delta(y,q) \le (\gamma + 68)  5^{i}$.

%Next, consider a cross edge $(p,q)$ in the FT spanner $\cH$. The edge $(p,q)$ corresponds to a bipartite clique between $S(x)$ and $S(y)$,
%where $x$ and $y$ are $i$-level nodes for which $p \in S(x), q \in S(y)$.

Next, we prove the second assertion.
Let $x'$ be the $(i+\tau)$-level ancestor of $x$, and notice that $\delta(x', x) \le 4 \cdot 5^{i+\tau}$.
We have $\delta(x', q) ~\le~ \delta(x', x) + \delta(x,p) + \delta(p,q) ~\le~ 4 \cdot 5^{i+\tau} + 10 \cdot 5^i
+ (\gamma+68) 5^i ~\le~ 10 \cdot 5^{i+\tau},$
where the last inequality holds for $\gamma \ge 3$. It follows that $q$ is a 10-friend of $x'$, and we are done.
\QED
\end{proof}

By construction, a node with at least one dirty child must be dirty too. (See Case 2 in Procedure $ComputeSets_{(i)}$, $i \ge 1$.)
This means that all ancestors of a dirty node are dirty, or equivalently,
all descendants of a clean node are clean. 
Notice also that the construction guarantees that any dirty node must be complete. More specifically, a dirty node has exactly $k+1$ surrogates.
We summarize these observations in the following statement.
\ignore{
Also, we will show that all ancestors of a complete (respectively, large) node are complete (resp., large), 
all descendants of an incomplete (resp., small) node are incomplete (resp., small). 

We argue that \emph{all ancestors of a dirty node are in fact large}.
To see this, first note that any large node must be dirty.
Even though the converse statement holds for non-leech nodes by Lemma \ref{key}, it does not necessarily hold for leeches.
Thus, to conclude the assertion, it is left to show that the parent $\pi(x)$ of any $i$-level leech node $x$ must be large. 
Indeed, the host $z$ of $x$ must be large,
and so $|F(z)| \ge k+1$. 
If $term(z)$ is not over at level $i$, then $\pi(x)$ will be a 24-friend of $z$'s parent $\pi(z)$, and will thus be a leech.
Moreover, it is easy to see that all the 10-friends of $z$ will be 10-friends of $\pi(x)$, which means that $\pi(x)$
will be large. If $term(z)$ is over at level $i$, then Lemma \ref{key} yields $|F(z)| \ge 2k+2$. In this case $\pi(x)$ will have
at least $2k+2$ clean 10-friends at the beginning of level $i+1$. It is easy to see that at most $k+1$ of these points may be appointed as new surrogates
at level $i+1$ (possibly by $\pi(x)$ itself), while all the other at least $k+1$ points must remain clean at level $i+1$.
It follows that $\pi(x)$ will be large in this case as well. %The corollary follows.
}
%The following statement shows that all the ancestors of a dirty (and thus large) node are necessarily large (and thus dirty).
\begin{observation} [Once Dirty, Always Dirty and Complete] \label{ancestor}
%All ancestors of a dirty (possibly a small leech) $i$-level node $x$ are large (and thus dirty).
All ancestors of a dirty node (including itself) are dirty.  % and thus complete. 
Also, each dirty node has exactly $k+1$ surrogates, and is thus complete.
\end{observation} 
{\bf Remark:} Note that any large node is dirty.
Even though the converse statement holds for non-leech nodes (see Lemma \ref{key} in Section 3.2.2), it does not necessarily hold for leeches.
Nevertheless, it can be proved that all ancestors of a (possibly leech) dirty node are in fact large. (This assertion is   stronger than Observation \ref{ancestor}
and requires proof; we do not provide the proof of this assertion since we do not use it in the sequel.)

The following observation is implied by Observation \ref{ancestor} and the construction.
%\begin{observation} [The Surrogates of Clean Nodes are  Surrogates of   Their Clean Ancestors]
% \label{bas}
%For any descendant $x'$ of a clean $i$-level node $x$, we have $S(x') \subseteq S(x)$. 
%\end{observation}
%\begin{proof}
%It suffices to prove the lemma for a child $x'$ of $x$. (This argument can be applied inductively.)
%
%By construction, $S(x)$ contains all points of $S(x')$ that are still clean at level $i$.
%Note also that $x'$ must be clean by Observation \ref{ancestor}.
%Hence, Observation \ref{bas} implies that $S(x) = F(x), S(x') = F(x')$.
%
%If all points of $S(x')$ remain clean until level $i$, then we have $S(x') = F(x') \subseteq F(x) = S(x)$, and
%we are done.
%Suppose for contradiction otherwise, and let $p$ be some point of $S(x')$ which became dirty at level $i-1$.  % (and thus it would not belong to $S(x) = F(x)$).
%(Note that $p$ cannot become dirty before level $i-1$, as otherwise it would not belong to $S(x') = F(x')$.)
%This means that $p$ is appointed as a new surrogate of some non-leech $(i-1)$-level node $y'$, and $y'$ is thus dirty.
%Since $p$ is a 10-friend of both $x'$ and $y'$, it follows that $x'$ is a 20-friend of $y'$, i.e., it is close enough to become its leech.
%Hence $x'$ cannot be clean, a contradiction.
%\QED
%\end{proof}
%The following observation is implied by our construction.
\begin{observation}  [The Descendants of Clean Nodes are Surrogates] \label{bas}
For any clean node $x$, we have $S(x) = D(x)$.
%The surrogate set $S(x)$ of a clean node $x$ is equal to its descendant set $D(x)$, i.e., 
Since any descendant $x'$ of $x$ is clean, it follows  that $S(x') \subseteq S(x)$.
%all fully-clean points of $F(x)$
%the surrogate sets of all its children.
%descendants (i.e., its children, grandchildren, and so forth).
%In particular, applying this observation recursively yields $D(x) \subseteq S(x)$.  %all points of $D(x)$.
%as well as all clean (possibly semi-clean) 
%points of $D(x)$.
\end{observation}
{\bf Remark:} Let $p$ be an arbitrary point, and let $i \in [0,\ell]$. Observation \ref{bas} shows that the only clean $i$-level node $x$ (if any) of which $p$ is a surrogate (i.e., $p \in S(x)$) is the $i$-level base bag of $p$. Consequently, if the $i$-level base bag of $p$ is dirty, then $p$ will not be a surrogate
of $j$-level clean nodes, for any $j \ge i$.

%We will use the following definition in the sequel.

%Note that all surrogates of a clean node are clean by definition. 
The next lemma shows that 
%not only the surrogates of a clean node are clean,
%but 
all descendants (and thus all surrogates) of a clean node must be clean. 
%In particular, if a point is clean at level $j$,
%then all its $i$-level base bags, $j \in [0,i]$, are clean as well.
\begin{claim} [All Descendant Points of Clean Nodes are Clean] \label{friends}
%All 10-friends of a clean $j$-level node $x$ are clean at level $j$.
For a clean $j$-level node $x$, all points of $D(x)$ are clean at level $j$. 
In other words, if a point $p$ is dirty at level $j$, then the $j$-level base bag of $p$ is dirty too.
%Moreover, we have $D(x) \subseteq S(x)$. [[S: remove]]
\end{claim}
\begin{proof}
Suppose for contradiction that there is a point $p \in D(x)$ that is dirty at level $j$, and let $z$ be 
the appointing node which made $p$ dirty. Observe that $z$ is a dirty non-leech $i$-level  node, such that $p \in S(z)$
and $i \le j$.  
Let $\tilde x$ be the $i$-level base bag of $p$, i.e., $p \in D(\tilde x)$. 
Note that $x$ is the $j$-level base bag of $p$, and so 
$\tilde x$ is a descendant of $x$.
By Observation \ref{ancestor},
$\tilde x$ must be clean. Hence $\tilde x \ne z$.
Note also that $p$ is a 4-friend of $\tilde x$ and a 10-friend of $z$.   
It follows that $\tilde x$ is a 14-friend of $z$, and would become its leech and thus dirty, yielding a contradiction. \QED
\end{proof}

We have shown in Claim \ref{friends} that all surrogates of a clean node are clean.
Also, all surrogates of a dirty node are dirty by the construction. We summarize this property in the following statement.
%Finally, it is easy to see that all surrogates of a  clean (respectively, dirty) node are clean (resp., dirty).
\begin{corollary} \label{preanc}
All surrogates of a clean (respectively, dirty) node are clean (resp., dirty). 
\end{corollary}

The following corollary follows from Observation \ref{bas}, Claim \ref{friends} and the construction.
\begin{corollary}  \label{bas2}
For a clean node $x$, we have $S(x) = D(x) \subseteq F(x)$, and so $|S(x)|=|D(x)| \le |F(x)| = |F(x) \cup S(x)| < 2k+2$. 
\end{corollary}

Observation \ref{bas} implies that all surrogates of a clean node are 4-friends of this node.
The surrogates of a dirty non-leech node are 10-friends of it, but the surrogates of a (dirty) leech are not necessarily 10-friends of it.
Since a leech is a 24-friend of its host, it follows that the surrogates of a leech are 34-friends of it.
We summarize these observations in the following statement.
\begin{observation} \label{farsur}
All surrogates  of either a clean or dirty node are 34-friends of it.
\end{observation}

Recall that our construction makes a persistent effort to re-use old surrogates. In particular, if an $i$-level node $x$ has a dirty non-leech  child $y$,
such that the term of the $k+1$ surrogates of $S(y)$ is not over at level $i-1$, then we set $S(x) = S(y)$.
(If $x$ has other such children, then it may choose to use surrogates from another child.)
In this case $x$ will be a dirty non-leech node as well, and we say that $x$ and $y$ are \emph{copies}.
Note that $y$ itself may be a copy of one of its dirty non-leech children, and so forth. In general, there is a path of dirty non-leech copies in the tree
which leads from a dirty non-leech node $x$ down to its \emph{appointing copy} $a(x)$, which is a dirty large node which appoints $k+1$ new surrogates.
By construction, all these non-leech copies are dirty and complete; 
in fact, we will show in Claim \ref{auxlem} that they are large.

Recall that surrogates of large nodes are appointed for a term of at least $\tau + 3$ levels. This does not mean that a specific surrogate
will necessarily be re-used for $\Omega(\tau)$ levels due to an abundance of surrogates; see the second remark at the end of Section 3.1.4.
This only means that no $j$-level ancestor of an $i$-level appointing node, for any $j \le i + \tau + 2$,
 will appoint to itself new surrogates.  % in the next $\Omega(\tau)$ (more specifically, $2 \tau + 1$) levels.
We define $term(x)$ as the term between the level of the appointing copy $a(x)$ of $x$ and the 
highest level of any copy of $x$ (which is either $x$ or an ancestor of $x$). By construction, $term(x)$ lasts at least $\tau+3$ levels.

The ``disjointness'' properties of Claims \ref{fofx} and \ref{disj} will be heavily used in the sequel.
\begin{claim} [Surrogates and Old Friends are Disjoint] \label{fofx}
%For a pair $(p,i),(q,j)$ of non-leech nodes,  $(S^*(p,i) \cup U^*(p,i)) \cap  (S^*(q,j) \cup U^*(q,j)) = \emptyset$.
For any pair $x,y$ of (distinct) dirty non-leeches at the same level, $S(x) \cap S(y) = F(a(x)) \cap F(a(y))= \emptyset$.
More generally, let $F^*(a(x))$ and $F^*(a(y))$ denote the set of all 10-friends of $a(x)$ and $a(y)$, respectively.
Then $F^*(a(x)) \cap F^*(a(y)) = \emptyset$.
%Also, 
% $(S(a(x)) \cup F(a(x))) \cap  (S(a(y)) \cup F(a(y))) = \emptyset$.
\end{claim}
\begin{proof}
%Let $(\tilde p,\tilde i)$ (respectively, $(\tilde q,\tilde j)$) be the appointing copy of $(p,i)$ (resp., $(q,j)$),
%and suppose without loss of generality that $(\tilde p,\tilde i)$
Write $a(x) = (p,i)$ and $a(y) = (q,j)$,  %) be the appointing copy of $x$ (resp., $y$),
and suppose without loss of generality that the appointment of $S(a(x)) = S(x)$ takes place before the appointment of $S(a(y)) = S(y)$.
It must hold that $i \le j$. Let $x'$ be the $j$-level copy of $x$, with $S(x') = S(x)$. By construction, $x'$ is a dirty non-leech.

Suppose for contradiction that there is a point $s$ in $F^*(a(x)) \cap F^*(a(y))$.
%$(S(x) \cap S(y)) \cup (F(a(x)) \cap F(a(y)))$.
%$(S^*(x) \cup F^*(x)) \cap  (S^*(y) \cup F^*(y))$.
Hence $s$ is a 10-friend of both $a(x)$ and $a(y)$.
By Claim \ref{anctrivial}, $s$ is also a 10-friend of $x'$, and so
$x'$ and $a(y)$ are 20-friends. It follows that $a(y)$ would become a leech of $x'$,
%. Thus $a(y)$ cannot be a host,
and would not be an appointing node, a contradiction.
%and suppose without loss of generality that $(\tilde p,\tilde i)$
\QED
\end{proof}

%We use the following claim to prove Corollary \ref{disj}.
%\begin{claim} \label{auxil}
%Let $x$ be a dirty $i$-level node. For any $i$-level node $y$,
%either $S(x) = S(y)$ or $S(x) \cap S(y) = \emptyset$ must hold. 
%In the former case $S(x) = S(y)$, $y$ must be dirty. 
%Moreover, either one among $x$ and $y$ is a leech of the other
%or both of them are leeches of a third node $z \ne x,y$.
%\end{claim}
%\begin{proof}
%\QED
%\end{proof}

\begin{claim} [Dirty Surrogate Sets are Either Equal or Disjoint] \label{disj}
Let $x$ be a dirty $i$-level node. For each level $j \in [i,\ell]$ and any $j$-level node $y$,
either $S(x) = S(y)$ or $S(x) \cap S(y) = \emptyset$ must hold. 
In the former case $S(x) = S(y)$, $y$ must be dirty. 
%Moreover, let $\tilde x$ be the $j$-level ancestor of $x$.
%Then if $S(x) = S(y)$,  either one among $\tilde x$ and $y$ is a leech of the other
%or both of them are leeches of a third $j$-level node $z \ne x,y$.
\end{claim}
\begin{proof}
Let $y$ be an arbitrary $j$-level node.
If $S(x) \cap S(y) = \emptyset$, then we are done. 

We henceforth assume that $S(x) \cap S(y) \ne \emptyset$. We will show that $S(x) = S(y)$, and that $y$ is dirty.

Define $v$ as $x$ if it is a non-leech, and as the host of $x$ otherwise. We have $S(x) = S(v)$.

We argue that $term(v)$ cannot be over until level $j-1$. 
Indeed, otherwise no point of $S(x)$ will be a surrogate of any $j$-level node, including $y$. 
It follows that $S(x) \cap S(y) = \emptyset$, a contradiction.

We henceforth assume that $term(v)$ is not over before level $j$. 
Consider the $j$-level copy $v'$ of $v$, where $S(v') = S(v) = S(x)$. Observe that $v'$ is a dirty non-leech.
%Claim ?? implies that
%either $S(v) = S(y)$ or $S(v) \cap S(y) = \emptyset$ must hold, and in the former case $S(v) = S(y)$, $y$ is dirty.
%Let $y$ be an arbitrary $j$-level node that satisfies $S(v') \cap S(y) \ne \emptyset$, and 
Note also that $S(v') \cap S(y) \ne \emptyset$, and let $p$ be a point in $S(v') \cap S(y)$. 
Since $v'$ is dirty, $p$ must be dirty too. Claim \ref{friends} implies that $y$ is dirty as well.

If $y = v'$, then $S(x) = S(y)$ must hold, and we are done.
%Moreover, in this case either one among $\tilde x
% and $y$ is a leech of either $x$ if $x = v'$ or of $x$'s host.
%In this case we are done.
We henceforth assume that $y \ne v'$. In this case $y$ must be a leech of $v'$. 
Indeed, otherwise we have $S(v') \cap S(y) = \emptyset$ by Claim \ref{fofx}, which is a contradiction.
Since $y$ is a leech of $v'$, we have  $S(y) = S(v') = S(x)$.
%By construction, If $y$ is a leech of $v'$, then we are done. 
%We henceforth assume that $y$ is a leech of some node $z \ne y$, and note that $S(y) = S(z)$.
%Also, define $\tilde y$ as $y$ if it is a host, and as the host of $y$ otherwise. We have $S(\tilde y) = S(y)$.
%We argue that $z = v'$.
%Indeed, otherwise we have $S(v') \cap S(z) = S(v') \cap S(y) = \emptyset$ by 
%Claim \ref{fofx}, yielding a contradiction.
%However, this means that $x'$ and $y$ are leeches of a third node $z = \tilde y \ne x,y$.
%Since $S(v') \cap S(y) \ne \emptyset$ and $S(y) = S(z)$, it 
The corollary follows. \QED
\end{proof}

%The following claim expands on Claim \ref{auxlem}.
%The following corollary follows from Claim \ref{fofx} and the construction.
\begin{claim} [Clean 10-Friends of Non-Leeches Remain Clean] \label{rmclean}
Let $x$ be a dirty non-leech $i$-level node. 
%If $x = a(x)$ is an appointing node, 
Then all points in $F(x) \setminus S(x)$ remain clean during the entire $term(x)$.
%at level $i$.
%If $x \ne a(x)$, then all points in $F(x) = F(x) \setminus S(x)$ remain clean during $term(x)$.
%at level $i$.
More generally, all the 10-friends of $x$ that are clean at the beginning of level $i$ will remain clean
%at the end of level $i$, 
during the entire $term(x)$, except for the surrogates $S(x)$ of $x$ which get dirty in the case when $x = a(x)$ is an appointing node.
\end{claim}
\begin{proof}
Let $p$ be an arbitrary 10-friend of $x$ that is clean at the beginning of level $i$.
Suppose for contradiction that $p$ became dirty at level $i$, and $p \nin S(x)$. This means that some other $i$-level appointing node $y$, $y \ne x$,
appointed $p$ as one of its $k+1$ new surrogates at level $i$. Since $p$ is a 10-friend of both $x$ and $y$, it follows that $x$ and $y$ are 20-friends.
If $x$ is either a non-appointing node (i.e., $x \ne a(x)$) or it is an appointing node which appoints its surrogates before $y$ does,
then $y$ would become $x$'s leech, and would not become an appointing node, a contradiction. Otherwise, $x$ is a an appointing node
which appoints its surrogates after $y$ does, but then $x$ would become $y$'s leech, yielding the same contradiction.

We have shown that all 10-friends of $x$ (except for those in $S(x)$) that are clean at the beginning of level $i$  will remain clean at the end of level $i$,
and so they will be clean at the beginning of level $i+1$.
By Claim \ref{friendship}, all these points are 10-friends of $\pi(x)$, which is also a non-leech. 
Thus we can apply the same argument for level $i+1$, and carry on in this way also for subsequent levels. 
%By applying this argument inductively, 
% and continue this way until the last level of $term(x)$. % is reached.
Consequently, we get that all 10-friends of $x$ that are clean at the beginning of level $i$ will remain clean during the entire $term(x)$,
except for the surrogates $S(x)$ of $x$ which get dirty in the case $x = a(x)$.
\QED
\end{proof}

\begin{claim} [If the Appointing Copy is Large, All Other Copies are Large Too] \label{auxlem}
Let $x$ be a dirty non-leech $i$-level node. Then there is a path of dirty non-leech copies of $x$ in the tree which leads down from $x$
to its appointing copy $a(x)$. Also, all points of $F(a(x)) \setminus S(x)$ belong to $F(x)$, unless $|F(x)| = 3k+3$.
In particular, if $a(x)$ is large, then $x$ will be large too. 
%all the  10-friends of $a(x)$ that are clean after the appointment of  $S(x) = S(a(x))$
%are clean  at level $i$. %In particular, all the points in $F(a(x))$ are clean at level $i$.
% Thus $x$ is large.
\end{claim}
\begin{proof}
First, the fact that there is a path of dirty non-leech copies of $x$ in the tree which leads down from $x$
to its appointing copy $a(x)$ follows easily from the construction.

By Claim \ref{rmclean}, all 10-friends of $a(x)$ that are clean after the appointment of $S(x) = S(a(x))$ (in particular, all points in $F(a(x)) \setminus S(a(x))$) remain clean during the entire $term(x)$. Hence they will be clean at level $i$.
%Assuming $a(x)$ is large, 
By construction, all the points in $F(a(x)) \setminus S(x)$ will belong to $F(x)$, unless $|F(x)| = 3k+3$.

If $a(x)$ is large,
then we have $|F(a(x)) \setminus S(x)| \ge k+1$. 
Consequently, all the at least $k+1$ points in $F(a(x)) \setminus S(x)$ will belong to $F(x)$, unless $|F(x)| = 3k+3$.
It follows that $|F(x) \cup S(x)| \ge 2k+2$, and so $x$ is large.
\QED
\ignore{We prove the claim for the parent $\pi(a(x))$ of $a(x)$. (This argument can be applied inductively.)

By Observation \ref{ancestor} and the construction, $\pi(a(x))$ is a dirty non-leech.
%The following lemma shows that all these points remain clean during the entire $term(x)$.
Since $a(x)$ is large, there are at least $k+1$ 
clean 10-friends of $a(x)$ in $F(a(x)) \setminus S(a(x))$.   % after the appointment of $S(a(x))$.
Moreover, by Claim \ref{anctrivial}, all these points are 10-friends of $\pi(a(x))$.
%We next show that these points 
%[[S: need to change this accordingly to above]]
%Hence $\pi(a(x))$ will be large by definition.
%Finally, the following claim shows that these points will remain clean 10-friends of any copy of $x$,
%which implies that all the copies of $x$ must be large. 
%After the appointment of $S(a(x))$ there are at least $k+1$ clean 10-friends of $a(x)$.
%The same argument can be repeated iteratively.
Denote the level of $a(x)$ by $l$, 
and suppose for contradiction that some clean 10-friend $p$ of $a(x)$ becomes dirty after the appointment of $S(a(x))$
at level $l'$, where $l'$ is either $l$ or  $l+1$.  % denote that level by $\tilde l$.
Define $x'$ as $a(x)$ if $l' = l$, or as $\pi(a(x))$ if $l' = l+1$.
This means that there is an $l'$-level appointing node $y$, $y \ne x'$, which appoints $p$ as one of its $k+1$ new surrogates at level $l'$.
%note that $y \ne a(x)$ if $l' = l$ and $y\ne \pi(a(x))$ if $l' = l+1$.
%Note that the $j$-level copy of $x$ is a complete non-leech node, and it is a 20-friend of $y$.
Observe that $p$ is a 10-friend of $x'$ and $y$, which implies that
$y$ is a 20-friend of $x'$. However, $x'$ is a dirty non-leech.
By construction, $y$ would become a leech of $x'$,  and would not be an appointing node,  a contradiction. 
}
%We can apply the same argument for $\pi(\pi(a(x)))$, and repeat iteratively until $x$ is reached.
%By Claim \ref{anctrivial}, all the 10-friends of $a(x)$ are also 10-friends of $x$, which implies that $x$ is large.
\end{proof}

%[[S: need to bound $|F(\cdot)|$ by $3k+3$, not $2k+2$. Need to prove that if $|F(\cdot)| < 3k+3$, then $F(\cdot)$
%contains all clean 10-friends of that node]]

\vspace{0.07in}
\noindent
{\bf 3.2.2 ~A Key Lemma and Its Corollary.~}
The following lemma is central in our analysis.
In particular, it shows  that whenever the need to appoint new surrogates for some node $x$ arises, we have $|F(x)| \ge 2k+2$.
(See Corollary \ref{defined}.)
Thus Procedure $ComputeSets_{(i)}$ from Section 3.1.4 is well-defined. 
%that node has at least $2k+2$ clean 10-friends.
%[[S: perhaps add to this statement the issue of leeches. Prove them together, this will save anyway the proof of the subsequent corollary]]
%that can be appointed. 
%\begin{lemma} [Proof in Appendix C] \label{key}
%Any non-leech large node has at least $k+1$ clean 10-friends. 
%\end{lemma}
%Recall that at the first time in which a point is appointed as a new surrogate, that point remains clean--and it will remain so during the entire term.
%It becomes dirty at the second time in which it is appointed as a new surrogate.
%Let us call a clean point \emph{semi-clean} if it has been appointed as a new surrogate once. 
%A dirty point will be called \emph{fully-dirty}.
%We say that $x$ is a \emph{first-term node} if its appointing copy is atomically large?]]
%[[S: need to explain that just before a term starts, we have at least $k+1$ points that are very close to $p(x)$,
%specifically, the distance is at most $4 * ra(x) + 10 * ra(x-2tau) \le 5 * ra(x)$.]]
%[[S: need to define $\hat F(x)$; remove $l(x)$]]
%Lemma \ref{key} follows from the next lemma.
%[[S: the main idea is to look for a compound node; if one of the surrogates is semi-clean, then a close compound node must exist!!!]]
\begin{lemma} [All Non-Leeches are Large] \label{key}
Let $x$ be a dirty non-leech $i$-level node.  %, with an appointing copy $a(x)$. 
Then: \\(a) The appointing copy $a(x)$ of $x$ is large, thus $x$ is large too by Claim \ref{auxlem}.
(It is possible that $x = a(x)$.)
\\(b) If $i$ is the last level of $term(x)$, then $|F(x)| \ge 2k+2$. (Note that $F(x) = F(x) \setminus S(x)$ in this case.)
 Moreover, if $|F(x)| < 3k+3$,
then at least $k+1$ points from $F(x)$ do not belong to $F(a(x))$. 
\end{lemma}
%Next, we turn to proving Lemma \ref{key}.
%\vspace{0.06in}
%\\{\bf Proof of Lemma \ref{key}:}~
\begin{proof}
The proof of the two assertions of the lemma is by induction on $i$.

A dirty node that has only clean children is called \emph{atomically-dirty}; otherwise it is     \emph{compound-dirty}.
For the basis of the induction we may consider nodes whose appointing copy is atomically-dirty. Such nodes must be large by the construction.
(An atomically-dirty non-leech  $y$ may get dirty only if $|F(y)| \ge 2k+2$,
in which case $y$ is large by definition.)
The first assertion follows immediately. 
\\
To prove the second assertion we use an argument that works for both the basis of the induction and the induction step.
We omit this argument here for conciseness,
but provide it in the induction step.
\\\emph{Induction Step: Assume that the  lemma holds for all smaller values of $i, i \ge 1$, and prove it for $i$.}
%[[S: need to explain that all ancestors of a dirty node (up to level $i$) must be dirty.

We start with the first assertion. We have shown that the case when $a(x)$ is atomically-dirty is trivial.

We henceforth assume that $a(x)$ is compound-dirty. By definition, $a(x)$ has at least one dirty child.
Let $i_a$ be the level of $a(x)$,  with $i_a \le i$.
Next, we argue that every dirty child of $a(x)$ is either a non-leech   whose term is over at level $i_a-1$, or a leech of some 
host whose term is over at level $i_a -1$.
Indeed, if $a(x)$ had a dirty non-leech child $y$ whose term is not over at level $i_a-1$, then Procedure $ComputeSets_{(i_a)}$ would assign $S(a(x)) = S(y)$,
and $a(x)$ would not be an appointing copy,  a contradiction.
Similarly, if $a(x)$ had a dirty child which is a leech of some host $y$ whose term is not over at level $i_a-1$, then $y$'s parent would be a non-leech by the construction.
However, it is easy to see that $a(x)$ would be close enough to $y$'s parent to become its leech by Claim \ref{trait}, and would not be an appointing copy, a contradiction.

Consider such a node $y$ of level $i_a -1 <  i$, which is either a dirty non-leech  child of $a(x)$
or a (dirty) host of a leech child of $a(x)$. We have shown that $i_a-1$ must be the last level of $term(y)$.  %, and in particular at the third phase of $term(z)$.
By the induction hypothesis for $y$,
$|F(y)| \ge 2k+2$. 
By construction, all points of $F(y)$ are clean at the beginning of level $i_a -1$.
By Claim \ref{rmclean}, all points of $F(y)$ will remain clean at level $i_a-1$, and thus will be clean at the beginning of level $i_a$.
%[[S: Need to explain that they will remain clean at level $i_a-1$, and that's enough I think for concluding the $a(x)$ will be large]]
If $y$ is a child of $a(x)$, then for any point $s$ of $F(y)$,  
$\delta(a(x),s) ~\le~ \delta(a(x),y) + \delta(y,s) ~\le~ 3 \cdot 5^{i_a} + 10 \cdot 5^{i_a -1} ~\le~ 5 \cdot 5^{i_a}.$
Otherwise, $y$ is a host of a child $c$ of $a(x)$, in which case we have $$\delta(a(x),s) ~\le~ \delta(a(x),c) + \delta(c,y) +  \delta(y,s)
~\le~  3 \cdot 5^{i_a} + 24 \cdot 5^{i_a -1} + 10 \cdot 5^{i_a -1} ~\le~ 10 \cdot 5^{i_a}.$$
In both cases $a(x)$ has at least $2k+2$ clean 10-friends at the beginning of level $i_a$, namely, the points of $F(y)$. By construction, all these points will
belong to $F(a(x))$, unless $|F(a(x))| = 3k+3$.
In either case we have $|F(a(x))| \ge 2k+2$,
implying that $a(x)$ is large. The first assertion follows.

Next, we prove the second assertion.
We henceforth assume that $i$ is the last level of $term(x)$.
If $|F(x)| = 3k+3$, then we are done. We henceforth assume that $|F(x)| < 3k+3$.
The first assertion  implies that $a(x)$ is large, and so
there must be at least $k+1$ clean 10-friends of $a(x)$ in $F(a(x)) \setminus S(a(x))$.
%(after the appointment of $S(a(x))$ at level $i_a$).
By Claim \ref{auxlem}, all these points remain clean during the entire $term(x)$, and they will belong to $F(x)$.
%Moreover, by construction, all these points will belong to $F(x)$.

Consider the last level of the first phase in $term(x)$, denoted $l$,
and let $q_1,\ldots,q_{k+1}$ be $k+1$ arbitrary points (among a total of at most $k + \xi^2 \cdot (2k+1)$ points) 
which increased the degree of points in $S(a(x)) = S(x)$ since the beginning of $term(x)$ (due to non-redundant cross edges). 
For each point $q_j$, let $l_j$ be the first level in $term(x)$, such that the basic spanner $H$ contains a cross edge between 
either the 
$l_j$-level copy $x_j$ of $x$ or one of its leeches and some $l_j$-level node $y_j$ for which both $q_j \in S(y_j)$ and $S(y_j) \ne S(x)$ hold; observe that $l_j \le l$.
(For such a \emph{non-redundant} edge, our FT spanner  $\cH$ contains a bipartite clique between $S(x)$ and $S(y_j)$.)
By Claim \ref{disj}, $S(y_j) \cap S(x) = \emptyset$. Moreover, we argue that $q_j$ cannot be a 10-friend of $a(x)$, which implies that it does not belong to $F(a(x))$.
If $y_j$ is a dirty non-leech, this assertion follows immediately  from Claim \ref{fofx}. 
In the case that $y_j$ is a dirty leech, we can apply Claim \ref{fofx} 
with the host of $y_j$ instead of $y_j$ itself, and get the same result.
We henceforth assume that $y_j$ is clean. Since $q _j \in S(y_j)$, Observation \ref{bas} implies that $q_j$ is a 4-friend of $y_j$.
If $q_j$ were a 10-friend of $a(x)$, then it must also be a 10-friend of the $l_j$-level copy $x_j$ of $x$.
Consequently, $y_j$ would be a 14-friend of $x_j$, and would thus become its leech and   dirty, a contradiction.

The second phase of $term(x)$ starts at level $l+1$, and ends at level $i = l+\tau + 2$. 
Suppose first that all $k+1$ points $q_1,\ldots,q_{k+1}$ are clean at the beginning of level $i -2 = l+\tau$.
In this case the $l_j$-level node $y_j$ for which $q_j \in S(y_j)$ must be clean by Corollary \ref{preanc} and the construction, for each $j \in [k+1]$.
%Observation \ref{bas} implies that $y_j$ is the $l_j$-level base bag of $q_j$ [[S: needed?]].
Moreover, by Claim \ref{friendship}, all the $k+1$ points $q_1,\ldots,q_{k+1}$ are 10-friends of the $(i-2)$-level copy of $x$.
Claim \ref{anctrivial} implies that they will also be 10-friends of all the ancestors of that node.
Moreover, all these points will remain clean until level $i$ by Claim \ref{rmclean}. 
%Indeed, otherwise some $j$-level node, for $i-2 \le j \le i$, appoints one of these points (or more) as its new surrogate; this node will a 20-friend of the %dirty $j$-level copy of $a(x)$, and would thus become its leech, a contradiction.
Finally, recall that all $k+1$ points $q_1,\ldots,q_{k+1}$ are not 10-friends of $a(x)$, and thus they do not belong $F(a(x))$. 
On the other hand, we argue that these points will belong to $F(x)$ by the construction. Indeed, 
since there is a cross edge between $y_j$ and either the $l_j$-level copy $x_j$ of $x$ or one of its leeches, for each $j \in [k+1]$,
it follows that $q_j$ will belong to the reserve set of either $x_j$ or one of its leeches. Consequently, $q_j$ will also belong
to the reserve set of the $(i-2)$-level copy $x_{i-2}$ of $x$. (The latter assertion is immediate if the cross edge is between $y_j$ and $x_j$. 
In the complementary case where the cross edge
is between $y_j$ and some leech of $x_j$,  $q_j$ will belong to the reserve set of the $(i-2)$-level ancestor of that leech, denoted $w_{i-2}$.
If $w_{i-2} = x_{i-2}$, then we are done. 
Otherwise, there is a cross edge between $w_{i-2}$ and $x_{i-2}$. 
Claim \ref{friendship} implies that $q_j$ is a 10-friend of $w_{i-2}$,
hence our construction guarantees that $q_j$ will be added to the reserve set of $x_{i-2}$.)
Since $q_j$ is a 10-friend of $x$ and $|F(x)| < 3k+3$,
it follows that $q_j$ will belong to $F(x)$.
%it follows that there is also a cross edge between the parent $\pi(y_j)$ of $y_j$ and the $(l_j+1)$-level copy $\pi(x_j)$ of $x$ (assuming $\gamma \ge 6$). 
%Hence all points of $F(\pi(y_j))$ will be added to the reserve set of $\pi(x_j)$ (except for those that were already there).
%Since $|F(x)| < 3k+3$, it follows that
%Note also that $\pi(y_j)$ is the $(l_j+1)$-level base bag of $q_j$ [[S: needed?]]. 
%Since $q_j$ is clean at the beginning of level $i-2$,
%  it will belong to $F(\pi(y_j))$ by the construction, unless $|F(\pi(y_j))| = 3k+3$.
%It cannot hold that $|F(\pi(y_j))| = 3k+3$, as then we would have $|F(x)| = 3k+3$. [[S: this requires proof]]
%If $|F(\pi(y_j))| = 3k+3$
%Corollary \ref{bas2} and the construction imply that $q_j \in S(\pi(y_j)) \subseteq F(\pi(y_j))$ (otherwise $q_j$ will not be clean at level $i$).
%Hence each such point $q_j$ will be added to the reserve set of $\pi(x_j)$ (unless it was already there). 
%Moreover, each such point $q_j$ will be added to the friend sets $F(x')$ of all the copies $x'$ of $x$ that are 10-friends of $q_j$.
%(in particular to $F(x)$), except for the copies which already contain $3k+3$ points.
%However, if there is a copy $x'$ of $x$ which contains $3k+3$ points, then $x$ will contain $3k+3$ points as well by Claim \ref{rmclean}, which contradicts the %above assumption. 
We conclude that all points $q_1,\ldots,q_{k+1}$    belong to $F(x)$,
%for which $|F(x')| = 3k+3$.
which completes the proof of the second assertion in this case.
%It follows that $x$ has at least $2k+2$ clean 10-friends, at least $k+1$ of which
%do not belong to $F(a(x))$. Claim \ref{fr} implies that $|F(x)| \ge 2k+2$. Moreover, if $|F(x)| < 3k+3$, then $F(x)$ contains all
%the 10-friends of $x$ that are clean at the beginning of level $i$. Hence, in this case at least $k+1$ of the points in $F(x)$ do not belong to $F(a(x))$, as %required.

We henceforth assume that some point $q_j$ became dirty until level $i-3$.
Recall that $l_j$ is the first level in $term(x)$,
such that $H$ contains a non-redundant edge between 
either the $l_j$-level copy $x_j$ of $x$ or one of its leeches
%the $l_j$-level copy of $x$ 
and some node $y_j$ for which both $q_j \in S(y_j)$
and $S(y_j) \cap S(x) = \emptyset$ hold. 
%Note that the $(i-2)$-level ancestor of $y_j$ is a 10-friend of $q_j$. Since $q_j$ is dirty at level $i-2$, this means that
%this ancestor of $y_j$ is dirty as well.
Let $l'_j$ be the first level after $l_j$ (i.e., $l'_j > l_j$) in which the $l'_j$-level copy $x'_j$ of $x$ and the $l'_j$-level ancestor $y'_j$
of $y_j$ are 24-friends. We   argue that $l'_j \le i-2$.
Indeed, as mentioned above, Claim \ref{friendship} implies that $q_j$ is a 10-friend of the $(i-2)$-level copy of $x$.
Also, by Observation \ref{farsur}, the fact that $q_j$ is in $S(y_j)$ implies that it is a 34-friend of $y_j$. By Claim \ref{anctrivial}, $q_j$ is a 10-friend of any ancestor of $y_j$.
It follows that the $(i-2)$-level copy of $x$ and the $(i-2)$-level ancestor of $y_j$ are 20-friends, and so $l'_j \le i-2$.

Consider now level $\tilde l_j = l'_j - 1$, where $l_j \le \tilde l_j \le i-3$, and let $\tilde x_j$ and $\tilde y_j$ be the $\tilde l_j$-level ancestors of $x_j$ and $y_j$, respectively.
We argue that $S(\tilde y_j) \cap S(x) = \emptyset$. Indeed, this assertion clearly holds if $\tilde l_j = l_j$. %recall that $S(y_j) \ne S(x)$
Consider the complementary case $\tilde l_j > l_j$, and suppose for contradiction that $S(\tilde y_j) \cap S(x) \ne \emptyset$.
Since $S(x) = S(\tilde x_j)$, Claim \ref{disj} implies that $S(\tilde y_j) =  S(\tilde x_j)$ must hold. 
However, in this case $\tilde y_j$ must be a leech of $\tilde x_j$ and thus a 24-friend of it, which stands
in contradiction to the above definitions. 
%[[S: there is a delicate issue here (need to explain that $\tilde y_j$ is dirty, otherwise
%we need to apply another claim; but where do we use this anyway?]]
%This contradicts the definition of $l'_j$ and $\tilde l_j$.
%To see this, first  (this covers the case $\tilde l_j = l_j$). 

The analysis splits into two main cases.
\\\emph{Case 1: $\tilde y_j$ is a clean node.}
We argue that $q_j$ must become dirty at level $l'_j = \tilde l_j + 1$ (not before and not after).
%Suppose for contradiction otherwise.
To see this, first note that $q_j$ cannot become dirty at any level between $l_j$ and $\tilde l_j$. Indeed, otherwise the appointing node that made $q_j$ dirty
would be a 20-friend of the corresponding ancestor of $y_j$, which would become its leech and thus dirty. 
However, then $\tilde y_j$ would be dirty too by Observation \ref{ancestor}, a contradiction. On the other 
hand, if $q_j$ does not get dirty at level $l'_j$ or before,
then it must get dirty afterwards (because it gets dirty until level $i-3$). Suppose for contradiction that some node $z$
at level $h$ appoints $q_j$ as one of its $k+1$ new surrogates, where $l'_j + 1 \le h \le i-3$. 
This means that $\delta(z,q_j) \le 10 \cdot 5^{h}$.
Recall that $q_j$ is a 10-friend of any ancestor of $y_j$, and in particular of $y'_j$.
Also, $y'_j$ is a 24-friend of $x'_j$, and so 
$\delta(x'_j,q_j) ~\le~ \delta(x'_j,y'_j) + \delta(y'_j, q_j) ~\le~ 24 \cdot 5^{l'_j} + 10 \cdot 5^{l'_j} ~<~ 7 \cdot 5^{l'_j+1}
~\le~ 7 \cdot 5^h.$
(The last inequality holds since $h \ge l'_j + 1$.)
Let $x_h$ be the $h$-level copy of $x$. Observe that $\delta(x_h,q_j) \le \delta(x_h,x'_j) + \delta(x'_j,q_j) \le 4 \cdot 5^h + 7 \cdot 5^h = 11 \cdot 5^h$.
Consequently, $\delta(x_h,z) \le \delta(x_h,q_j) + \delta(q_j,z) \le 11 \cdot 5^h + 10 \cdot 5^h \le 21 \cdot 5^h$. Thus $x_h$ and $z$ are 21-friends,
and so $z$ would have to become a leech of $x_h$, and would not be an appointing node, a contradiction. 
It follows that $q_j$ becomes dirty at level $l'_j = \tilde l_j + 1$, as was argued above.

Denote  by $u$ the $l'_j$-level appointing node that appoints $q_j$ as one of its $k+1$ new surrogates at level $l'_j$.
By construction, $u$ is a dirty non-leech. By the induction hypothesis, $u$ must be large.
Thus there are at least $k+1$ clean 10-friends of $u$ in $F(u)\setminus S(u)$.
Claim \ref{fofx} implies that $F(u) \cap F(a(x)) = \emptyset$. By Claim \ref{auxlem}, all the at least $k+1$ points in $F(u) \setminus S(u)$
remain clean during the entire $term(u)$. Moreover, recall that $y'_j$ is a 24-friend of $x'_j$ and a 10-friend of $q_j$,
and so $x'_j$ and $q_j$ are 34-friends.
Since $u$ is a 10-friend of $q_j$, we get that $x'_j$ and $u$ are 44-friends.
Hence there is a cross edge between $x'_j$ and $u$ (assuming $\gamma > 44$).
Since $u$ is dirty, it follows that the second phase of $term(x)$ will start at level $l'_j+1$ (or before), which implies that $term(u)$ will last at least
until the last level $i$ of $term(x)$.
In particular, all the at least $k+1$ points in $F(u) \setminus S(u)$ will remain clean until level $i$. 
Since $x'_j$ and $u$ are 44-friends,
all points of $F(u) \setminus S(u)$ are 54-friends of $x'_j$. All these points are also 10-friends of the $(l'_j+2)$-level copy of $x$, and let alone of $x$ itself.
By construction, since  there is a cross edge between $x'_j$ and $u$ (and as $|F(x)| < 3k+3$),
all points of $F(u) \setminus S(u)$ will belong to $F(x)$.  % unless $|F(x)| = 3k+3$.
% It follows that $x$ has at least $2k+2$ clean 10-friends, at least $k+1$ of which
%do not belong to $F(a(x))$. 
%[[S: need to explain that . This is why
%$F(x)| \ge 2k+2$, and ..]]
%Claim ?? implies that $|F(x)| \ge 2k+2$. Moreover, if $|F(x)| < 3k+3$, then $F(x)$ contains all
%the 10-friends of $x$ that are clean at the beginning of level $i$. 
%It follows that $\delta(z,x) \le \delta(z,q
Hence,  at least $k+1$ points from $F(x)$ do not belong to $F(a(x))$, and we have $|F(x)| \ge |F(a(x)) \setminus S(a(x))| + |F(u) \setminus S(u)| \ge 2k+2$,
which provides the required result.
\\\emph{Case 2: $\tilde y_j$ is a dirty node.} 
It is possible that $\tilde y_j$ is a leech; in this case we consider the host of $\tilde y_j$.
Define $\tilde u_j$ as $\tilde y_j$ if it is a non-leech, and as the host of $\tilde y_j$ otherwise.
By the induction hypothesis, $\tilde u_j$ is large, and so there are at least $k+1$ clean 10-friends in $F(a(\tilde u_j)) \setminus S(\tilde u_j)$.
Moreover, all these points will also belong to $F(\tilde u_j)$ by Claim \ref{auxlem}, unless $|F(\tilde u_j)| = 3k+3$.
Recall that $x'_j$ and $y'_j$ are 24-friends. Hence 
$$\delta(\tilde x_j,\tilde y_j) ~\le~ \delta(\tilde x_j,x'_j) + \delta(x'_j,y'_j) + \delta(y'_j,\tilde y_j) ~\le~
3 \cdot 5^{\tilde l_j + 1} + 24 \cdot 5^{\tilde l_j + 1} + 3 \cdot 5^{\tilde l_j + 1} ~=~ 30 \cdot 5^{\tilde l_j + 1} ~=~ 150 \cdot 5^{\tilde l_j}.$$
Note also that $\delta(\tilde y_j,\tilde u_j) \le 24 \cdot 5^{\tilde l_j}$, and so  $\tilde x_j$ and $\tilde u_j$ are 174-friends. 
By construction, there is a cross edge between $\tilde x_j$ and $\tilde u_j$ (assuming $\gamma \ge 174$),
and so all points of $F(\tilde u_j)$ will belong to $F(x)$ by the construction.
Since $|F(x)| < 3k+3$, it must hold that $|F(\tilde u_j)| < 3k+3$. It follows that all the at least $k+1$ points of 
$F(a(\tilde u_j)) \setminus S(\tilde u_j)$ belong to $F(\tilde u_j)$ (and thus also to $F(x)$).
On the other hand, recall that $S(\tilde y_j) \cap S(x) = \emptyset$. Since $S(\tilde u_j) = S(\tilde y_j)$ and $S(\tilde x_j) = S(x)$,
we have $\tilde u_j \ne \tilde x_j$.
Claim \ref{fofx} implies that $F(a(x)) \cap F(a(\tilde u_j)) = \emptyset$,
and so all the at least $k+1$ points in $F(a(\tilde u_j)) \setminus S(\tilde u_j)$ belong to $F(\tilde u_j)$ but do not belong to $F(a(x))$.
%Also, all these points belong to $F(\tilde u_j)$.

Suppose first that $\tilde l_j$ is not the last level of $term(\tilde u_j)$. In this case all the at least $k+1$ points in 
$F(a(\tilde u_j)) \setminus S(\tilde u_j)$ remain clean at level $\tilde l_j + 1$ by Claim \ref{auxlem}. However, it is not difficult to see 
that all these points are 10-friends of the
$(\tilde l_j+2)$-level copy of $x$. Hence, by Claim \ref{rmclean}, these points will remain clean during the entire $term(x)$.
As mentioned, %there is a cross edge between $\tilde x_j$ and $\tilde u_j$ (assuming $\gamma \ge 174$),
all points of $F(a(\tilde u_j)) \setminus S(\tilde u_j)$ belong to both $F(\tilde u_j)$ and $F(x)$, but not to $F(a(x))$. 
Consequently, at least $k+1$ points from $F(x)$ do not belong to $F(a(x))$, and we have $|F(x)| \ge |F(a(x)) \setminus S(a(x))| + 
|F(a(\tilde u_j)) \setminus S(\tilde u_j)| \ge 2k+2$. This completes the proof in this case.

% and we are done.

%[[S: should elaborate]]

In what follows we assume that $\tilde l_j$ is the last level of $term(\tilde u_j)$. 
Note that $F(\tilde u_j) = F(\tilde u_j) \setminus S(\tilde u_j)$. 
By the induction hypothesis, $|F(\tilde u_j)|  \ge 2k+2$.
By Claim \ref{rmclean}, all points of $F(\tilde u_j)$ remain clean at level $\tilde l_j$.
The analysis splits into two subcases.  %There are two cases.
\\\emph{Case 2.a: $F(a(x)) \cap F(\tilde u_j) = \emptyset$}. In this case some of the points in $F(\tilde u_j)$ may get dirty at level $\tilde l_j+1$,
which occurs if they are appointed as new surrogates of some $(\tilde l_j + 1)$-level node. 
We argue that no more than $k+1$ of these points may get dirty at  level
$\tilde l_j+1$. 
Suppose for contradiction otherwise.
Since any appointing node appoints exactly $k+1$ points as new surrogates, 
this means that at least two distinct $(\tilde l_j + 1)$-level appointing nodes $v_1$ and $v_2$ must appoint points from $F(\tilde u_j)$ as their new surrogates.
Since all points of $F(\tilde u_j)$ are within distance $10 \cdot 5^{\tilde l_j} = 2 \cdot 5^{\tilde l_j+1}$ from $\tilde u_j$,
it follows that they are all within distance $4 \cdot 5^{\tilde l_j+1}$ from each other. 
Since $v_1$ (respectively, $v_2$) appoints at least one point $p_1$ (resp., $p_2$) from $F(\tilde u_j)$ as a new surrogate, 
we have $\delta(v_1,p_1),\delta(v_2,p_2) \le 10 \cdot 5^{\tilde l_j+1}$. Hence $$\delta(v_1,v_2) ~\le~ \delta(v_1,p_1) + \delta(p_1,p_2) + \delta(p_2,v_2)
~\le~ 10 \cdot 5^{\tilde l_j+1} + 4 \cdot 5^{\tilde l_j+1} + 10 \cdot 5^{\tilde l_j+1} ~=~ 24 \cdot 5^{\tilde l_j+1}.$$
Hence $v_1$ and $v_2$ are 24-friends, and so one would have to become a leech of the other and would not be an appointing node, a contradiction.
Since $|F(\tilde u_j)| \ge 2k+2$,
it follows that at least $k+1$ points from $F(\tilde u_j)$ must be clean at the beginning of level $\tilde l_j + 2$. Moreover,
%as mentioned above,
it is not difficult to see that all these points are 10-friends of the $(\tilde l_j+2)$-level copy of $x$.
Hence, they will remain clean during the entire $term(x)$ by Claim \ref{rmclean}, and will belong to $F(x)$.  
We have proved that $F(x)$ contains at least $2k+2$ points, at least $k+1$ of which do not belong to $F(a(x))$,
which completes the proof in this case.
%which completes the proof in this case.
\\\emph{Case 2.b: $F(a(x)) \cap F(\tilde u_j)  \ne \emptyset$.} In this case $\tilde u_j$ must be a 20-friend of the $\tilde l_j$-level copy $\tilde x_j$ of $x$.
It is easy to see that all points of $F(\tilde u_j)$ will be 10-friends of the $(\tilde l_j+1)$-level copy of $x$,
and so they will remain clean during the entire $term(x)$ by Claim \ref{rmclean}. 
By construction, $F(x)$ will contain all points of $F(\tilde u_j)$.
%In particular, in this case $F(x)$ will contain all points of $
%If $|F(\tilde u_j)| = 3k+3$, then $x$ will have at least $3k+3$ clean 10-friends.
%In this case we have $|F(x)| = 3k+3$. Otherwise $|F(\tilde u_j)| < 3k+3$, which by Claim ?? means
%that $F(\tilde u_j)$ contains all the clean 10-friends of $\tilde u_j$. 
We have also shown that all the at least $k+1$ points of $F(a(\tilde u_j)) \setminus S(\tilde u_j)$ belong to $F(\tilde u_j)$. 
Finally, Claim \ref{fofx} implies that $F(a(x)) \cap F(a(\tilde u_j)) = \emptyset$, and we are done. Lemma \ref{key} follows.\QED
\end{proof}

%Before proving Lemma \ref{key}, we provide an important corollary of this lemma.
Lemma \ref{key} yields the following corollary, which, in turn, implies that our construction is well-defined.
\begin{corollary} [Procedure $ComputeSets_{(i)}$ is Well-Defined] \label{defined} 
For any appointing node $x$, we have $|F(x)| \ge 2k+2$. 
\end{corollary}
\begin{proof}
By construction, the statement clearly holds for both leaf nodes and nodes for which all children are clean.
We henceforth assume that $x$ is an $i$-level appointing node, for $i \ge 1$, with at least one dirty child $y$.
Define $z$ as $y$ if it is a non-leech, or as $y$'s host otherwise.
We argue that $term(z)$ must be over at level $i-1$. Indeed, otherwise $z$'s parent will be a dirty non-leech,
and so $x$ will not become an appointing node, a contradiction.
(This follows immediately from the construction if $y$ is a non-leech. If $y$ is a leech and $z$ is its host,
then $x$ would be a 24-friend of $z$'s parent by Claim \ref{trait}, and would become a leech of either $z$'s parent or some other $i$-level dirty non-leech 
24-friend of $x$, and not an appointing node.)
%(Refer to the proof of the first assertion of the induction step in Lemma \ref{key} for a rigorous argument.)

We have shown that $term(z)$ is over at level $i-1$. Hence, Lemma \ref{key} implies
that $|F(z)| \ge 2k+2$. Moreover, all points of $F(z)$ must remain clean at level $i-1$ by Claim \ref{rmclean}.
%(Indeed, a point $p$ of $F(z)$ that became dirty at level $i-1$ must have been appointed as a new surrogate of some $(i-1)$-level appointing node $w$; 
%but then
%$w$ would be a 20-friend of $z$, and would become its leech, a contradiction.)
It is easy to see that all points of $F(z)$ 
are 10-friends of $x$. Consequently,
all these points must belong to $F(x)$ by the construction, unless $|F(x)| = 3k+3$.  % and we are done.
%before the second part of Procedure $ComputeSets_{(i)}$ starts
%In either case we have $|F(x)| \ge 2k+2$,
%Moreover, no other $i$-level appointing node $u \ne x$ may appoint a point from $F(z)$ as its surrogate, as otherwise $u$ would become a leech of $x$
%or $x$ would become a leech of $u$, a contradiction.
The corollary follows.
%In other words, $x$ has at least $2k+2$
%clean 10-friends at the beginning of level $i$. By Claim ??, either all these points belong to $F(x)$ or $|F(x)| = 3k+3$ must hold.
%The corollary follows in either one of these cases.
\QED
\end{proof}

\ignore{
[[S: remove]]
The next observation follows from the construction. We will use it in conjunction with Lemma \ref{key} to prove Corollary \ref{semi}.
\begin{observation} \label{ob:semi}
Suppose that a  point $p$ becomes semi-clean at level $i$, and let $x_1,\ldots,x_t$ be
all the $i$-level nodes of which $p$ is a surrogate, where $t = O(1)^{O(d)}$.
If $p$ is a surrogate of some clean $j$-level node $x$, for any  $j > i$, 
then $x$ must be an ancestor of one of the $t$ nodes $x_1,\ldots,x_t$.
\end{observation}

[[S: remove]]
The next statement shows that once a point becomes semi-clean, it will almost immediately stop serving as a surrogate
of clean nodes. More specifically, if a point becomes semi-dirty at level $i$, it may serve as a surrogate of clean nodes at the following level $i+1$, but not at the subsequent levels $i+2,\ldots,\ell$.
% (it may serve as a surrogate of clean nodes 
\begin{corollary} \label{semi}
If a point $p$ becomes semi-clean at level $i$,
it will not serve as a surrogate of clean $j$-level nodes, for any $j \ge i+2$.
\end{corollary}
\begin{proof}
Assume that $p$ becomes semi-clean at level $i$ due to some $i$-level node $u$ that became dirty, where $p \in F(u) \setminus S(u)$,
and consequently all clean points in $F(u) \setminus S(u)$ (including $p$) were marked as semi-clean. 

Suppose for contradiction that $p$ serves as a surrogate of some clean $j$-level node $x$, with $j \ge i+2$.
By Observation \ref{ob:semi}, $x$ is an ancestor of some $i$-level node $\tilde x$ such that $p \in S(\tilde x)$.
Since $x$ is clean,  all its descendants 
must be clean   by Observation \ref{ancestor}.
In particular, $\tilde x$, $\pi(\tilde x)$ and $\pi(\pi(\tilde x))$ are all clean.
The fact that $p$ is a 10-friend of both $\tilde x$ and $u$ implies that $\tilde x$ and $u$ are 20-friends,
and so $\delta(\tilde x,u) \le 20 \cdot 5^i$.
If $u$ is a non-leech, then $\tilde x$ will become its leech and thus dirty, a contradiction.
We henceforth assume that $u$ is a leech and $z$ is its host. Thus $u$ and $z$ are 24-friends, and so $\delta(u,z) \le 24 \cdot 5^i$.
We have two cases.
\\\emph{Case 1: $term(z)$ is not over at level $i$.} By construction, $\pi(z)$ will be a dirty non-leech.
Observe that $$\delta(\pi(\tilde x),\pi(z)) ~\le~ \delta(\pi(\tilde x),\tilde x) + \delta(\tilde x,u) + \delta(u,z) + \delta(z,\pi(z))
~\le~ 3 \cdot 5^{i+1} + 20 \cdot 5^i + 24 \cdot 5^i + 3 \cdot 5^{i+1} ~\le~ 15 \cdot 5^{i+1}.$$
Hence $\pi(\tilde x)$ is a 15-friend of $\pi(z)$, and will become its leech and thus dirty, a contradiction.
\\\emph{Case 2: $term(z)$ is over at level $i$.} By Lemma \ref{key}, we have $|F(z)| \ge 2k+2$.
Also, all the points in $F(z)$ are clean (possibly semi-clean) at the beginning of level $i$. For any such point $s \in F(z)$, we have
$$\delta(\pi(\pi(\tilde x)),s) ~\le~ \delta(\pi(\pi(\tilde x)),\tilde x) + \delta(\tilde x,z) + \delta(z,s) ~\le~ 4 \cdot 5^{i+2} + 44 \cdot 5^i + 10 \cdot 5^i ~<~ 7 \cdot 5^{i+2}.$$
In other words, all points in $F(z)$ are 7-friends of $\pi(\pi(\tilde x))$.
\\If all these points are clean at the beginning of level $i+2$, then they will belong to $F(\pi(\pi(\tilde x)))$ by the construction.
In this case either $\pi(\pi(\tilde x))$ will become a leech of some dirty $(i+2)$-level node or
it will become an appointing node. In either case $\pi(\pi(\tilde x))$ will become dirty, a contradiction.
\\Otherwise, this means that some other appointing node $v$ made at least one of these points $s$ dirty either at level $i$ or at level $i+1$.
%before level $i+2$.
Since $v$ is an appointing node, both $v$ and $\pi(v)$ are dirty non-leech nodes by the construction.
Assume first that $v$ is an $i$-level node. Since $s$ is a 10-friend of both $v$ and $z$, it follows that $v$ is a 20-friend of $z$, 
and would become a leech of $z$ and not an appointing node, a contradiction.
We henceforth assume that $v$ is an $(i+1)$-level node.  % or   an $(i+2)$-level node. 
%Suppose first that $w$ is an $(i+1)$-level node.
%Define $w'$ as $\pi(w)$ in the former case, and as $w$ in the latter.
Since $s \in S(v)$, it holds that %$\delta(\pi(w),s) \le 10 \cdot 5^{i+2}$. 
%(This is obvious if $w' = w$ is an $(i+2)$-level node;
%in the case that $w' = \pi(w)$ is an $(i+1)$-level node, we have 
$$\delta(\pi(v),s) ~\le~ \delta(\pi(v),v) + \delta(v,s) \le 3 \cdot 5^{i+2} + 10 \cdot 5^{i+1} = 5 \cdot 5^{i+2}.$$
Hence $\pi(v)$ is a dirty non-leech $(i+2)$-level node, which is a 5-friend of $s$. Since $s$ is a 7-friend of $\pi(\pi(\tilde x))$,
 $\pi(\pi(\tilde x))$ must be a 12-friend of $\pi(v)$, and would  become its leech and thus dirty, a contradiction.
\QED
%We argue that $w'$ and $\pi(\pi(x))$ are 24-friends, which implies that $\pi(\pi(x))$ would become a leech of $w'$
\end{proof}
%Here too there are two cases. Whether there is equality to 3k+3, or less than that. And even more cases, whether these 2k+2 points are
%disjoint to the points of $F(a(x))$ or not... continue tomorrow...
}

%\vspace{0.07in}
%\noindent
%{\bf 3.2.3 ~Some Corollaries of Lemma \ref{key}.~}
%In this section we provide some  corollaries of Lemma \ref{key}, which play a central role not only in the degree analysis,
%but also in subsequent sections of the paper. [[S: rewrite]]

%\subsection{Some   Corollaries of Lemma \ref{key}}

\vspace{0.07in}
\noindent
{\bf 3.2.3 ~Completing the Degree Analysis.~}
In this section we prove that the degree of our FT spanner $\cH$ is at most $\eps^{-O(d)} \cdot k$. 
We do this in two stages.  
In the first stage  we bound the degree of $\cH$ due to cross edges, and
in the second stage we show that the degree contribution due to tree edges is negligible.

\vspace{0.07in}
\noindent
{\bf First Stage.~}
Denote the degree of an arbitrary point $p$ due to cross edges until level $i$ by $\deg_i(p)$.
Recall that $\xi = \eps^{-O(d)}$ is an upper bound for the maximum degree of any tree node in the basic spanner $H$. 
Since the surrogates of any node are 34-friends of that node (see Observation \ref{farsur})
and the set of all $i$-level nodes constitutes a $5^i$-packing, for each $i \in [0,\ell]$,
Fact \ref{prop:small_net} implies that any point can serve as a surrogate of at most $O(1)^{O(d)}$ nodes at each level.  
By construction, the number of surrogates in any node is at most $2k+1$, and so a bipartite clique that replaces a cross edge may increase
the degree of each of the relevant surrogates by at most $2k+1$ units.
It follows that  the degree of any surrogate (due to cross edges) may increase by at most $\xi \cdot O(1)^{O(d)} \cdot (2k+1) \le \xi^2 (2k+1)$ units at each level.

Recall that $D = (\tau+4) \xi^2 (2k+1)$, and that $\ell$ is the last level in the net-tree $T$ (the level of the root).

In what follows we show that $\deg_\ell(p) \le 2D$, for any point $p \in X$.
\begin{lemma} \label{threshold}
For any surrogate $p \in S(x)$ of a $j$-level clean  node $x$, we have $\deg_{j}(p) \le D$.
\end{lemma}
\begin{proof}
Let $p$ be the point in $S(x)$ of maximum degree, i.e., for every point $q \in S(x)$, $\deg_{j}(q) \le \deg_{j}(p)$.
By Observation \ref{bas}, $p \in D(x)$, implying that $x$ is the $j$-level base bag of $p$.
\\If $\deg_{j}(p) < 2k+2$, then we are done. 
We henceforth assume that $\deg_{j}(p) \ge 2k+2$.
Let $i$ be the minimum level such that $\deg_i(p) \ge 2k+2$,
and let $q_1,\ldots,q_{2k+2}$ be $2k+2$ arbitrary neighbors of $p$ due to cross edges until level $i$;
observe that $2k+2 \le \deg_i(p) \le 2k +1+ \xi^2(2k+1)$.
By Claim \ref{friendship}, $p$ is within distance $(\gamma+68) 5^i$ from each of the $2k+2$ points $q_1,\ldots,q_{2k+2}$.
%The fact that $p \in 
Moreover, Claim \ref{friendship} implies that all $2k+2$ points $q_1,\ldots,q_{2k+2}$ will be 10-friends of the $(i+\tau)$-level base bag $y$ of $p$.
Write $i' = i + \tau$.
\\If all these $2k+2$ points are clean at the beginning of level $i'$, then they would all belong to $F(y)$ by the construction,
unless $|F(y)| = 3k+3$. (Indeed, Corollary \ref{preanc} implies that any node $y_l$ for which $q_l \in S(y_l)$ at level until $i'$ must be clean, for each $l \in [2k+1$], and so we have $q_l \in F(y_l)$
by Observation \ref{bas2}. Since there is a cross edge between $p$ and $q_l$ in our FT spanner $\cH$ at level at most $i \le i'$, it follows that $q_l$ will be added to the reserve set of the
corresponding descendant of $y$, and will eventually be added to $F(y)$, unless $|F(y)| = 3k+3$.)
In either case $|F(y)| \ge 2k+2$, and so $y$ will be dirty (in fact, it will be large). 
By Observation \ref{ancestor}, $\pi(y)$ and $\pi(\pi(y))$ will be dirty as well.
\\Otherwise, at least one point $q_l$ becomes dirty before level $i'$. Let $w$ be the (dirty) appointing node which made $q_l$ dirty, where $q_l \in S(w)$,
and let $w'$ be the $i'$-level ancestor of $w$. By Observation \ref{ancestor}, $w'$ must be dirty.
Note that $q_l$ is a 10-friend of $w$, and so it is also a 10-friend of $w'$ by Observation \ref{anctrivial}.
Since $q_l$ is a 10-friend of both $y$ and a $w'$, it follows that $y$ and $w'$ are 20-friends.
If $w'$ is a non-leech, then $y$ will become its leech and thus dirty. By Observation \ref{ancestor}, 
$\pi(y)$ and $\pi(\pi(y))$ will be dirty as well.
Otherwise $w'$ is a leech of some host $z$, and so $w'$ and $z$ are 24-friends.
There are two cases. 
%The first case is that $term(z)$ is not over at level $i+\tau$,
%and the second case is that $term(z)$ is over at level $i+\tau$. In both these cases it can be proved that $\pi(\pi(y))$
%is dirty. (This proof follows similar lines as those in the proof of Corollary \ref{semi}, replacing $x'$ and $u$ from the proof
%of Corollary \ref{semi} with $y$ and $w'$ here, respectively.)
 %I'll provide it now and then remove.
\\\emph{Case 1: $term(z)$ is not over at level $i'$.} By construction, $\pi(z)$ will be a dirty non-leech. Observe that
$$\delta(\pi(y),\pi(z)) ~\le~ \delta(\pi(y),y) + \delta(y,w') + \delta(w',z) + \delta(z,\pi(z))
~\le~ 3 \cdot 5^{i'+1} + 20 \cdot 5^{i'} + 24 \cdot 5^{i'} + 3 \cdot 5^{i'+1} ~<~ 15 \cdot 5^{i'+1}.$$
Hence $\pi(y)$ is a 15-friend of $\pi(z)$. In other words, $\pi(y)$ is close enough to become a leech of $\pi(z)$,
and so it must be dirty.
By Observation \ref{ancestor}, $\pi(\pi(y))$ will be dirty too.
%\\\emph{Case 2: $term(z)$ is over at level $i$.} By Lemma \ref{key}, we have $|F(z)| \ge 2k+2$.
\\\emph{Case 2: $term(z)$ is over at level $i'$.} By Lemma \ref{key}, we have $|F(z)| \ge 2k+2$.
Also, all the points in $F(z) = F(z) \setminus S(z)$ are clean at the beginning of level $i'$, and will remain clean at level $i'$ by Claim \ref{rmclean}. For any such point $s$, we have
$$\delta(\pi(\pi(y)),s) ~\le~ \delta(\pi(\pi(y)),y) + \delta(y,z) + \delta(z,s) ~\le~ 4 \cdot 5^{i'+2} + 44 \cdot 5^{i'} + 10 \cdot 5^{i'} ~<~ 7 \cdot 5^{i'+2}.$$
In other words, all points in $F(z)$ are 7-friends of $\pi(\pi(y))$.
\\If all these points are clean at the beginning of level $i'+2$, then they will belong to $F(\pi(\pi(y)))$ by the construction, unless $|F(\pi(\pi(y)))| = 3k+3$.
(Indeed, there is a cross edge between $y$ and $z$ (assuming $\gamma \ge 44$), and so all points of $F(z)$ will be added to the reserve set of $y$,
and will eventually be added to $F(\pi(\pi(y)))$, unless  $|F(\pi(\pi(y)))| = 3k+3$.)
Hence $\pi(\pi(y))$ will either become a leech of some dirty $(i'+2)$-level node or
it will become an appointing node. In either case $\pi(\pi(y))$ will be dirty.
\\Otherwise, this means that some $(i'+1)$-level appointing node $v$ made at least one of these points $s \in F(z)$ dirty at level $i'+1$.
%before level $i+2$.
Since $v$ is an appointing node, both $v$ and $\pi(v)$ must be dirty non-leeches by the construction.
%By Claim \ref{rmclean}, it cannot hold that $v$ is an $i'$-level node. 
%(Indeed, since $s$ is a 10-friend of both $v$ and $z$, it follows that $v$ is a 20-friend of $z$, 
%and would become a leech of $z$ and not an appointing node, a contradiction.)
%Thus $v$ must be an $(i'+1)$-level node.  % or   an $(i+2)$-level node. 
%Suppose first that $w$ is an $(i+1)$-level node.
%Define $w'$ as $\pi(w)$ in the former case, and as $w$ in the latter.
The fact that $v$ appoints $s$ as a new surrogate implies that $v$ and $s$ are 10-friends.
%Since $s \in S(v)$, it holds that %$\delta(\pi(w),s) \le 10 \cdot 5^{i+2}$. 
Hence
%(This is obvious if $w' = w$ is an $(i+2)$-level node;
%in the case that $w' = \pi(w)$ is an $(i+1)$-level node, we have 
$$\delta(\pi(v),s) ~\le~ \delta(\pi(v),v) + \delta(v,s) \le 3 \cdot 5^{i'+2} + 10 \cdot 5^{i'+1} = 5 \cdot 5^{i'+2}.$$
It follows that $\pi(v)$ is a dirty non-leech $(i'+2)$-level node, which is a 5-friend of $s$. Since $s$ is a 7-friend of $\pi(\pi(y))$,
we conclude that $\pi(\pi(y))$ is a 12-friend of $\pi(v)$. In other words, $\pi(\pi(y))$ is close enough to become a leech of $\pi(v)$, and so it must be dirty.
%  and would  become its leech and thus dirty.

Summarizing, we have shown that $\pi(\pi(y))$ is dirty.
Note that $\pi(\pi(y))$ is the   $(i'+2)$-level base bag of $p$ and recall that $x$ is the $j$-level base bag of $p$.
Since $x$ is clean, Observation \ref{ancestor} implies that $j \le i'+ 1 = i+\tau+1$.
%Since $p \in D(\pi(\pi(y)))$,
%Lemma [[S: change]] implies that it is either semi-clean or dirty at level $i+\tau+2$.
%Corollary [[S: change]] implies that $j \le i+\tau+3$.
Recall that $\deg_i(p) \le 2k +1+ \xi^2(2k+1)$.
%Suppose for contradiction that $j \ge i+\tau + 2$, and let $y'$ be the $j$-level ancestor of $\pi(\pi(y))$.
%Since $\pi(\pi(y))$  is dirty, $y'$ must be dirty too by   \ref{ancestor}.
%Notice that $p$ is a 10-friend of $x$ and a 4-friend of $y'$, which is a contradiction to the above.
%Hence $j \le i+\tau +1$. 
Also, in each of the $j - i  \le \tau+1$ levels $i+1,\ldots,j$, the degree of $p$ may increase by at most
$\xi^2 (2k+1)$ units, and so $$\deg_{j}(p) ~\le~ \deg_i(p) + (\tau+1) \xi^2 (2k+1) ~\le~ 2k +1+ \xi^2(2k+1) + (\tau+1) \xi^2 (2k+1) ~\le~ D. \inQED$$
\end{proof}

%If $p$ serves as a surrogate of only incomplete nodes, then by Lemma \ref{threshold} its degree will be at most $D$.
%Suppose otherwise, and 
Let $i$ be the first level in which $p$ becomes dirty. (If no such $i$ exists, we can define $i = \ell+1$.)

By construction, $p$ is appointed as a new surrogate of a large (and thus dirty) $i$-level appointing
node $x$. Corollary \ref{preanc} implies that $p$ cannot be a surrogate of dirty nodes before level $i$, and so $\deg_{i-1}(p) \le D$ must hold by Lemma \ref{threshold}. 
In the degenerate case when $p$ is not appointed as a surrogate of dirty nodes (i.e., $i = \ell+1$), we have $\deg_\ell(p) = \deg_{i-1}(p) \le D$, and we are done. We henceforth assume that $i \le \ell$.

Recall that $p$ is appointed at level $i$ for a term of at least $\tau + 3$ levels, which consists of two phases.
The first phase starts at level $i$, and ends when the degree of $p$ due to (non-redundant) cross edges since its appointment at level $i$ reaches $k+1$.
When the second phase starts, the degree of $p$ is between $k+1$ and  $k + \xi^2(2k+1)$.
The second phase lasts $\tau + 2$ levels. In each of these levels $p$ may incur $\xi^2 (2k+1)$ additional units
to its degree. Consequently, when the term of $p$ is over, the degree of $p$ due to cross edges since its appointment is bounded above by
$k + \xi^2(2k+1) + (\tau+2) \xi^2(2k+1) \le  D$. (In fact, it is possible that $p$ does not finish its term due to abundance of surrogates;
see the second remark at the end of Section 3.1.4. Thus we can refer to $term(x)$ instead of the term of $p$.)
%At this stage $p$ is still clean, and so it can be re-appointed as a new surrogate for another \emph{second} term.
%During this term the degree of $p$ will increase by another at most $4D$ units. 
When $p$ was appointed as a new surrogate of the large $i$-level node $x$ (at level $i$), it was marked as dirty.
Therefore, by construction, $p$ will not be re-appointed again as a surrogate (of either clean nodes or dirty ones) after $term(x)$ is over.

Summarizing, 
the degree of $p$ may increase by at most $D$ units 
while it is
%$$p$ may incur at most $D$ units of degree 
a surrogate of clean nodes, and by at most $D$ additional units while it is a surrogate of dirty nodes.
%Also, Lemma \ref{threshold}
%shows that the degree of $p$ may increase by at most $D$ additional units while it is a surrogate of clean nodes.
%at most $D$ units of degree are contributed to $p$ from clean nodes, and additional $D$ units of degree are contributed to it
%from large nodes.
It follows that $\deg_\ell(p) \le 2D$.
%   and we are done.  %thus the degree of any point due to cross edges is $O(D)$.
%(See Section \ref{app:tree} for the issue of tree edges.)

\vspace{0.07in}
\noindent
{\bf Second Stage.~}
%[[S: We should separate between clean and dirty nodes. All tree edges in the subtree of some clean nodes
%can be bounded together -- since there are at most $O(k)$ surrogates in that node, the degree due to cross
%edges is bounded by $O(k)$. The problem is that the same point $p$ may be in many such nodes. So we probably 
%need to look at the highest level clean node of which $p$ is a surrogate. This point must be clean at that level.]]
Consider a tree edge $(x,\pi(x))$ between a node $x = (p,j-1)$ and its parent $\pi(x) = (q,j)$ in the tree.
% between a node $(p,i)$ and one of its children $(q,i-1)$ in the tree.
This edge is translated into a bipartite clique $C_{x}$ between $S(x)$ and $S(\pi(x))$ in our FT spanner $\cH$.
If $S(x) = S(\pi(x))$, then this edge is \emph{redundant} and we do not replace it by a bipartite clique.
%We say that the level of the clique is the level of $\pi(x)$.
%We say that the clique $C_x$ is \emph{clean} if $\pi(x)$ is clean (by Observation \ref{ancestor}, $x$ is clean too);
%otherwise it is dirty.

Next, we argue that the degree of a point $p$ in $\cH$ due to tree edges is small.
To this end we analyze the degree of $p$ over all cliques $C_x$ that correspond to tree edges $(x,\pi(x))$.
We say that $C_x$ is a \emph{$j$-level clique} if $\pi(x)$ is a $j$-level node (that is, the level of the clique $C_x$ is determined by the level of $\pi(x)$).

Let $i$  be the first level in which $p$ becomes dirty. 
%[[S: Problematic to say it like this, because in level $i-1$ the base bag of $p$ may be dirty]]
%If $p$ does not become semi-dirty but it becomes fully-dirty, or vice versa, then we define $i = i'$.
%If $p$ becomes both semi-dirty and fully-dirty, then  .
If $p$ does not become dirty, then we define $i = \ell+1$.
%(It is possible that $p$ does not become semi-dirty and fully-dirty, but we disregard this 
%is a surrogate of a dirty node.
%By construction, $p$ is appointed as a surrogate of a large appointing node at level $i$. 
%Moreover, until level $i-1$ becomes semi-clean or dirty. By construction, until level $i-1$
%$p$ may only be a surrogate of clean nodes. 
%By Observation \ref{ancestor}??, 
By Corollary \ref{preanc}, $p$ may only be a surrogate of clean nodes in  levels $0,\ldots,i-1$,
%Since $p$ becomes full-dirty at level $i'$, there is an appointing $i'$-level node of which $p$ is a surrogate.
%In levels $i$ and $i+1$, $p$ may be a surrogate of either clean or dirty nodes (but not both).
%By Corollary \ref{semi}, 
whereas it may only be a surrogate of dirty nodes in  levels $i,\ldots,\ell$.
Note that the $j$-level base bag of $p$, for $j < i$, is not necessarily clean. Let $l$ be the first level 
in which the $l$-level base bag of $p$ is dirty, and assuming $l \ge 1$, let $z$ denote the $(l-1)$-level base bag of $p$.
%(It is possible that $l = 0$, i.e., 
%(If the 1-level base bag of $p$ is dirty, we have $l = 0$.)
Notice that all ancestors of $z$ are base bags of $p$, and they are dirty by Observation \ref{ancestor}. 
Since $z$ is clean, Observation \ref{ancestor} implies that all the descendants of $z$ are clean too.
Claim \ref{friends} implies that $l \le i$. Moreover, $p$ cannot be a surrogate
of any $j$-level node, for all $j \in [l,i-1]$.
%In the borderline levels between $i$ and $i'-1$, $p$ may be a surrogate of either clean or dirty nodes.
%However, Corollary [[S: change]] implies that $p$ may only be a surrogate of dirty nodes in levels $i+2,\ldots,\ell$.
%yields $i \le i' \le i+2$
%In what follows we analyze each of these cases.
% and a surrogate of dirty nodes
%in levels $i,\ldots,\ell$.

We first analyze the degree of $p$ over all cliques until level $l-1$, i.e., over $j$-level cliques $C_x$ that correspond to tree edges $(x,\pi(x))$, 
where $\pi(x)$ is a $j$-level node and $j \le l-1$, and $p$ belongs to either $S(x)$ or $S(\pi(x))$.
Observation \ref{bas} implies that $\pi(x)$ is the $j$-level base bag of $p$.
Consequently, both $x$ and $\pi(x)$ must be clean. 
%(Indeed,  , and so both $\pi(x
%if $p$ belongs to $S(x)$, then $x$ and $\pi(x)$ are the $(j-1)$-level and $j$-level base bags of $p$,
%respectively, and we have shown that they are clean. If $p$ belongs to $S(\pi(x))$, then $\pi(x)$ is the $j$-level base bag of $p$, and all its descendants
%and it is clea
%clean (by Observation \ref{ancestor}, $x$ is clean too)
%[[S: why do we assume that $\pi(x)$ is clean -- need to explain that if not, then it will not be a surrogate of $p$ before level $i$]]
%(Note that $p$ must be clean.)
%Notice that $x$ and $\pi(x)$ must be base bags of $p$. More generally, 
Notice that $\pi(x)$ must 
be a descendant of the $(l-1)$-level base bag $z$ of $p$.
Observation \ref{bas} implies that  $S(x) \subseteq  S(\pi(x)) \subseteq S(z)$.
%$S(x) \subseteq S(\pi(x))$, 
%and so $p \in S(z)$. %Moreover, $\pi(x)$ must be a base bag of $p$.
  % by Observation \ref{bas}.
%Denote by $y_1,\ldots,y_t$ the $(i-1)$-level nodes of which $p$ is a surrogate,
%where $t = O(1)^{O(d)}$. 
%By Claim [[S: change]], any $j$-level node $x$ of which $p$ is a surrogate, for $j \le i-2$, must be a descendant of some $y_i$.
%Moreover, the surrogate set $S(x)$ of any descendant $x$ of $y_i$ is contained
%in the surrogate set $S(y_i)$ of $y_i$,
Consequently, the clique $C_x$ that corresponds to the edge $(x,\pi(x))$ is contained in the \emph{internal clique} over $S(z)$,  which includes
an edge between any pair of points from $S(z)$.
It follows that all cliques $C_x$ that may increase $p$'s degree until level $l-1$  are contained in the internal clique  over $S(z)$.
Corollary \ref{bas2} implies that $|S(z)| < 2k+2$.  % for each $i \in [t]$. 
%Hence $p$ will be connected to at most $2k$ points in the internal clique over $S(y)$, 
Hence the degree of $p$ over all cliques $C_x$ until level $l-1$ is bounded above by $2k$.

For a fixed $j$, the degree of $p$ due to all $j$-level cliques $C_x$ may increase by at most $O(1)^{O(d)} \cdot (2k+1)$ units.
(This is because (i) $p$ may serve as a surrogate of at most $O(1)^{O(d)}$ nodes at each level,
(ii) each tree node has at most $O(1)^{O(d)}$ children, and (iii) each node has at most $2k+1$ surrogates.)
In particular, the degree of $p$ due to all $l$-level cliques (where $x$ and $\pi(x)$ are in levels $l-1$ and $l$, respectively) is at most $O(1)^{O(d)} \cdot (2k+1)$.
Similarly, the degree of $p$ due to all $i$-level cliques (where $x$ and $\pi(x)$ are in levels $i-1$ and $i$, respectively) is also in check.
Recall also that $p$ cannot be a surrogate of any $j$-level node, for all $j \in [l,i-1]$.
Consequently, the degree of $p$ over all cliques $C_x$ until level $i$ is at most $2k + O(1)^{O(d)} \cdot (2k+1)$.

%So far we have shown that the degree of $p$ over all cliques until level $i-1$ is bounded above by $2k$.
%Next, we analyze the degree of $p$ over all cliques of level at least $i$.
%Recall that $p$ may only be a surrogate of dirty nodes at such levels, and so it suffices to consider cliques at level at least $i'$.
In what follows we analyze the degree of $p$ over all cliques of level at least $i+1$, i.e., over $j$-level cliques $C_x$ that correspond to tree edges
$(x,\pi(x))$, where $\pi(x)$ is a $j$-level node, for $j \ge i+1$, and $p$ belongs to either $S(x)$ or $S(\pi(x))$. 
% We henceforth consider cliques of level at least $i+1$.

By definition, $p$ is appointed as a surrogate of a large $i$-level node, denoted $y$, and subsequently becomes dirty.
The degree of $p$ may increase due to tree edges during $term(y)$. 
However, our construction guarantees that $p$ will not become a surrogate of any node after $term(y)$ is over.
Let $i'$ denote the last level of $term(y)$, let $y'$ be the $i'$-level ancestor of $y$, and
let $\Pi = \Pi_{y,y'}$ be the path between $y$ and $y'$ in the net-tree $T$.
Observe that cliques of level larger than $i'+1$ cannot increase the degree of $p$,
and so it suffices to bound the degree of $p$ over all cliques of level between $i+1$ and $i'+1$.

In level $i'+1$ (which is the last level of interest to us) the degree of $p$ will increase by at most $O(1)^{O(d)} \cdot (2k+1)$ units.
Consider now a $j$-level clique $C_x$ that corresponds to a tree edge 
$(x,\pi(x))$, where $\pi(x)$ is a $j$-level node,  for $j \in [i+1,i']$,
and $p$ belongs to either $S(x)$ or $S(\pi(x))$. We also assume that the edge $(x,\pi(x))$ is non-redundant, i.e., $S(x) \ne S(\pi(x))$.
We argue that the clique $C_x$ is \emph{incarnated} in some  clique that corresponds to 
either a $(j-1)$-level cross edge or a $j$-level cross edge.
%a cross edge between a pair of either $j$-level nodes or $(j+1)$-level nodes. 
Since we have already shown that the degree of $p$ due to cross edges is in check,
proving this assertion would imply that the degree of $p$
due to tree edges is in check as well, thus completing the degree analysis.
%Our construction may assign $x'$ as the host of $O(1)^{O(d)} \cdot k$

Observe that all nodes along the path $\Pi$ are dirty non-leech copies of $y$. Moreover, by Claim \ref{disj},
for any $j$-level node $v$, where $j \in [i,i']$,  either $S(v) = S(y)$ or $S(v) \cap S(y) = \emptyset$ must hold,  
and in the former case Claim \ref{fofx} implies that $v$ is either the $j$-level copy of $y$ or one of its leeches. In particular, if $p$ belongs to $S(v)$,
then we have $S(v) = S(y)$.
Consequently, if $p$ belongs to both $S(x)$ and $S(\pi(x))$, then we have $S(x) = S(\pi(x)) = S(y)$, and so the edge 
$(x,\pi(x))$ is redundant. We henceforth assume that $p$ belongs to either $S(x)$ or $S(\pi(x))$, but not to both, and so $S(x) \cap S(\pi(x)) = \emptyset$.
This means that $x$ cannot belong to $\Pi$, as otherwise $\pi(x)$ would belong to $\Pi$ too, and the edge $(x,\pi(x))$ would be redundant.

Suppose next that $x$ does not belong to $\Pi$ but $\pi(x)$ belongs to $\Pi$. In this case $\pi(x)$ has a dirty non-leech child $u$ in $\Pi$, 
with $S(u) = S(\pi(x)) = S(y)$. Since $x$ and $u$ are siblings, they must be $\gamma$-friends (assuming $\gamma \ge 30$). Hence they are connected by a cross edge in the basic spanner $H$,
and so $C_x$ is incarnated in the clique $C_{x,u}$ that corresponds to the $(j-1)$-level cross edge $(x,u)$.

We may henceforth assume that neither $x$ nor $\pi(x)$ belong to $\Pi$. 
The first case is that $p$ belongs to $S(x)$ but does not belong to $S(\pi(x))$. In this case $x$ is a leech of the $(j-1)$-level copy $y_{j-1}$ of $y$,
and so $\delta(x,y_{j-1}) \le 24 \cdot 5^{j-1}$. Let $y_{j}$ denote the $j$-level copy of $y$, and note that $y_{j} = \pi(y_{j-1})$. We have 
$$\delta(y_{j},\pi(x)) ~\le~ \delta(y_{j},y_{j-1}) 
+ \delta(y_{j-1},x) + \delta(x,\pi(x)) ~\le~ 3 \cdot 5^{j} + 24 \cdot 5^{j-1} + 3 \cdot 5^{j} ~<~ 11 \cdot 5^{j}.$$ 
Hence
 $\pi(x)$ is an 11-friend of the dirty non-leech $y_{j}$, and so it must be a leech too by the construction.
Note that the host $h$ of $\pi(x)$ cannot be $y_{j}$, otherwise the edge $(x,\pi(x))$ would be redundant.
Observe also that $\delta(y_{j},h) \le \delta(y_{j},\pi(x)) + \delta(\pi(x),h) \le 11 \cdot 5^{j} + 24 \cdot 5^{j} = 35 \cdot 5^{j}$.
In other words, $y_{j}$ and $h$ are $\gamma$-friends (assuming $\gamma \ge 35)$. Thus they are connected by a cross edge in 
the basic spanner $H$, and so $C_x$ is incarnated in the clique $C_{y_{j},h}$ that corresponds to the $j$-level cross edge $(y_{j},h)$.

The remaining case is that $p$ does not belong to $S(x)$ but belongs to $S(\pi(x))$. Thus $\pi(x)$ is a leech of the $j$-level
copy $y_{j}$ of $y$, and so $\delta(y_{j},\pi(x)) \le 24 \cdot 5^{j}$ and $S(\pi(x)) = S(y_{j}) = S(y_{j-1})$. We have
$$\delta(y_{j-1},x) ~\le~ \delta(y_{j-1},y_{j}) + \delta(y_{j},\pi(x)) + \delta(\pi(x),x) ~\le~ 3 \cdot 5^{j} + 24 \cdot 5^{j} + 3 \cdot 5^{j} ~=~ 150 \cdot 5^{j-1}.$$ In other words, $y_{j-1}$ and $x$ are $\gamma$-friends (assuming $\gamma \ge 150$).
Thus they are connected by a cross edge in the basic spanner $H$, and so $C_x$ is incarnated in the clique $C_{y_{j-1},x}$ that corresponds to the $(j-1)$-level cross
edge $(y_{j-1},x)$. This concludes the proof of the above assertion and completes the degree analysis.

\subsection{Fault-Tolerance and Stretch}
In this section we show that    $\mathcal H$ 
is a $k$-FT $(1+\eps)$-spanner for $(X,\delta)$.  % is easy, and deferred to Appendix \ref{app:stretch}.
%The difficult task is to show that $deg(\cH) = O(k)$. In Lemma ?? we show that all 10-friends of an incomplete node
%must be clean. This 
To this end we show that for any pair $p,q \in X \setminus F$ of functioning  points (where $|F| \le k$), there is a $(1+\eps)$-spanner path
in $\mathcal H \setminus F$.   

Consider the $(1+\eps)$-spanner path $\Pi_{p,q}$ between $p$ and $q$ in the
basic spanner $H$ that is guaranteed by Lemma \ref{lemma:cross_edge}. The path $\Pi_{p,q}$ is obtained by climbing up from the leaf nodes $(p,0)$ and $(q,0)$
to some $j$-level ancestors $(p',j)$ and $(q',j)$, respectively, which are connected by a cross edge  $((p',j),(q',j))$.
By Observation \ref{bas}, $p$ is a surrogate of every clean node along the sub-path of $\Pi_{p,q}$ that climbs from $(p,0)$ to $(p',j)$. Also, every dirty node along this path is complete, hence it has $k+1$ surrogates, at least one of which must function. Since every edge along this
path is translated into a bipartite clique in $\cH$ between the corresponding surrogate sets, we obtain this way a functioning
path that starts at $p$ and ends at an arbitrarily chosen surrogate $\tilde p \in S(p',j) \setminus F$ of $(p',j)$. 
Similarly, we obtain a functioning path that starts at $q$ and ends at an arbitrarily chosen surrogate $\tilde q \in S(q',j) \setminus F$ of $(q',j)$. 
The cross edge between $\tilde p$ and $\tilde q$ in $\cH$ is functioning too, i.e., $(\tilde p, \tilde q) \in \cH \setminus F$. 
(In the degenerate case where $S(p',j) = S(q',j)$, the cross edge between $(p',j)$ and $(q',j)$ in $H$ is redundant, and we do not translate it into a bipartite clique in $\cH$. Hence,  the edge $(\tilde p,\tilde q)$ may not belong to $\cH$. However, in this case we can simply take $\tilde p = \tilde q$.)
In this way we obtain a functioning path $s(\Pi_{p,q})$ between $p$ and $q$ in our FT spanner $\cH$.
Observation \ref{farsur} implies that any node along the original path $\Pi_{p,q}$ in $H$ is a 34-friend of all its surrogates (and in particular, of all its functioning surrogates). 
Lemma \ref{lemma:cross_edge2} implies that $s(\Pi_{p,q})$ is a $(1+O(\eps))$-spanner path between $p$ and $q$, but we can reduce the stretch to $1+\eps$ by scaling $\gamma$ up by some constant.

\section{Lightness}
The radius $rad(x)$ of an $i$-level node is its distance scale; we   have $rad(x) = 5^i$, but
in most other papers (such as \cite{GGN04,CGMZ05}) $rad(x)$ is equal to $2^i$.
%other papers usually have $rad(x) = 2^i$.
The standard analysis of \cite{CGMZ05} (see also \cite{CLN12,Sol12,CLNS13}) implies that
the sum of radii of all tree nodes is $O(\log n)\cdot \omega(MST(X))$.
Since the degree of any tree node in the basic spanner construction $H$ is $\eps^{-O(d)}$
and since all edges in $H$ that are incident on a node $x$ have weight at most $O(\frac{1}{\eps}) \cdot rad(x)$,
it follows that the lightness of $H$ is bounded above by $\eps^{-O(d)} \cdot \log n$. 
Our FT spanner construction $\cH$ replaces each edge $(x,y)$ in $H$ by a bipartite clique of size $O(k^2)$ between $S(x)$ and $S(y)$.
%(There are also internal cliques for each node, but these cliques [[S:??]])
By Observation \ref{farsur}, the weight of each of the $O(k^2)$ edges of the clique exceeds the weight $\omega(x,y)$ of the original edge
by an additive factor of at most $34 \cdot (rad(x) + rad(y)) = O(rad(x) + rad(y))$. Hence the lightness of our
FT spanner $\cH$ is bounded by $\eps^{-O(d)} \cdot k^2 \log n$. (See \cite{CLNS13} for a   detailed argument.)

To obtain the desired lightness bound of $\eps^{-O(d)} (k^2+ k \log n)$, we  need to work  harder.
The key idea is to replace each bipartite clique by a bipartite matching. We can do this ``matching replacement'' only for edges $(x,y)$
where both nodes $x$ and $y$ are complete. Note that complete nodes are not necessarily dirty. 
For technical reasons it is more convenient to do this matching replacement for edges $(x,y)$ where both nodes $x$ and $y$ are dirty.
In particular, the surrogate set of a dirty node contains exactly $k+1$ points (see Observation \ref{ancestor}).
In this case $|S(x)| = |S(y)| = k+1$, and we will use a perfect bipartite matching (of size $k+1$) to connect $S(x)$ with $S(y)$.
%via a perfect matching. 
If $x$ or $y$ (or both) are clean, we will connect $S(x)$ and $S(y)$ by a bipartite clique as before.
Using a bipartite matching instead of a bipartite clique shaves a factor of $k$ in the lightness bound;
more specifically, the total lightness of all the bipartite matchings is $\eps^{-O(d)} \cdot k \log n$. 

Next, we analyze the lightness of the bipartite cliques. Each such clique corresponds to an edge $(x,y)$ of $H$, in which
either $x$ or $y$ (or both) are clean. We charge the total weight of all the clique edges with the clean nodes among $x$ and $y$.
%(If both $x$ and $y$ are clean, we can choose one of them arbitrarily and charge it.)
In what follows we bound the total charge $W$ over all clean nodes.

Recall that all descendants of a clean node are clean as well (see Observation \ref{ancestor}).
For a clean node $x$, let $W(x)$ denote the total charge over $x$ and 
%of all clean nodes in the subtree $T_x$ of $T$ rooted at $x$,
all its descendants.
In other words, $W(x)$ stands for the total weight
of clique edges that correspond to edges of $H$ that are incident on either $x$ or its descendants.
%We orient the edges of each clique towards a complete node; if both endpoints
%correspond to complete nodes, the edge is oriented arbitrarily.
%The minimum spanning tree of a  metric space $(X,\delta)$ is denoted by $MST(X)$
%(or simply $MST$ if $(X,\delta)$ is clear from the context).
%Also, we denote by $w(MST)$ the weight of $MST$.
%Given a spanner $H$ for $(X,\delta)$, the \emph{lightness} $\Psi(H)$ of $H$ is defined
%as the ratio of  the weight of $H$ to the weight of $MST$.
%\begin{fact} \label{fact:mst}
%~~1. $w(MST) \ge \Delta$.
%~~2. Let $S \subseteq X$ be an $r$-packing, with $r\leq \Delta$. Then,  $w(MST) \geq \frac{1}{2}r\cdot|S|$.
%\end{fact}
\begin{lemma}  \label{rad}
%For an incomplete node $x$, let $W(x)$ denote the total weight
%of clique  edges that are incident on surrogates of either $x$ or its descendants.
%cross edges **in the FT spanner** that are incident on either $x$ 
%or on its descendants. 
For a clean $i$-level node $x$, we have $W(x) \le \eps^{-O(d)} \cdot k^2 \cdot rad(x)$. 
\end{lemma}
\begin{proof}
Let $m(x)$ denote the number of  edges in $\cH$ that correspond to edges of $H$ that are incident on either $x$ or its descendants.
Observation \ref{bas} implies that for any descendant $x'$ of $x$, $S(x') \subseteq S(x)$.
%Recall that $\deg_i(p)$ stands for  the degree of a point $p$ due to cross edges until level $i$.
It follows that $m(x) \le \sum_{p \in S(x)} \widehat{\deg}_i(p)$, where $\widehat{\deg}_i(p)$ stands for the total degree of a point $p$
due to cross edges until level $i$ and due to tree edges until level $i+1$ (i.e., including the edge $(x,\pi(x))$).
%(due to both tree edges and cross edges) until level $i$.
Recall that $\deg_i(p)$ stands for the degree of a point $p$ due to cross edges until level $i$.
By Lemma \ref{threshold}, $\deg_i(p) \le D$, for each $p \in S(x)$.
We have also shown that the degree of such a point $p$ due to tree edges until level $i+1$ (including the edge $(x,\pi(x))$) is 
bounded above by $2k + O(1)^{O(d)} \cdot (2k+1)$. 
%By construction, 
%each node has at most $2k+1$ surrogates.
%Hence, 
%the degree of $p$ due to the tree edge $(x,\pi(x))$ may increase
%by additional $2k+1$ units.
%at most $\deg_i(p) + O(1)^{O(d)}$.
We thus have $\widehat{\deg}_i(p) \le \deg_i(p) + 2k + O(1)^{O(d)} \cdot (2k+1)  \le \eps^{-O(d)} \cdot k$. 
By construction, $|S(x)| \le 2k+1$. It follows that
%Since $|S(x)| < 2k+2$, we have
%(we disregard tree edges, which may increase $m(x)$ by at most a constant factor). 
$$m(x) ~\le~ \sum_{p \in S(x)} \widehat{\deg}_i(p) ~\le~ \sum_{p \in S(x)} \eps^{-O(d)} \cdot k ~\le~ (2k+1)\cdot (\eps^{-O(d)} \cdot k) ~\le~ \eps^{-O(d)} \cdot k^2.$$
By Claim \ref{friendship},
each of the cross edges in $\cH$ that increase $p$'s degree until level $i$ has weight at most $(\gamma+68)5^i$. (In fact, since $x$ is clean, $S(x) = D(x)$. Hence the distance between $x$ and its surrogates is at most $4 \cdot 5^i$, and so
the weight of these edges is at most 
$(\gamma+4 + 34) 5^i = (\gamma + 38) 5^i$.)
Also, since the distance between any two surrogates of $S(x)$ is at most $8 \cdot 5^i$, each of the tree edges in $\cH$ that increase $p$'s degree until level $i$ has weight at most $8 \cdot 5^i$. By Observation \ref{farsur}, the surrogates of $\pi(x)$ are 34-friends of it,
hence the edges of $\cH$ that correspond to the tree edge $(x,\pi(x))$ in $H$ have weight at most
$4 \cdot 5^i + 3 \cdot 5^{i+1} + 34 \cdot 5^{i+1} = 189 \cdot 5^i \le (\gamma+68)5^i$; this inequality holds for $\gamma \ge 121$.
Recall also that $rad(x) = 5^i$.
We conclude  that the total weight  of all $m(x)$ edges is at most $m(x) (\gamma+68)5^i = \eps^{-O(d)} \cdot k^2 \cdot rad(x)$.
%Consider a cross edge in the FT spanner that is incident on some descendant $x'$ of $x$.
%More specifically, it is incident on some point in $S(x')$, and increases its degree at that level.
%By the previous lemma, $S(x')\subseteq S(x)$, and therefore this edge 
\QED
\end{proof}

A  clean node is called \emph{almost-dirty} if its parent in $T$ is dirty.
Let $A$ be the set of all almost-dirty nodes in $T$.
By Observation \ref{ancestor}, for any pair  $x,y \in A$ of distinct almost-dirty nodes, neither one of them can be an ancestor of the other.
Moreover, any clean node $z \nin A$ that is not almost-dirty is a descendant of a single almost-dirty node,
which implies that the total charge $W$ of all clean nodes is equal to  $\sum_{x \in A} W(x)$.
By Lemma \ref{rad}, we have $W = \sum_{x \in A} W(x) \le \eps^{-O(d)} \cdot k^2  \cdot \sum_{x \in A} rad(x)$.

Next, we show that $\sum_{x \in A} rad(x) = O(\omega(MST(X)))$, which implies that $W \le \eps^{-O(d)} \cdot k^2  \cdot \omega(MST(X))$.
Since the total weight of the bipartite cliques is at most $W$, the desired lightness bound would follow.

%We argue that the net-points $p(x),p(y)$ that correspond to any pair $x,y \in A$ of distinct almost-dirty nodes are distinct.
The following lemma implies that the net-points $p(x),p(y)$ that correspond to a pair $x,y \in A$ of distinct almost-dirty nodes
are far away from each other, and in particular distinct (i.e., $p(x) \ne p(y)$). %achieves a stronger property.
%In particular, they are 
\begin{lemma} [Almost-Dirty Nodes are Far Away From Each Other] \label{wtlem}
For any pair 
%$x,y \in A$ of nodes, it holds that $\delta(x,y) 
$(p,i),(q,j)$ of distinct nodes from $A$, with $i \le j$, we have $\delta(p,q) \ge 5^{j-1} = \frac{rad(q,j)}{5}$.
\end{lemma}
\begin{proof}
%Suppose without loss of generality that $i \le j$. 
Notice that both $p$ and $q$ belong to the $i$-level net $N_i$. Since $N_i$ a $5^i$-packing, we have $\delta(p,q) \ge 5^i$.
If $i \ge j-1$, then we are done.
We henceforth assume that $i \le j-2$. 

Suppose for contradiction that $\delta(p,q) <  5^{j-1}$.
Let $(r,j-1)$ be the $(j-1)$-level ancestor of $(p,i)$. Since $j-1 \ge i+1$, Observation \ref{ancestor} implies that $(r,j-1)$ must be dirty. We have $\delta(r,p) \le 4 \cdot 5^{j-1}$,
and so $\delta(r,q) \le \delta(r,p) + \delta(p,q) \le 4 \cdot 5^{j-1} + 5^{j-1} = 5 \cdot 5^{j-1}$. 
Note also that $(q,j-1)$ is a clean child of $(q,j)$.

If $(r,j-1)$ is a non-leech, then $(q,j-1)$ would have to become its leech (since $(q,j-1)$ and $(r,j-1)$ are 5-friends) and thus dirty, a contradiction.
%Observation \ref{ancestor} implies that $(q,j)$ will be dirty too, a contradiction.
We may henceforth assume that $(r,j-1)$ is a leech of some $(j-1)$-level node $(h,j-1)$. By construction, we have $\delta(r,h) \le 24 \cdot 5^{j-1}$. There are two cases.
\\\emph{Case 1: The parent $(\pi(h),j)$ of $(h,j-1)$ is a non-leech.} Notice that $$\delta(\pi(h),q) ~\le~ \delta(\pi(h),h) + \delta(h,r) + \delta(r,q)
~\le~ 3 \cdot 5^j + 24 \cdot 5^{j-1} + 5 \cdot 5^{j-1} ~<~ 9 \cdot 5^j.$$ This means that $(q,j)$ is a 9-friend of $(\pi(h),j)$, and would become its leech
and thus dirty, a contradiction.
\\\emph{Case 2: The parent $(\pi(h),j)$ of $(h,j-1)$ is a leech.} By construction, the term of the surrogates of $(h,j-1)$ must be over at level $j-1$.
%Moreover, we have $F(h,j-1) = F(h,j-1) \setminus S(h,j-1)$.
Lemma \ref{key} implies that $|F(h,j-1)| \ge 2k+2$. Also, all points of $F(h,j-1)$ remain clean at level $j-1$ by Claim \ref{rmclean}.
By construction, each point of $F(h,j-1)$ is a (clean) 10-friend of $(h,j-1)$.
%  there are at least $2k+2$ clean 10-friends in $F(h,j-1)$, i.e., .
Moreover, 
each point $s \in F(h,j-1)$ is an 8-friend of $(q,j)$; to see this, note that
$$\delta(s,q) ~\le~ \delta(s,h) + \delta(h,r) + \delta(r,q) ~\le~ 10 \cdot 5^{j-1} + 24 \cdot 5^{j-1} + 5 \cdot 5^{j-1} ~<~ 8 \cdot 5^j.$$
Observe also that $\delta(h,q) \le \delta(h,r) + \delta(r,q) \le 24 \cdot 5^{j-1} + 5 \cdot 5^{j-1} = 29 \cdot 5^{j-1}$. Hence  $(h,j-1)$ and $(q,j-1)$ are 29-friends, and so there is a cross edge between them in the basic spanner $H$ (assuming $\gamma \ge 29$).
By construction, all the at least $2k+2$ points of $F(h,j-1)$ will be added to the reserve set $R(q,j-1)$ of $(q,j-1)$.
Moreover, all these points will be subsequently added to $F(q,j)$, unless $|F(q,j)| = 3k+3$. In any case
$(q,j)$ will be large, and thus dirty, a contradiction. The lemma follows.
%Suppose next that some of the points in $F(h,j-1)$ get dirty at level $j$, before Procedure $ComputeSets_{(j)}$ handles the node $(q,j)$. 
%Let $y$ be a  $j$-level appointing node which appoints at least one point from $F(h,j-1)$ as a new surrogate. 
%Notice that $y$ and $(q,j)$ are 18-friends, and so
%$(q,j)$ would have to become a leech of $y$, and thus dirty, a contradiction.
%We henceforth assume that all points of $F(h,j-1)$ remain clean until Procedure $ComputeSets_{(j)}$ handles the node $(q,j)$. 
%In this case $(q,j)$ will become an appointing node, and thus dirty, 
\QED
\end{proof}

Let $X_A \subseteq X$ be the set of net-points that correspond to nodes in $A$, i.e., $X_A = \{p(x) ~|~ x \in A\}$.
%Lemma \ref{wtlem} implies the following corollary, which states that the sum of radii of all nodes in $A$ is  $O(\omega(MST(X)))$.
%Consequently, we will show that the total lightness of all edges in the basic (respectively, FT) spanner construction that are incident on nodes in $A$ will be %$O(1)$ (resp., $O(k^2)$).
%\begin{corollary}
%\end{corollary}
%\begin{proof}
Consider the $MST(X_A)$ of $(X_A,\delta)$, and note that $\omega(MST(X_A))) \le 2 \cdot \omega(MST(X))$.
Let $rt$ be an arbitrary node in $A$. 
We root $MST(X_A)$ at $p(rt)$, and orient all the edges towards $p(rt)$.
In this way for each node $x \in A \setminus \{rt\}$, the corresponding net-point $p(x)$ has a single outgoing edge in $MST(X_A)$, denoted $e(x)$.
By Lemma \ref{wtlem}, $\omega(e(x)) \ge \frac{rad(x)}{5}$.
Also, for any pair $x,y \in A$ of distinct almost-dirty nodes, their corresponding net-points $p(x)$ and $p(y)$ are distinct,
which implies that
the corresponding outgoing edges $e(x)$ and $e(y)$  are distinct as well.
Hence $\omega(MST(X_A)) = \sum_{x \in A \setminus \{rt\}} \omega(e(x)) \ge \sum_{x \in A \setminus \{rt\}} \frac{rad(x)}{5}$.
Finally, note that $rad(rt) \le  \omega(MST(X))$. %The result follows from the fact that $\omega(MST(X_A))) = O(\omega(MST(X)))
It follows that $\sum_{x \in A} rad(x) = O(\omega(MST(X)))$.
%\QED
%\end{proof}

Denote by $\tilde \cH$ the spanner obtained from $\cH$ by replacing the bipartite cliques with bipartite matchings, whenever possible.
The above discussion shows that $\tilde \cH$ achieves the desired lightness bound $\eps^{-O(d)}  (k^2 + k \log n)$.
Also,  the degree of $\tilde \cH$ is no greater than the degree of $\cH$, and so it is in check too.
Finally, we can guarantee that  $\tilde \cH$ will be a $k$-FT $(1+\eps)$-spanner using a simple adjustment, namely,
connect the surrogate set $S(x)$ of each dirty leaf $x$ via an \emph{internal clique} (which includes an edge between any pair of points from $S(x)$). 
It is easy to see that this adjustment does not increase the degree and lightness by more than constants. Indeed, any point $p$ may be a
surrogate of at most $O(1)^{O(d)}$ leaves, and so its degree due to the internal cliques will increase by an additive factor of $O(1)^{O(d)} \cdot k$.
In particular, this means that the number of edges due to the internal cliques is bounded above by $O(1)^{O(d)} \cdot k n$. By Observation \ref{farsur}, any such edge
has weight at most $2 (34 \cdot 5^0) = O(1)$, and so the total weight due to the internal cliques is at most $O(1)^{O(d)} \cdot k n \le O(1)^{O(d)} \cdot k \cdot \omega(MST(X))$,
where the last inequality follows from the fact that the minimum inter-point distance of $X$ is equal to 1.
Hence, the lightness will be in check as well.

%The reason we use $2k+1$ surrogates instead of $k+1$ is that each point may be appointed as a new surrogate
%of a complete node at most twice. (Because in the second time it is marked as dirty, and will not be re-appointed as a new surrogate again.)
%This means that even if there are $k$ faulty points, they will ``ruin'' at most $2k$ surrogates along any $(1+\eps)$-spanner path, 
%and so each bipartite matching should always provide at least one functioning edge.

Next, we show that for any pair $p,q \in X \setminus F$ of functioning  points (where $|F| \le k$), there is a $(1+\eps)$-spanner path
in $\tilde \cH \setminus F$.   
Consider the $(1+\eps)$-spanner path $\Pi_{p,q}$ between $p$ and $q$ in the
basic spanner $H$ that is guaranteed by Lemma \ref{lemma:cross_edge}, obtained by climbing up from the leaf nodes $(p,0)$ and $(q,0)$
to some $j$-level ancestors $(p',j)$ and $(q',j)$, respectively, which are connected by a cross edge  $((p',j),(q',j))$.

Consider first the sub-path $\Pi_{p,p'}$ of $\Pi_{p,q}$ that climbs from $(p,0)$ to $(p',j)$. 
We argue that there is a short functioning path from $p$ to some point in $S(p',j) \setminus F$.
Note that $p$ is a surrogate of every clean node along this path.
The case when $(p',j)$ is clean is immediate.
%Suppose first that $(p',j)$ is dirty,  % (because the case when  is clean is simple).
%then all its descendants are incomplete too, and we are fine.
Suppose that $(p',j)$ is dirty, and consider the highest clean ancestor $x^{(0)}$ of $(p,0)$ along $\Pi_{p,p'}$, and its dirty parent $x^{(1)} = \pi(x^{(0)})$.
Observation \ref{bas} implies that $p$ belongs to $S(x^{(0)})$. By Observation \ref{ancestor}, $|S(x^{(1)}| = k+1$.
There is a bipartite clique between $S(x^{(0)})$ and $S(x^{(1)})$ by the construction,
and so $p$ is connected to each of the $k+1$ points of $S(x^{(1)})$.
%There are two exceptions. In case $(p',j)$ is clean, we have $x^{(0)} = (p',j)$.
We have disregarded the case where $x^{(1)} = (p,0)$ is a dirty leaf. Recall that our adjustment added an internal clique
for each dirty leaf. Hence, in this case $p$ is connected to each of the other $k$ points of $S(x^{(1)})$.
%Similarly, let $y^{(0)}$ be the highest clean ancestor of $(q,0)$, and let $y^{(1)} = \pi(y^{(0)}$ be its dirty parent.

By Observation \ref{ancestor}, all ancestors of $x^{(1)}$ are dirty.
Let $\tilde x^{(1)}$ be the highest (dirty) ancestor of $x^{(1)}$ along $\Pi_{p,p'}$ with the same surrogate set as $x^{(1)}$, i.e., 
$S(\tilde x^{(1)}) = S(x^{(1)})$. (It is possible that $\tilde x^{(1)} = x^{(1)}$.)
Let $x^{(2)} = \pi(\tilde x^{(1)})$ be the parent of $\tilde x^{(1)}$, 
and let $\tilde x^{(2)}$ be the highest (dirty) ancestor of $x^{(2)}$ along $\Pi_{p,p'}$ with the same surrogate set as $x^{(2)}$, i.e., 
$S(\tilde x^{(2)}) = S(x^{(2)})$. (It is possible that $\tilde x^{(2)} = x^{(2)}$.)
By definition, $S(\tilde x^{(2)}) = S(x^{(2)}) \ne S(\tilde x^{(1)}) = S(x^{(1)})$.
%For any ancestor of $y$ with the same surrogate set, we proceed smoothly. Let $y'$ be the first ancestor of $y$
%with a different surrogate set. 
Moreover, Claim \ref{disj} implies that $S(x^{(1)}) \cap S(x^{(2)}) = \emptyset$. 
By construction,   $|S(x^{(1)})| = |S(x^{(2)})| = k+1$,
and there is a perfect bipartite matching (of size $k+1$) between $S(\tilde x^{(1)}) = S(x^{(1)})$ and $S(x^{(2)}) = S(\tilde x^{(2)})$.
Similarly, let $x^{(3)} = \pi(\tilde x^{(2)})$ be the parent of $\tilde x^{(2)}$,
and let $\tilde x^{(3)}$ be the highest (dirty) ancestor of $x^{(3)}$ along $\Pi_{p,p'}$ with the same surrogate set as $x^{(3)}$, i.e., 
$S(\tilde x^{(3)}) = S(x^{(3)})$.
By definition, $$S(\tilde x^{(3)}) ~=~ S(x^{(3)}) ~\ne~ S(\tilde x^{(2)}) ~=~ S(x^{(2)}) ~\ne~ S(\tilde x^{(1)}) ~=~ S(x^{(1)}) ~\ne~ S(x^{(3)}).$$
%For any ancestor of $y$ with the same surrogate set, we proceed smoothly. Let $y'$ be the first ancestor of $y$
%with a different surrogate set. 
Moreover, Claim \ref{disj} implies that $$S(x^{(2)}) \cap S(x^{(3)}) ~=~ S(x^{(1)}) \cap S(x^{(3)}) ~=~  S(x^{(1)}) \cap S(x^{(2)}) ~=~ \emptyset.$$
By construction,  $|S(x^{(2)})| = |S(x^{(3)})| = k+1$,
and there is a perfect bipartite matching (of size $k+1$) between $S(\tilde x^{(2)}) = S(x^{(2)})$ and $S(x^{(3)}) = S(\tilde x^{(3)})$.
It follows that there are $k+1$ vertex-disjoint paths between $S(x^{(1)})$ and $S(\tilde x^{(3)})$.

We continue this way until we reach $\tilde x^{(l)} = (p',j)$. By applying this argument inductively, we get   $k+1$ vertex-disjoint
paths between $S(x^{(1)})$ and $S(\tilde x^{(l)}) = S(p',j)$. Since $|F| \le k$ (i.e., there are at most $k$ faulty points), at least
one of these paths must function. Recall also that $p$ is connected to all points of $S(x^{(1)})$, which means that there is a path
in $\tilde \cH \setminus F$ between $p$ and some functioning point $\tilde p$ of $S(p',j)$. Moreover, this functioning path is a $(1+O(\eps))$-spanner path by Lemma \ref{lemma:cross_edge2},
and we can reduce the stretch to $1+\eps$ by scaling $\gamma$ up by some constant.
This proves the above assertion, but we are not done yet. In exactly the same way as above, it can be shown that 
there is a functioning $(1+\eps)$-spanner path between $q$ and some functioning point $\tilde q$ of $S(q',j)$.
However, the cross edge $(\tilde p,\tilde q)$ need not be in our spanner $\tilde \cH$. (This is because we use bipartite matchings rather than bipartite cliques,
whenever possible.)
Consequently, it may be impossible to glue
these two functioning paths together via the edge $(\tilde p,\tilde q)$.
We can easily overcome this hurdle in the degenerate case where either $(p',j)$ or $(q',j)$ (or both) are clean.  % the above argument suffices.
Indeed, suppose without loss of generality that $(p',j)$ is clean. In this case $p$ belongs to $S(p',j)$. Moreover, by construction, there is a bipartite clique between
$S(p',j)$ and $S(q',j)$, and so $p$ is connected via cross edges to all points of $S(q',j)$.
%In other words, the cross edge $(\tilde p,\tilde q)$ must belong to our spanner $\tilde \cH$ in this
If $(q',j)$ is clean too, then there is a direct edge between $p$ and $q$. Otherwise,
we can use the above functioning $(1+\eps)$-spanner path between
$q$ and $\tilde q$, and then get from $\tilde q$ to $p$ via a direct edge.
%and so $p$ is connected to the point $\tilde q$ of $S(q',j)$. 
We thus obtain a short functioning path between $p$ and $q$, as required. 
%The case when $(q',j)$ is clean but $(p',j)$ is dirty can be handled in exactly the same way.
We will next show how to overcome the above hurdle in the more interesting case where both $(p',j)$ and $(q',j)$ are dirty.

%We overcome this hurdle in the following way. 
The above argument 
%does not deal with  the entire path $\Pi(p,q)$ into three sub
%Instead of 
handles the two sub-paths $\Pi_{p,p'}$ (between $p$ and $p'$) and $\Pi_{q,q'}$ (between $q$ and $q'$) of $\Pi_{p,q}$ separately.
More specifically, it first builds a short functioning path for each of these two sub-paths of $\Pi_{p,q}$, and then it tries to glue these sub-paths together into a single functioning path
using a possibly non-existing cross edge $(\tilde p,\tilde q)$ (i.e., $(\tilde p,\tilde q)$ may not belong to $\tilde \cH$).
%Instead of dealing with these two sub-paths separately, 
Instead of breaking the path $\Pi_{p,q}$ into two sub-paths and handling them separately, we consider the entire path $\Pi_{p,q}$
and build for it a single functioning path in roughly the same way as we did for each of these two sub-paths.
%above.
%The only difference is that $\tilde x^{(i)}$ will not neccasarrily be

As before, let $x^{(0)}$ be the highest clean ancestor of $(p,0)$ along $\Pi_{p,p'}$, and let $x^{(1)} = \pi(x^{(0)})$ be its dirty parent.
Similarly, let $y^{(0)}$ be the highest clean ancestor of $(q,0)$ along $\Pi_{q,q'}$, and let $y^{(1)} = \pi(y^{(0)})$ be its dirty parent.
%Let $x^{(1)}$ be defined as above. Similarly
Recall also that we defined $\tilde x^{(1)}$ as the highest (dirty) ancestor of $x^{(1)}$ along the sub-path $\Pi_{p,p'}$ with the same surrogate set as $x^{(1)}$. Now we define $\tilde x^{(1)}$ differently. More specifically, unlike the above argument, we would also like to consider nodes %with the same surrogate set as $x^{(1)}$
that belong to the other sub-path $\Pi_{q,q'}$ of $\Pi_{p,q}$. %This involves only a minor adjustment.
Let $\Pi_{x^{(1)},y^{(1)}} = (v_1 = x^{(1)},v_2 = \pi(x^{(1)}),\ldots,v_t = (p',j),v_{t+1} = (q',j),\ldots,v_h = y^{(1)})$ be the sub-path of $\Pi_{p,q}$ from $x^{(1)}$ to $y^{(1)}$, where $1 \le t < h \le 2j+2$.
Let $i$ be the maximum index in $[h]$, such that $S(x^{(1)}) = S(v_i)$; we define $\tilde x^{(1)}$ as $v_i$.
%the node that is ``furthest away'' from $x^{(1)}$ on the path $\Pi$ with the same surrogate set as $x^{(1)}$,
%i.e., $S(x^{(1)}) = S(\tilde x^{(1)})$. In other words, 
%then $\tilde x^{(1)} = v_i$. 
(Unlike before, now it is indeed possible that $\tilde x^{(1)}$ would belong to $\Pi_{q,q'}$.)

In the above argument, we defined $x^{(2)} = \pi(\tilde x^{(1)})$. This makes sense, as the objective was to get closer and closer to $(p,j')$.
Here the objective is to get closer and closer to $y^{(1)}$.
So we will define $x^{(2)} = \pi(\tilde x^{(1)})$ only in the case that $\tilde x^{(1)}$ belongs to $\Pi_{p,p'} \setminus (p',j)$. 
If $\tilde x^{(1)} = (p',j)$, then we define $x^{(2)} = (q',j)$.
Finally, in the case that $\tilde x^{(1)}$ belongs to $\Pi_{q,q'}$, 
%However, if $\tilde x^{(1)})$ belongs to the other sub-path $\Pi_{q,q'}$,
we define $x^{(2)}$ to be the child of $\tilde x^{(1)}$ along $\Pi_{q,q'}$.

In the same way we define $\tilde x^{(2)}, x^{(3)}$, and so forth.
We continue this way until we reach $\tilde x^{(l)} = y^{(1)}$. In exactly the same way as before we get
$k+1$ vertex-disjoint paths between 
$S(x^{(1)})$ and $S(\tilde x^{(l)}) = S(y^{(1)})$. Since $|F| \le k$, at least
one of these paths must function. Observe also that $p$ and $q$ are connected to all points of $S(x^{(1)})$ and 
$S(y^{(1)})$, respectively. (Recall that our adjustment
added an internal clique for each dirty leaf. Hence, in the case where $x^{(1)} = (p,0)$ (respectively, $y^{(1)} = (q,0)$) is a dirty leaf, 
the point $p$ (resp., $q$) is connected to each of the other $k$ points of $S(x^{(1)})$ (resp., $S(y^{(1)})$).)
Consequently, we obtain a path
in $\tilde \cH \setminus F$ between $p$ and $q$. Moreover, this functioning path is a $(1+O(\eps))$-spanner path by Lemma \ref{lemma:cross_edge2},
and we can reduce the stretch to $1+\eps$ by scaling $\gamma$ up by some constant.

\ignore{
For any ancestor of $y'$ with the same surrogate set, we proceed smoothly. Let $y''$ be the first ancestor of $y'$
with a different surrogate set. We have $|S(y'')| = k+1$ and $S(y') \cap S(y'') = \emptyset$.
By construction, we have a perfect bipartite matching between $S(y')$ and $S(y'')$.
We continue this way until we reach $(p',j)$. Let $y,y',\ldots,y^{(l)}$ be the ancestors of $y$ with distinct surrogate sets.

Some of the edges of this matching may be incident on faulty points, and will thus get ruined.
However, since at most $k$ points are faulty and each point may serve as a new surrogate of at most two complete nodes, 
they may altogether ruin at most $2k$ edges of any matching, so that we can always proceed up the path with at least one functioning edge.
In other words, we will have a functioning
path that starts at $p$ and ends at some surrogate $\tilde p$ of $(p',j)$, as required. 
Similarly, we   have a functioning path that starts at $q$ and ends at some surrogate $\tilde q$ of $(q',j)$. 
The cross edge between $\tilde p$ and $\tilde q$ also belongs to $\cH \setminus F$. 
 In this way we get a functioning path  between $p$ and $q$ in  $\cH$.
By Lemma 2.3, this is a $(1+O(\eps))$-spanner path, but we can reduce the stretch to $1+\eps$ by scaling $\gamma$ up by some constant.
%In this way we get a functioning path of small stretch between $p$ and $q$ in our FT spanner $\tilde \cH$.
%By Lemma 2.3, $s(\Pi_{p,q})$ is a $(1+O(\eps))$-spanner path, but we can reduce the stretch to $1+\eps$ by scaling $\gamma$ up by some constant.
}

%Explain that any point can serve as a new surrogate only twice. 
%For incomplete nodes we don't have any problem, because we are still using a clique, and still $p$ will belong to every incomplete ancestor.

\section{Diameter}
In this section we show how to obtain diameter $O(\log n)$, without increasing any of the other parameters.

The basic idea was used before in \cite{CLN12,Sol12,CLNS13}, and involves shortcutting the \emph{light subtrees} of the net-tree,
i.e., those with distance scales less than $\frac{\Delta}{n}$, where $\Delta = \max_{u, v \in X}\delta(u, v)$ is the diameter of the metric $X$.  %maximum inter-point distance in the metric.
This shortcutting is carried out using the 1-spanners of \cite{SE10} (cf.\ Theorem 3 therein \cite{SE10}).
\begin{theorem} [\cite{SE10}]  \label{SE}
Let $T$ be an arbitrary $n$-node tree.
% and denote by $M_T$ the tree metric induced by $T$.
One can build in 
$O(n \log n)$ time, 
%for any integer $\rho \ge 4$, 
a 1-spanner  for (the tree metric induced by) $T$ with $O(n)$ edges,
degree at most $\deg(T) + O(1)$ and diameter $O(\log n)$.
\end{theorem}

Next, we would like to add shortcut edges to the spanner $\tilde \cH$ of Section 4.
The na\"{\i}ve approach would be to build the shortcut spanner of Theorem \ref{SE} for each of the light subtrees of the net-tree $T$, 
and then replace each shortcut edge by a bipartite matching if possible (i.e., if both endpoints of that edge are dirty),
and by a bipartite clique otherwise. This would give rise to a $k$-FT $(1+\eps)$-spanner with logarithmic diameter.
However, the degree and lightness would grow by too much, and would exceed the desired thresholds of $\eps^{-O(d)} \cdot k$ and $\eps^{-O(d)} (k^2 + k \log  n)$,
respectively.
%not difficult to see that
%the degree and lightness would increase to $O(k^2)$ and $O(k^2 \log n)$, respectively.

In order to control the degree and lightness due to shortcut edges, we first \emph{prune} some nodes from $T$,
and only later proceed to shortcutting the light subtrees of the resulting \emph{pruned tree} $T^*$. 
Then we proceed as before, replacing each shortcut edge by a bipartite matching (if possible) or a bipartite clique.

Our ultimate spanner $\cH^*$ is obtained by adding the  shortcut edges to the spanner $\tilde \cH$ of Section 4.

We turn to describing how the pruned tree $T^*$ is obtained.
\begin{itemize}
\item
First, we remove all clean nodes from the net-tree. 
%Thus the resulting tree has only complete nodes.
Indeed, shortcutting clean nodes is redundant, as the first dirty node on any $(1+\eps)$-spanner path is connected 
to its clean child on that path via a bipartite clique (i.e., there is a bipartite clique between the corresponding surrogate sets). 
Since the surrogate set of any clean node contains its entire descendant set (see Observation \ref{bas}),
one can go from any leaf point to all the surrogates 
of the relevant dirty node via a direct edge. In the degenerate case where the leaf itself is dirty, we have an internal clique
over the surrogate set of that leaf (recall that the spanner $\tilde \cH$ of Section 4 contains an internal clique over the surrogate set of each dirty leaf), and so one can still go from any leaf point to all the surrogates of that dirty leaf via a direct edge.
%We will explain why this 

\item
Second, we go over the resulting tree (which contains only dirty nodes) top-down, looking for nodes $x$ whose surrogate set is the same
as that of its parent $\pi(x)$ and as that of every one of its dirty siblings.
Such nodes are called \emph{redundant}, and they are removed from the tree.
%Observe that the surrogate set of a redundant node is the same as the surrogate set of all its dirty siblings,
%because all siblings are connected via cross edges. 
When a node $x$ and its siblings are removed from the tree, we connect $\pi(x)$ with its grandchildren. 
These grandchildren become the new children of $\pi(x)$, and they will be siblings in $T^*$ (hereafter, \emph{$T^*$-siblings}). 
However, they too may be redundant (if their surrogate set is the same as $\pi(x)$), and thus removed from the tree.
%referred to as \emph{$T^*$-siblings}
This process may be repeated over and over, 
hence the children of $\pi(x)$ in the resulting pruned tree $T^*$ may be of significantly smaller level than that of $\pi(x)$.
That is, the parent $x'$ of an $i$-level node $x$
in the pruned tree $T^*$ may be of level $i'$, where $i \ll i'$.
%connect 
%an $i$-level node $x$ with its $i'$-level ancestor $x'$ in $T$, 
Notice that $S(x') \ne S(x)$.
Moreover, for every $j$-level node $z$ on the path between $x$ and $x'$ in the original tree $T$, where $j \in [i+1,i'-1]$, we have $S(z) = S(x')$. (Indeed, otherwise
the node $z$ would not be redundant;  since $x$ is dirty, $z$ will be dirty too by Observation \ref{ancestor}, and so it would have to belong to $T^*$, which is a contradiction.) In other words, the edge $(x,x')$ in the pruned tree $T^*$ corresponds to a path $(x = v_i,v_{i+1},v_{i+2},\ldots,x' = v_{i'})$ in the original tree
$T$, where all vertices $v_{i+1},\ldots,v_{i'-1}$ are redundant 
and $S(v_{i+1}) = S(v_{i+2}) = \ldots = S(x') \ne S(x)$.
%[[S: need to explain that all ancestors of this $i$-level node have the same surrogate set as the $j$-level ancestor]]
In particular, this means that for any pair $x,y$ of (dirty non-redundant) nodes in $T^*$, $x$ is an ancestor of $y$ in $T^*$ if and only if $x$ is an ancestor of $y$ in the original net-tree $T$.

Note also that the degree of the resulting pruned tree $T^*$ may be larger than the degree of the original tree $T$.
However, by Fact \ref{prop:small_net}, there can be only $O(1)^{O(d)}$ nodes with the same surrogate set at each level.
This means that the degree of $T^*$ will exceed the degree of $T$ by at most a factor of $O(1)^{O(d)}$.
Since the degree of the original tree $T$ is $O(1)^{O(d)}$, it follows that the degree of the pruned tree $T^*$ is at most
$O(1)^{O(d)} \cdot O(1)^{O(d)} = O(1)^{O(d)}$, i.e., the degree does not grow by too much.
%Hence, similarly to the original tree $T$, the degree of the pruned tree $T^*$ will be bounded by $O(1)^{O(d)}$.
\end{itemize}

For a node $x$ in $T$, denote by $\pi(x)$ (respectively, $\pi^*(x)$) its parent in the original tree $T$ (resp., pruned tree $T^*$).
If $x$ does not belong to $T^*$, then $\pi^*(x)$ denotes the first ancestor of $x$ in $T$ which belongs to $T^*$.

The pruned tree $T^*$ has some important advantages over the original net-tree $T$.
%\begin{itemize}

First, all nodes of $T^*$ are dirty, and so any shortcut edge for $T^*$ will translate into a bipartite matching
rather than a bipartite clique. This property, however, is not required for achieving the desired bounds.
That is, even if we are being wasteful and translate the shortcut edges into bipartite cliques, the desired bounds on the degree and lightness
will still be in check; we shall address this issue in the sequel.
%means that the lightness due to shortcut edges will be in check.

We use the fact that there are no clean and redundant nodes in $T^*$ to prove the following lemma, which is critical for achieving
the desired bounds on both the degree and the lightness.
\begin{lemma} \label{anotherkey}
Every point $p \in X$ may serve as a surrogate of at most $\eps^{-O(d)}$ nodes of $T^*$. 
\end{lemma} 
%We will next prove this assertion.  % in Claim \ref{red}.
\begin{proof}
Consider an arbitrary point $p \in X$, 
%and let $x$ be the appointing node which made $p$ dirty. 
and let $i$ be the first level in which $p$ is a surrogate of some (dirty non-redundant) $i$-level node $x$ in $T^*$.
%We argue that $p$ will serve as a surrogate of at most $\eps^{-O(d)}$ nodes in $T^*$. 
If $x$ is the root of $T^*$, then we are done. We henceforth assume that $x$ is not the root of $T^*$.
Since $x$ is non-redundant,
%Since $x$ is non-redundant, 
its surrogate set must be different than either the surrogate set $S(\pi^*(x))$ of its parent $\pi^*(x)$ in $T^*$ or the surrogate set $S(y)$ of some dirty $T^*$-sibling $y$ of $x$. 
Define $x'$ as $x$ if it is a non-leech, or as the host of $x$ otherwise. Note that $x'$ is dirty and $S(x') = S(x)$.
It is possible that $x'$ is redundant, and thus does not belong to the pruned tree $T^*$.

We argue that $term(x')$ must be over until level $i+\tau + 3$. 
Note that the parent $\pi(x')$ of $x'$ in $T$ may be different than the parent $\pi^*(x')$ of $x'$ in $T^*$;
in this case $\pi(x')$ is a descendant of $\pi^*(x')$ in $T$.
However, in either case  $S(\pi^*(x')) = S(\pi(x'))$ holds by the construction. Similarly, we have $S(\pi^*(x)) = S(\pi(x))$.

In the case that $S(x') \ne S(\pi^*(x'))$, we have $S(x') \ne S(\pi(x'))$,
which means that $term(x')$ is over at level $i$. 
We henceforth assume that $S(x') = S(\pi^*(x'))$, and so $S(x) = S(x') = S(\pi(x'))$.
%Denote by 
\\Suppose first that $S(x) \ne S(\pi(x))$, which means that  $S(\pi(x)) \ne S(\pi(x'))$, and so $\pi(x) \ne \pi(x'),x \ne x'$.
In this case $x$ is a leech of $x'$.
By Observation \ref{trait}, $\pi(x)$ and $\pi(x')$ are 24-friends,
%otherwise we would have $S(x^*) \ne S(\pi(x^*))$, a contradiction.
and so they are connected by a cross edge.
Since $x$ is dirty, $\pi(x)$ is also dirty by Observation \ref{ancestor}, and we have $|S(\pi(x))| = k+1$.
%so Observation \ref{preanc} implies that
The fact that $S(\pi(x)) \ne S(\pi(x'))$
implies that the degree of the surrogates of $S(\pi(x'))$ due to non-redundant cross edges since their appointment reaches $k+1$
at level $i+1$ or before. It follows that the second phase of $term(x')$ will start at level $i+2$ or before. The second phase lasts precisely $\tau+2$ levels,
and so $term(x')$ will be over until level $i + \tau + 3$.
\\Next, assume that $S(x) = S(\pi(x))$. By construction, $S(\pi(x)) = S(\pi^*(x))$, and so $S(x) = S(\pi^*(x))$.
In this case it must hold that $S(x) \ne S(y)$, for some dirty $T^*$-sibling $y$ of $x$. % due to cross edges. %By
Since $x$ and $y$ are $T^*$-siblings, $S(\pi(x)) = S(\pi(y))$. (It is possible that $\pi(x) = \pi(y)$.)
Observation \ref{farsur} implies that $\pi(x)$ and $\pi(y)$ are 68-friends, and so $x$ and $y$ are $370$-friends.
Since $x$ is a leech of $x'$, it follows that $x'$ and $y$ are $394$-friends,
and so there is a cross edge between $x'$ and $y$ (assuming $\gamma \ge 394$).
 %construction, $y$ is non-redundant as well,
Since $y$ is dirty, we have $|S(y)| = k+1$. Note also that $S(x') = S(x) \ne S(y)$. 
Hence the degree of the surrogates of $S(x')$ due to non-redundant cross edges since their appointment reaches $k+1$ at level $i$ or before.
Consequently, the second phase of $term(x')$ will start at level $i+1$ or before, and so
$term(x')$ will be over until level $i+ \tau + 2$.
%In either case, it is easy to see that the second phase of $term(x^*)$ will start at level $i+1$ (or before).

We have shown  that $term(x')$  must be over until level $i+\tau+3$. Moreover, by Claim \ref{disj} and the construction, $p$ cannot be a surrogate of any $j$-level node,
for all $j \in [i+\tau+4,\ell]$.
In other words, $p$ may only serve as a surrogate of $j$-level nodes, where $j \in [i,i+\tau+3]$.
Since $p$ may serve as a surrogate of at most $O(1)^{O(d)}$ nodes at each level of the original tree $T$, it will also serve
as a surrogate of at most $O(1)^{O(d)}$ nodes at each level of the pruned tree $T^*$.
We conclude that $p$ may serve as a surrogate of at most $(\tau+4) \cdot O(1)^{O(d)} \le \eps^{-O(d)}$ nodes of the pruned tree $T^*$,
as required.
\QED
\end{proof}
%which proves the above assertion.
%\end{itemize}

Denote by $T^*_1,\ldots,T^*_m$ the light subtrees of the pruned tree $T^*$. Notice that the node sets of all the light subtrees are pairwise disjoint.
Consider now an arbitrary light subtree $T^*_i$ of   $T^*$. Its degree $\deg(T^*_i)$ is at most $\deg(T^*) = O(1)^{O(d)}$.
We use the 1-spanners of Theorem \ref{SE} to shortcut $T^*_i$, and then translate each non-redundant shortcut edge $(x,y)$ into a perfect bipartite matching between 
the corresponding surrogate sets $S(x)$ and $S(y)$. Notice that such a matching between $S(x)$ and $S(y)$ increases the degree of each point in $S(x) \cup S(y)$ by only one unit.
Any redundant shortcut edge $(x,y)$ (with $S(x) = S(y)$) is simply disregarded.
By Theorem \ref{SE}, the shortcut edges increase the degree of each tree node of $T^*_i$ by at most $\deg(T^*_i) + O(1) = O(1)^{O(d)}$ units.
Also, Lemma \ref{anotherkey} shows that every point may serve as a surrogate of at most $\eps^{-O(d)}$ nodes of $T^*$.
It follows that the degree of any point (due to the bipartite matchings that correspond to the non-redundant shortcut edges) is at most 
$\eps^{-O(d)}$, and so the degree of the spanner $\cH^*$ will be in check.
In fact, even if we are being wasteful and translate the non-redundant shortcut edges into bipartite cliques, 
%the desired bounds on the degree and lightness
%will still be in check;
%and use bipartite cliques instead of bipartite matchings, 
the degree of  $\cH^*$  will still be bounded above by
the desired threshold of $\eps^{-O(d)} \cdot k$.

%Recall that the number of edges in the basic spanner $H$ is at most $\eps^{-O(d)} \cdot n$.
%Hence, it is not difficult to see that the 
By Lemma \ref{anotherkey}, the number of nodes in the pruned tree $T^*$ is 
bounded above by $\eps^{-O(d)} \cdot n$.
As mentioned, the node sets of all the light subtrees of $T^*$ are pairwise disjoint.
Hence, by Theorem \ref{SE}, the total number of all shortcut edges is bounded above by $O(\eps^{-O(d)} \cdot n) = \eps^{-O(d)} \cdot n$.
Since the distance scale of each light subtree is at most $\frac{\Delta}{n}$ (and as the distance scales decrease geometrically with the level),  the weight of any shortcut edge is at most $O(\frac{\Delta}{n})$.
Moreover, each of the $k+1$ edges of the bipartite matching that replaces a shortcut edge is also bounded by $O(\frac{\Delta}{n})$.
Note also that $\Delta \le \omega(MST(X))$.
%bounded by 
%Moreover, since the shortcut edges increase the degree of each tree node of $T^*$ by at most $O(1)^{O(d)}$ units,
It follows that the total weight of the bipartite matchings that correspond to all the non-redundant shortcut edges is at most 
$\eps^{-O(d)} \cdot n \cdot O(\frac{\Delta}{n}) (k+1) \le \eps^{-O(d)} \cdot k \cdot \omega(MST(x))$, and so the lightness of the spanner $\cH^*$ will
be in check.
In fact, even if we are being wasteful and translate the non-redundant shortcut edges into bipartite cliques, 
%In fact, even if we are being wasteful and use bipartite cliques instead of bipartite matchings,
 the lightness of   $\cH^*$ will still be bounded above by the desired threshold of $\eps^{-O(d)}  (k^2 + k \log n)$.
 %$\eps^{-O(d)} \cdot k$.
%$O(1)^{O(d)} \cdot k \log n$.
%and so the lightness  will remain 
%This follows from the fact the 
%shortcut edges each tree node of $T^*$ 

Finally, we show that $\cH^*$ is a $k$-FT $(1+\eps)$-spanner with diameter $O(\log n)$.
Consider an arbitrary pair $p,q \in X \setminus F$ of functioning points (where $|F| \le k$), and let $\Pi_{p,q}$
be the $(1+\eps)$-spanner path between $p$ and $q$ in the spanner $H$ that is guaranteed by Lemma \ref{lemma:cross_edge}. 
Recall that $\Pi_{p,q}$ is obtained as the union of the two sub-paths $\Pi(p,p')$ and $\Pi(q,q')$ between $(p,0)$ and $(p',j)$
and between $(q,0)$ and $(q',j)$, respectively, which are glued together
via the $j$-level cross edge between $(p',j)$ and $(q',j)$.
Note that some of the nodes along $\Pi_{p,q}$ may not belong to $T^*$. 
Let $\Pi^*_{p,p'}$ be the path obtained from $\Pi_{p,p'}$
by (i) replacing each node $x$ along $\Pi_{p,p'}$ which does not belong to $T^*$ with $\pi^*(x)$ (i.e., with the first ancestor of $x$ in $T$ that belongs to $T^*$),
and (ii) removing duplicates, i.e., leaving only the last occurrence of each node in that path.
Let $x^*$ be the last node in $\Pi^*_{p,p'}$, and
observe that $\Pi^*_{p,p'}$ is almost a sub-path of $\Pi_{p,p'}$. More specifically, if $(p',j)$ belongs to $T^*$,
then $\Pi^*_{p,p'}$ is a sub-path of $\Pi_{p,p'}$, with $x^* = (p',j)$.
Otherwise the last node $x^*$ along $\Pi^*_{p,p'}$ is $\pi^*(p',j)$, and the path $\Pi^*_{p,p'} \setminus \pi^*(p',j)$ is 
a sub-path of $\Pi_{p,p'}$. 
We define $\Pi^*_{q,q'}$ in the same way, and denote the last node in $\Pi^*_{q,q'}$ by $y^*$.
Let $\Pi^*_{p,q} = \Pi^*_{p,p'} \circ (x^*,y^*) \circ \Pi^*_{q,q'}$ be the union of the two paths $\Pi^*_{p,p'}$ and $\Pi^*_{q,q'}$, which are glued together via the edge $(x^*,y^*)$. Observe that $x^*$ and $y^*$ may be different than $(p',j)$ and $(q',j)$, respectively;
in particular, the edge $(x^*,y^*)$ may not belong to the spanner $H$. 
%which raises the question whether our spanner $\cH^*$ contains
%a bipartite clique or matching between $S(x%&
Nevertheless, we must have  $S(x^*)= S(p',j)$ and $S(y^*) = S(q',j)$ by the construction.
%the fact that $x^*$ and $y^*$ may be different than $(p',j)$ and $(q',j)$, respectively,
Consequently, the possible differences between $x^*$ and $(p',j)$ and between $y^*$ and $(q',j)$
are insignificant, as $x^*$ and $y^*$ can play the role of $(p',j)$ and $(q',j)$, respectively.
That is, the bipartite clique or matching that should replace the possibly non-existing edge $(x^*,y^*)$ 
is incarnated in the bipartite clique or matching that replaces the $j$-level cross edge between $(p',j)$ and $(q',j)$.
%Notice that 

Observe that the path $\Pi^*_{p,q}$ may contain many edges. This is where the shortcut edges come in handy.
That is, we can use the shortcut edges to reduce the number of edges in the path $\Pi^*_{p,q}$ to $O(\log n)$, without increasing the weight of the path.
%can be used to obtain a $(1+\eps)$-spanner path $\Pi^*_{p,q}$  
%that contains only $O(\log n)$ edges.
The proof for this assertion is simple, and follows similar lines as those in the works of \cite{CLN12,Sol12,CLNS13}.
(The main idea behind this proof is that there are only $O(\log n)$ levels with distance scales at least $\frac{\Delta}{n}$.
In other words, by removing the light subtrees $T^*_1,\ldots,T^*_m$ from the pruned tree $T^*$, we obtain
a tree $\bar T$ with only $O(\log n)$ levels. To maneuver over the (possibly many) remaining levels of $\bar T$ we use the shortcut edges.)
We will henceforth assume that $\Pi^*_{p,q}$ contains only $O(\log n)$ edges.
% but is slightly more intricate, as
%our shortcutting is done for the pruned tree $T^*$ rather than the original net-tree $T$. 
%This is a  minor issue, though,
%since any dirty node that does not belong to $T^*$ must have the same surrogate set as its parent in $T^*$. 	
%Thus, even though the path $\Pi^*_{p,q}$ is not necessarily
%a sub-path of $\Pi_{p,q}$ as in the works of \cite{CLN12,Sol12,CLNS13}, it is almost a sub-path. Specifically, 
%if the least common ancestor of the leaf nodes that contain $p$ and $q$ belongs to $\Pi_{p,q}$ but not to $T^*$,
%we need to replace it with its parent in $T^*$; %This replacement does not increase the weight of the path.
%disregarding this single replacement, the rest of the path $\Pi^*_{p,q}$ is a sub-path of $\Pi_{p,q}$.

In Section 4 we used a sub-path $\Pi_{x^{(1)},y^{(1)}}$ of $\Pi_{p,q}$ to show that there are $k+1$ vertex-disjoint paths between
$S(x^{(1)})$ and $S(y^{(1)})$, where $x^{(1)}$ (respectively, $y^{(1)}$) stands for the first dirty node along $\Pi(p,p')$ (resp., $\Pi(q,q')$).
%Let $x^*
By using the path $\Pi^*_{p,q}$ instead of $\Pi_{x^{(1)},y^{(1)}}$,  
we can show in exactly the same way that there are $k+1$ vertex-disjoint paths between the surrogate sets $S(x_1)$ and $S(y_1)$ of the
first and last nodes $x_1$ and $y_1$ along $\Pi^*_{p,q}$, respectively. Moreover, since $\Pi^*_{p,q}$ contains only $O(\log n)$ edges, 
all these paths have at most $O(\log n)$ edges.
The fact that $|F| \le k$ implies that at least one of these paths must function.
Finally, observe that $p$ and $q$ are connected to all points of $S(x_1)$ and $S(y_1)$, respectively.
It follows that there is a $(1+\eps)$-spanner path in $\cH^* \setminus F$ between $p$ and $q$ that consists of only $O(\log n)$ edges.

%Given this path $\Pi^*_{p,q}$, we can build a functioning $(1+\eps)$-spanner path between $p$ and $q$ in exactly the same way as
%we did in Section 4. Since $\Pi^*_{p,q}$ contains only $O(\log n)$ edges, the functioning path that we obtain in this way will contain  at most $O(\log n)$ %edges as well.

\section{Running Time}
Our construction uses the net-tree spanner   of \cite{Rod07,GR082}
and the 1-spanners of \cite{SE10} as black-boxes. 
The construction of \cite{Rod07,GR082} (respectively, \cite{SE10})
can be implemented within time $\eps^{-O(d)}\cdot n \log n$ (resp., $O(n \log n))$.
%and we may disregard this time in the sequel.
%The underlying net-tree is 

A central ingredient in the net-tree spanner of \cite{Rod07,GR082} is the underlying net-tree,
which is based on a sophisticated hierarchical partition from \cite{CG06}.
To build the net-tree and the underlying spanner efficiently, it is important to be able to concentrate on only $\eps^{-O(d)} \cdot n$ \emph{useful} tree nodes.
In particular, a   node is called \emph{lonely} if it has exactly one child in $T$ (which corresponds to the same net-point as the parent);
otherwise it is \emph{non-lonely}.  %Observe that the leaf nodes have no children and are non-lonely.
Following \cite{Rod07,GR082,CLNS13} and other works in this context, a long chain of lonely nodes
will be represented implicitly for efficiency reasons; implementation details can be found in \cite{GR082}. 

Our algorithm from Section 3.1 (see in particular Section 3.1.4) traverses the net-tree bottom-up.
The sets $S(\cdot),D(\cdot),F(\cdot),R(\cdot)$ are computed for the $i$-level (useful) nodes, only after we have computed these sets for 
the nodes at lower levels $j \in [0,i-1]$. The computation of these sets is carried out via Procedure $ComputeSets_{(i)}$,
which makes heavy use of the $i$-level cross edges that we are given 
from the net-tree spanner of \cite{Rod07,GR082}. Notice that each tree node $x$ is incident on only $\eps^{-O(d)}$ cross edges
and only $O(1)^{O(d)}$ tree edges. To compute the sets $S(x),D(x),F(x),R(x)$ for an $i$-level node $x$, we need to gather information from 
$x$'s neighbors $y$ either due to   cross edges or due to tree edges (more specifically, children),
where most of the information about such neighbors $y$ is held in the corresponding sets $S(y),D(y),F(y),R(y)$ of $y$.
With this information at hand, it is rather simple to compute the sets $S(x),D(x),F(x),R(x)$.
In particular, since a potential host for a node is a 24-friend of that node, there is a cross edge between a node and all
its potential hosts; hence finding a potential host can be carried out efficiently.

The descendant sets $D(\cdot)$ of all tree nodes can be computed in a straightforward way within $\eps^{-O(d)} \cdot n$ time.
It is more difficult to compute the other sets $S(\cdot),F(\cdot),R(\cdot)$ efficiently.
In particular, note that Procedure $ComputeSets_{(i)}$ upper bounds the sizes of the surrogate sets $S(\cdot)$ and the friend sets $F(\cdot)$ by $2k+1$ and $3k+3$, respectively.
Upper bounding the size of $S(\cdot)$ and $F(\cdot)$ by $\eps^{-O(d)} \cdot k$ is critical in order to perform efficiently.
Similarly, we also need to bound the size of the reserve sets $R(\cdot)$ by $\eps^{-O(d)} \cdot k$.
%We argue that the size of the reserve sets can be upper bounded by $\eps^{-O(d)} \cdot k$.
Consider an arbitrary node $x$.
We argue that it suffices to keep track of the clean points that were added to the reserve sets of either $x$ or its descendants in the last $\tau$ levels.
%$R(\cdot)$ (or its descendants) in the last $\tau$ levels.
Indeed, a point $p$ that belongs to the reserve set $R(y)$ of some $i$-level node $y$ will be a 10-friend of the $(i+\tau)$-level ancestor $y'$
of $y$ by  Claim \ref{friendship}. %In other words, $p$ is close enough to $y'$ to belong to $F(y')$.
Assuming $p$ belongs to $F(y')$, there is no reason to store it in $R(y')$ as well.
If some node $\tilde y$ on the path in $T$ between $y$ and $y'$ satisfies $|F(\tilde y)| = 3k+3$, 
%Assuming $|F(y')| < 3k+3$, $p$ will belong to $F(y')$ by the construction,
%and there is no reason to store it in $R(y')$ as well. In the case $|F(y')| = 3k+3$, 
then $p$ may not belong to $F(y')$. (This is the only case in which $p$ may not belong to $F(y')$.) However, in this case we should have enough points in $F(y')$. 
To be on the safe side, besides the points that were added to the reserve sets of either $x$ or its descendants in the last $\tau$ levels,
we may store in $R(x)$ additional $O(k)$ points that were added to the reserve sets of $x$'s descendants before the last $\tau$ levels;
as mentioned, such points are 10-friends of $x$. % are 10-friends of $x$ 
In this way we always keep track of enough clean 10-friends of a node.
By construction, in each level at most $\eps^{-O(d)} \cdot k$ new points are added to the reserve set of any node.
Since we only keep track of the last $\tau$ levels (and at most $O(k)$ additional points),
it follows that the size of the reserve sets $R(\cdot)$ can be upper bounded by $\tau \cdot \eps^{-O(d)} \cdot k + O(k) \le \eps^{-O(d)} \cdot k$.
%$\eps^{-O(d)} \cdot k$.
%The reason we can stop at $\chi$ is that we only 
%In order to 
%We can upper bound it by
%$2k+2$, keeping just the points that are $c$-friends of that node; $c$ can be 
Given that the size of all sets $S(\cdot),F(\cdot),R(\cdot)$ is bounded above by $\eps^{-O(d)} \cdot k$,
it is easy to see that
 the overall time needed to compute the sets $S(x),F(x),R(x)$ for all useful nodes $x$ in the tree will be $(\eps^{-O(d)} \cdot k) \cdot (\eps^{-O(d)} \cdot n)
\le \eps^{-O(d)} \cdot k n$.
Having computed the surrogate sets of all nodes, we replace each edge of the net-tree spanner $H$ of \cite{Rod07,GR082} by a bipartite clique between the corresponding
surrogate sets.
The time needed to replace all edges of the spanner $H$ by bipartite cliques is asymptotically the same as the number of edges in the resulting FT spanner $\cH$, which is $\eps^{-O(d)} \cdot k n$.
It follows that the overall time required to build the spanner $\cH$ from Section 3 
is $\eps^{-O(d)}\cdot (n \log n + kn)$.
%can keep all the , say, 100
%We would like to spend $\eps^{-O(d)} (k)$ time for computing these sets for a single useful node.

The transformation described in Section 4, which translates each bipartite clique into a bipartite matching whenever possible (i.e., when
both endpoints of the corresponding edge in the spanner $H$ are dirty),
can be easily implemented within time $\eps^{-O(d)}\cdot kn$. Building the internal cliques for all dirty leaves requires
another amount of $\eps^{-O(d)}\cdot kn$ time. Consequently, the spanner $\tilde \cH$ from Section 4 can be implemented within  time $\eps^{-O(d)}\cdot (n \log n + kn)$.

Finally, the construction of the pruned tree $T^*$ in Section 5 can be carried out within time $\eps^{-O(d)}\cdot n$.
%pruning the redundant nodes in Section 5 can be carried out within time $\eps^{-O(d)}\cdot n$.
By Theorem \ref{SE}, shortcutting the light subtrees of $T^*$ requires $\eps^{-O(d)}\cdot n \log n$ time.
We conclude that the overall running time required to build the ultimate spanner $\cH^*$ from Section 5 (which  proves Theorem 1.1) 
is at most $\eps^{-O(d)}\cdot (n \log n + kn)$.

\ignore{
\noindent \emph{Running Time.} All the subroutines that our construction uses
%have been shown (in the relevant references ) to take 
can be built within $O(n \log n)$ time \cite{Rod07,GR082,SE10}.
Hence, we disregard the running time analysis,
except for places which require clarification.

Given the standard net-tree spanner construction and , our construction
can be implemented within $O(k n)$   time in a rather straightforward manner.
Since the standard net-tree spanner can be built within time $O(n \log n)$ \cite{Rod07,GR082},
and since the same amount of time suffices to build the 1-spanners of \cite{SE10}, the overall running time of our construction is $O(n \log n +kn)$.
We remark that for low-dimensional Euclidean metrics, some variants of the net-tree (such as the fair split tree of \cite{CK92}) can
be built within time $O(n \log n)$ in the algebraic computation tree model. Consequently, the basic spanner construction 
can also be built within time $O(n \log n)$ in this model.
%spanner (obtain by using, say, the fair split tree instead of the net-tree) can be built
The   time bound $O(n \log n)$ for the 1-spanners of \cite{SE10} also applies to this model.
Therefore, for low-dimensional Euclidean metrics, the time bound of our construction applies to this model as well.
For arbitrary doubling metrics, one should also need a rounding operation (to allow one to find the $i$ for which
$2^i < x \le 2^{i+1}$, for any $x$) \cite{Gottlieb13}.
}

\vspace{0.13in}
\noindent
{\bf To summarize:} The spanner $\cH^*$ of Section 5 proves Theorem 1.1.

%\vspace{0.05in}
\section*{Acknowledgments}
%{\bf
The author is grateful to Michael Elkin for helpful and timely discussions.

%\clearpage
%\bibliographystyle{latex8}
%\bibliography{latex8}

%\clearpage
%\pagenumbering{roman}
%\appendix
%\centerline{\LARGE\bf Appendix}
%{\bf 3.3~ .~}

%\section{Fault-Tolerance with Small Stretch} \label{app:stretch} 

%\section{Completing the Degree Analysis} \label{app:deg}

%need some lemma that explains that every tree edge is also a cross edge at the previous level]]
%and also to explain that for incomplete nodes 
%tree edges are contained in cross edges--relevant in particular to Lemma 3.8]]
%mekadem farsh? ]]

%\section{Completing the Lightness Analysis} \label{app:weight}

%\subsection{Using $2k+1$ Surrogates} \label{app:weight2k}

\end {document}